\documentclass[11pt,reqno]{amsart}

\usepackage[utf8]{inputenc}
\usepackage{amsfonts,latexsym,amssymb,amsmath,amsthm,amsrefs,etex,slashed,cancel}
\usepackage{todonotes}
\usepackage{hyperref}
\usepackage{a4wide}
\usepackage{enumerate}
\usepackage{MnSymbol}
\usepackage{nicefrac}
\usepackage{graphicx}

\usepackage{pgfplots}
\usepackage{float}

\newcounter{stressaxiom}
\renewcommand{\thestressaxiom}{SE\arabic{stressaxiom}}

\newenvironment{stressaxioms}
  {\setcounter{stressaxiom}{0}%
   \begin{list}{}{\leftmargin=3.2em\labelwidth=2.6em\labelsep=.6em}}
  {\end{list}}

\newcommand{\stressitem}[2]{%
  \refstepcounter{stressaxiom}%
  \item[\textup{(\thestressaxiom)}]\label{#1}\textbf{#2.}%
}

\DeclareMathOperator{\supp}{supp}
\newcommand{\der}{\mathrm{d}}
\newcommand{\rmi}{\mathrm{i}}

\newcommand{\calO}{\mathcal{O}}

\newcommand{\fp}{\mathrm{fp}}
\newcommand{\sfp}{\mathrm{sfp}}

\newcommand{\R}{\mathbb{R}}
\newcommand{\C}{\mathbb{C}}
\newcommand{\calA}{\mathcal{A}}

\newcommand{\compp}{\mathrm{comp}}

\newcommand{\loc}{\mathrm{loc}}

\newcommand{\calH}{\mathcal{H}}

\newcommand{\WF}{\mathrm{WF}}

\newcommand{\calU}{\mathcal{U}}
\newcommand{\calV}{\mathcal{V}}
\newcommand{\calW}{\mathcal{W}}
\newcommand{\numb} {\mathsf{N}}

\newtheorem{theorem}{Theorem}
\newtheorem{definition}[theorem]{Definition}

\newtheorem{lemma}[theorem]{Lemma}

\newtheorem{proposition}[theorem]{Proposition}

\newtheorem{rem}[theorem]{Remark}
\newtheorem{conv}[theorem]{Convention}

\newcommand{\calN}{\mathcal{N}}

\title[Stress energy tensor in QFT]{An analytic approach to the stress energy tensor in quantum field theory}

\author[A.~Strohmaier]{Alexander Strohmaier}
\address{Leibniz University Hannover, Institute of Analysis, 30167 Hannover, Germany}  \email{a.strohmaier@math.uni-hannover.de} 

\begin{document}

\begin{abstract}
 We discuss a framework for quantum fields in curved spacetimes that possess a stress energy tensor as a connection one form on a suitable moduli space of metrics. In generic spacetimes the existence of such a tensor is thought to be a replacement  for the existence of symmetries that the Minkowski theory relies on. It is shown that the local time-slice property and the implementability of local isometries are consequences of the existence of a stress energy tensor that is a local field. 
We prove that the Klein-Gordon field, in an irreducible Fock representation determined by a quasifree Hadamard state, is an example. In this example we show that the scattering matrix for compactly supported metric perturbations exists in the Fock space and is smooth on a dense set with respect to the perturbation parameter. This generalises results by Dimock and Wald (\cite{MR0553516,MR0391814}). As a tool we also establish the precise microlocal properties of parameter dependent fundamental solutions.
\end{abstract}

\maketitle

\section{Introduction}

Axiomatic quantum field theory on Minkowski space greatly relies on translation invariance and the associated notions such as momentum and energy, defined via the Fourier transform. In spaces with symmetries such as deSitter space, or on stationary spacetimes, the theory can be developed somewhat similarly. However generic spacetimes will not have any isometries. 
Some scientists believe that the existence of a stress energy tensor is not only an appropriate replacement for a theory in curved space, but may also add more structure to a Minkowski theory. This paper is an attempt to give meaning to the idea for a quantum field theory to have a stress energy tensor that is a local field. 
In physics the stress energy tensor serves two purposes. Firstly, it appears on the right hand side of the Einstein equations, which makes it an important object to study the back-reaction of the quantum field on the geometry. Secondly, in the presence of a Killing field it can be used to construct a unitary implementer for the generating isometry. In the context of the former aspect Wald has extensively studied the properties of the stress energy tensor and has given a framework for its metric dependence (see for example \cite{MR1302174}). In this paper we will mainly focus on the latter aspect, linking the stress tensor with the quantum fields and its localisation properties. For the expectation value of the stress energy tensor quantum energy inequalities (see \cite{Fewster2017,MR4190325} for an overview) have been identified. These are important if one takes the right hand side of the Einstein equations as the expectation value of the stress tensor with respect to a state in what is often called the semi-classical approximation to quantum gravity. This also requires a good definition of what it means for a theory to have a stress energy tensor (see for example \cite{MR1738609, MR2004989, ar221001145} for proposals and discussions of this point).

To define a stress energy tensor we allow suitable metric perturbations to be implemented in the Hilbert space of the theory. This is also inspired by the covariant approach to algebraic quantum field theory (\cite{MR2007173}), in which a quantum field theory is a functor from a category of spacetimes to the category of $*$-algebras. More plainly, if we speak of the Klein-Gordon field at the level of algebras, we know not only how to construct it in Minkowski space, but also in all other globally hyperbolic spacetimes, and we obtain relations between isometric subregions (see for example \cite{MR3074853, MR2982635}  for a nice overview of ideas around this). The same is expected for interacting theories. As shown in \cite{MR2007173}
one can understand the stress energy tensor of the quantum field as a metric derivation of the relative Cauchy evolution between perturbed and unperturbed spacetimes. 
Relative Cauchy evolutions are morphisms between $C^*$-algebras and the variation therefore is an abstract derivation. The above point of view was developed further in relation to time-ordered products in \cite{MR2144674}.
The variations of the scattering matrix with respect to perturbations also play an important role in more recent constructions of interacting fields in algebraic quantum field theory (\cite{MR4115010}). Relative Cauchy evolutions have also been named M{\o}ller operators (\cite{MR4615613}). 

Here we take a more analytical and philosophically different point of view that is closer to scattering theory and partially inspired by causal perturbation theory. 
To explain the idea let us start with an informal consideration.
We fix a spacetime $(M,g_0)$ and assume that we are given a quantum field theory, for example as a field $\varphi_{g_0}(\cdot)$ that is an operator-valued distribution on $M$  on a Hilbert space $\mathcal{H} = \mathcal{H}_{g_0}$, that gives rise to field algebras $\calA_{g_0}(\calO)$ associated to each open subset $\calO$ of the spacetime $M$.
We will assume here that this system of $*$-algebras of bounded operators has been completed to a net of von Neumann algebras by passing to the weak closures.
We would then like to implement metric perturbations of the form $g_\epsilon= g_0 +\epsilon h$ in the Hilbert space of the theory.
Here $h$ will be a compactly supported symmetric $2$-tensor so that
$(M,g_\epsilon) = (M,g_0 + h \epsilon)$ is globally hyperbolic for $|\epsilon|$ sufficiently small. In other words there should be a family of unitaries
$V_\epsilon$ acting on $\mathcal{H}$ that change the theory on $(M,g_0)$ to the theory on $(M,g_\epsilon)$. 
The implementation of the change of metric in the Hilbert space may be thought of in different ways. One way, favoured in this paper, is to think of the two spacetimes $(M,g_\epsilon)$ and $(M,g_0)$ as being the same in the distant past and then being subjected to different gravitational interactions for a finite time. Since the theories are the same in the past their Hilbert spaces can be identified. Identifications of this form are natural when considering families of metric and also appeared for example in Hollands and Wald's treatment of Wick polynomials in covariant theories (\cite{MR1864435}).
The restriction to compactly supported smooth metric perturbation here is important. Indeed, even for free fields
general metric perturbations lead to inequivalent representations of the field algebra, the physical meaning being the creation of an infinite number of particles by the perturbation leading to a state that cannot be represented as a vector in Fock space.
It is however expected that compactly supported smooth metric perturbation will only create a finite expected number of particles and can be implemented by unitaries. For the free Klein-Gordon field this has been verified by Dimock and Wald for compact perturbations of Minkowski spacetime (\cite{MR0553516,MR0391814}). We will show here that this is true generally for representations induced by quasifree pure Hadamard states of the scalar field. Following Dimock and Wald we will call the implementer $V_\epsilon$ the scattering matrix.
Since the metric family 
$g_0 + h \epsilon$ is real analytic in the parameter $\epsilon$ one also expects the family $V_\epsilon$ to be strongly differentiable on a dense set of vectors at $\epsilon=0$. 
We will show here that for the case of the Klein-Gordon field in a Fock space representation corresponding to a pure quasifree Hadamard state the canonical implementer is in fact strongly $C^\infty$ on the smooth finite particle subspace as long as the perturbation is sufficiently small.

The derivative $T_{g_\epsilon}(h)$ is a linear functional in $h$ and becomes a possibly unbounded densely defined operator. This operator will be called the stress energy tensor.
We will view this operator as the primary object that allows to perturb the metric within a neighborhood of the background metric $g_0$.

Given the family $V_\epsilon$ we can in principle also define a field on the spacetimes $(M,g_\epsilon)$
by setting
$$
 \phi_{g_\epsilon}(f) = \phi_{g_0}(f)
$$
for test functions $f$ that are supported away from the future of the support of the family $g_\epsilon-g_0$, whereas we set
$$
 \phi_{g_\epsilon}(f) = V_\epsilon \phi_{g_0}(f) V_\epsilon^*
$$
if $f$ has its support away from the past of the support of $g_\epsilon-g_0$. A mild continuity assumption that allows to construct such time-ordered products allows to extend this to all metric perturbations within $\mathcal{M}$.
 This gives rise to local algebras
$\mathcal{A}_g(\calO)$ that depend on the metric $g$. Loosely speaking a formula for the stress energy in terms of local fields and the metric would then allow in the Heisenberg picture to at least formally find the family $V_\epsilon$ from the stress energy tensor $T_{g_0}$, as the stress energy tensor at an infinitesimally different point $T_{g_0+\delta h}$ can be computed from the variation of the metric and the above process. In this article we will try to make this idea precise, by assuming that we are already given the stress energy and local fields for individual metrics in a way compatible with the above idea.

{\bf Summary of main results:}
We will give a general definition of the stress energy as a connection one form on a suitable space of metrics. The corresponding connection defines a notion of parallel transport that is the unitary implementer in Hilbert space when the metric is changed. We formulate postulates that, together with standard assumptions, encode on the level of the parallel transport that
\begin{itemize}
 \item[(1)] the stress tensor is a local field in the sense that it is affiliated to the local algebra of the quantum field,
 \item[(2)] the stress tensor commutes with the other observable fields at spacelike separation,
 \item[(3)] the stress tensor is divergence free,
 \item[(4)] the stress energy tensor is constructed from the background metric and local fields.
\end{itemize}

We will show that the existence of such a connection already implies the local time-slice axiom (Theorem \ref{lts} of Section \ref{ltssec}). Under the assumption of the holonomy being central, we show that local isometries can be implemented by unitaries in the Hilbert space of the theory (Theorem \ref{li} of Section \ref{implsec}).
In that our approach shares some features with the $C^*$-algebraic approach based on a Lagrangian that also results in a stress energy tensor (\cite{MR4548527}) and the author has learned that the type of argument presented here for the time-slice axiom, using diffeomorphism invariance, has also been used in this context.

Our second main result is that the Klein-Gordon field is an example of such a theory in the representation of any quasifree pure Hadamard state (Theorem \ref{KGExamplTh} of Section \ref{exkgex}). This in particular implies that the scattering map considered by Dimock and Wald for Minkowski spacetime satisfies Shale's condition and can be unitarily implemented in Fock space. We also show that this implementation can be differentiated when applied on a dense set of vectors. In fact, we will show that the implementation is smooth on all smooth and simple finite particle vectors in Fock space (Theorem \ref{fptheorem} of Section \ref{implsecdiff}). We conjecture that Hadamard states are completely characterised by being the smooth vectors with respect to the scattering matrix.
The mathematical analysis of the Klein-Gordon field and the implementation of metric changes is done in detail in Section \ref{KGSection}.  One of the main tools is a precise analysis of parameter dependent fundamental solutions as Lagrangian distributions in Section \ref{Appendparam}.
These two sections may be of interest in their own right and can be read independently.

\subsection{Conventions}

Throughout the symbol $A \subset B$ means that $A$ is a subset of $B$ including possibly the case $A=B$. Complex Hilbert space inner products are conjugate linear in the first argument and linear in the second.
All finite dimensional differentiable manifolds will be assumed to satisfy the standard assumption, i.e. are paracompact, second countable and Hausdorff as topological space.
We use the signature convention $(+1,-1,\ldots,-1)$. Throughout the Einstein sum convention is used. The d'Alembert operator is then given by
$\Box = - \delta \der$, in local coordinates $\frac{1}{\sqrt{|g|}} \partial_j(g^{jk}\sqrt{|g|} \partial_k)$.  Spaces of functions such as $C^\infty(M)$ are complex valued unless otherwise stated, e.g. $C^\infty(M)=C^\infty(M,\C)$, whereas the space of real valued smooth functions is denoted by $C^\infty(M,\R)$. For any real number $\ell \in \R$ the spaces $H^\ell _\compp(M)$ denote the compactly supported $L^2$-based-Sobolev spaces on $M$, and similarly $H^\ell _\loc(M)$ are the local Sobolev spaces that are dual to $H^{-\ell }_\compp(M)$. Distributions are identified with functions by means of the Lorentzian volume density. The space of distributions on $M$ associated with the test function space $C^\infty_0(M)$ will be denoted by $\mathcal{D}'(M)$. The space of distributions associated with the test function space $C^\infty(M)$ is the space of compactly supported distributions $\mathcal{E}'(M) \subset \mathcal{D}'(M)$.
We refer to \cite{Ho1} and \cite{MR1852334} for the definitions and properties of these spaces.
When it comes to wavefront sets it is convenient to introduce the notation $\dot T^*M = T^*M \setminus 0$ for the cotangent bundle with its zero section removed.
We will also use the notation $\mathrm{I}^*(M,\Lambda)$ for Lagrangian distributions and $\mathrm{I}^*(M \times N, \Lambda')$ for Fourier integral operators from $N$ to $M$ associated to the canonical relation $\Lambda$ as in \cite{MR2512677} to which we refer for details. The notation $\Lambda' = \{(x,\xi,y,-\eta) \mid (x,\xi,y,\eta) \in \Lambda\}$ will be used throughout for canonical relations as well as for wavefront sets of distributional kernels. In particular if $K \in \mathrm{I}^*(M \times N, \Lambda')$ then $\mathrm{WF}(K) \subset \Lambda'$ whereas
$\mathrm{WF}'(K) \subset \Lambda$.

\subsection{Acknowledgement} The author would like to thank Edward Witten for the continued encouragement and for sharing many valuable insights, in particular that the existence of a stress energy tensor should imply the local time-slice axiom. I am also grateful to Klaus Fredenhagen, Jonathan Sorce, Rainer Verch, and Bob Wald for useful comments on an earlier draft version. Finally, I would also like to thank Jan Derezinski for insightful remarks about the properties of CCR algebras.

\subsection*{Data Availability} All data underlying the results are available as part of the article and no additional source
data are required.

\subsection*{Conflict of interest} The author of this work declares that he has no Conflict of interest.

\section{Mathematical setting and formulation of results}

Starting with a globally hyperbolic spacetime $(M,g_0)$,
let $\widetilde{\mathcal{M}_{g_0}}$ be the space of metrics $g$ on $M$ such that 
$g - g_0$ is smooth, compactly supported, and such that $(M,g)$ is a globally hyperbolic spacetime with time-orientations compatible with that of $(M,g_0)$. This space can be endowed with the topology inherited from the space of smooth compactly supported symmetric two-tensors.
We denote by $\mathcal{M}_{g_0}$ the connected component of $g_0$ in $\widetilde{\mathcal{M}_{g_0}}$.
Since
$\mathcal{M}_{g}=\mathcal{M}_{g_0}$ for any $g \in \mathcal{M}_{g_0}$ we will simply write  $\mathcal{M}$, thus singling out a connected component.
We then fix a Hilbert space $\calH = \calH_\mathcal{M}$ and assume that for 
each metric $g$ in $\mathcal{M}$ there exists a quantum field theory given by
a choice of von Neumann algebra $\calA_g(\calO)$ of bounded operators acting on $\mathcal{H}$ for each open subset $\calO$. Sometimes it may also be convenient to consider a smaller open subset $\mathcal{N} \subset \mathcal{M}$ of the full connected component $\mathcal{M}$.
We assume the usual properties
\begin{itemize}
 \item[(i)] $\calA_g(\calO_1) \subset \calA_g(\calO_2)$ if $\calO_1 \subset \calO_2$.
 \item[(ii)] $[\calA_g(\calO_1), \calA_g(\calO_2)]=\{0\}$ if $\calO_1$ and $\calO_2$ are spacelike separated, i.e. $\calO_2 \cap J^\pm_g(\calO_1) = \emptyset$,
\end{itemize}
where $J^\pm_g(K)$ denotes the causal future/past of a subset $K \subset M$ with respect to the metric $g$.
These are natural and very general assumptions that one expects to hold if the local field algebras are generated by a quantum field satisfying the Einstein causality property.
In addition we would also like to assume that the following {\sl diffeomorphism-covariance condition} holds
\begin{itemize}
 \item[(iii)]  Let $\mathrm{Diff}_\compp(M)$ be the group of all compactly supported diffeomorphisms $\phi: M \to M$, and  denote by $\mathrm{Diff}_{0,\compp}(M)$ the connected component of the identity in $\mathrm{Diff}_{\compp}(M)$.
Assume that $\phi \in \mathrm{Diff}_{0,\compp}(M)$ and $g' = \phi^* g$ is the pull-back of the metric $g$ with respect to $\phi$. Then
\begin{align*}
 \calA_{g'}(\calO) =  \calA_{g}(\phi(\calO)).
\end{align*}
\end{itemize}

The moduli space $\mathcal{M} / \mathrm{Diff}_{0,\compp}(M)$ of metrics modulo diffeomorphisms  will be denoted by $\mathcal{M}_\mathrm{mod}$.
The way we think of the implementation of metric change is that {\bf the stress energy tensor is a connection one-form on the bundle  $\mathcal{M}_\mathrm{mod} \times \calH$  over the space $\mathcal{M}_\mathrm{mod}$}. It will be convenient to consider the stress energy tensor as a connection one-form on the bundle $\mathcal{M}\times \calH$  over the space $\mathcal{M}$ that descends to the quotient space.
For each $g \in \mathcal{M}$ and each test function
$h \in C^\infty_0(M,\mathrm{Sym}^2 T^*M)$ this gives rise to an operator $T_g(h)$ with domain $\mathrm{dom}(T_g(h))$. We also expect that there exists a dense set $\mathcal{D}_g \subset \calH$ of vectors $v$ such that the map $h \mapsto T_g(h) v$ is linear and continuous in $h$. 
Hence, by a stress energy tensor we will mean a family of operator-valued distributions
$T_g$, indexed by elements $g \in \mathcal{M}$, and test function space $C^\infty_0(M,\mathrm{Sym}^2 T^*M)$ such that some natural conditions are satisfied. Formally, in local coordinates, the pairing of the distribution
$T_g$ with a test function $h$ is 
$$
 T_g(h) = \int_M T^{jk}(x) h_{jk}(x)\; \der \mathrm{Vol}_g(x),
$$
where the Einstein sum convention is assumed, and $h_{jk}$ are the components of the tensor $h$. Of course the above expression is not an integral but should be interpreted as a distributional pairing and in that sense $T^{jk}$ is an operator-valued distribution expressed in local coordinates.
The definitions and analysis is made rather complicated by the fact that the space $\mathcal{M}$ is infinite dimensional and that the operator $T_g$ is unbounded with domain potentially dependent on $g$. We will take a slightly indirect approach in defining a connection on $\mathcal{M} \times \calH$ via a notion of parallel transport. In finite dimensions this is an equivalent approach to choosing a connection one form, and we will be able to avoid technical definitions at this stage.
Before we do this one (unfortunately) needs to clarify all the notions of smoothness and continuity on the space of paths due to the infinite dimensional nature of $\mathcal{M}$.

As before assume that $\mathcal{N} \subset \mathcal{M}$ is a non-empty open subset of the space $\mathcal{M}$. In most cases we will be interested in fact in the case $\mathcal{N} = \mathcal{M}$, but allow more general situations.
We call a continuous path $\gamma: [a,b] \to \mathcal{N}$ {\sl piecewise smooth} if there exists a partition 
$a =t_0 < t_1 < t_2 < \ldots < t_N =b$ such that $\gamma|_{[t_j,t_{j+1}]}$ is a smooth section of 
the pull back of the bundle $\mathrm{Sym}^2 T^*M$ to 
$[t_j,t_{j+1}] \times M$. We identify paths that are smooth orientation preserving reparametrisations of one another, i.e.
$\gamma: [a,b] \to \mathcal{N}$ is equivalent to $\gamma': [a',b'] \to \mathcal{N}$ if there exists a diffeomorphism $q:[a',b'] \to [a,b]$ such that $\gamma' = \gamma \circ q$.
We will denote the set of equivalence classes of piecewise smooth paths by $\mathcal{P}_\mathcal{N}$ and for $g',g$
we write $\mathcal{P}_\mathcal{N}(g',g)$ for the set of piecewise smooth curves $\gamma: [a,b] \to \mathcal{N}$ with $\gamma(a) = g, \gamma(b) = g'$. Similarly we write $\mathcal{P}_\mathcal{N}^\infty$ for the set of smooth paths and $\mathcal{P}_\mathcal{N}^\infty(g',g)$ for $\mathcal{P}_\mathcal{N}^\infty \cap \mathcal{P}_\mathcal{N}(g',g)$.
The set of piecewise smooth paths form a groupoid as they can be concatenated in case the endpoint of the first point coincides with the second curve. 
We will write $\gamma' \circ \gamma$ for the composition of the paths as there is no danger of confusion with the composition of functions. The path $\gamma' \circ \gamma$ is the one starting at the start point of $\gamma$ and ending at the endpoint of $\gamma'$, i.e. the notation is read from the right to the left.
The inverse path $\gamma^{-}: [a,b] \to \mathcal{N}$ is defined as $\gamma^{-}(t) = \gamma(b+a-t)$.

For any $\gamma:[a,b] \to \mathcal{N}$ and any $s,t \in [a,b]$ with $s>t$
we also define $\gamma_{s,t}$ to be the restriction of $\gamma$ to $[t,s]$ with starting point $\gamma(t)$ and endpoint $\gamma(s)$.
In the following topology and smoothness on the infinite dimensional space will be defined by restriction to finite dimensional parametrisations.
In order to do this we use parameter manifolds. A parameter manifold $I$ is a smooth differentiable manifold with or without boundary, and of arbitrary but finite dimension. 
We say a family of smooth paths $\gamma_\lambda \in \mathcal{P}_\mathcal{N}^\infty$ depends smoothly on a parameter $\lambda$ in some parameter manifold $I$ if the domain of the family of paths in $I \times \R \times M$ is a smooth manifold with corners or with boundary, and if the map $\gamma_\lambda(t)$ is a smooth section of pull back of the bundle $\mathrm{Sym}^2 T^*M$. Hence, a family of paths is smooth if and only if the end-points depend smoothly on the parameter and the function
$(s,t,x) \mapsto \gamma_s(t)(x) \in \mathrm{Sym}^2 T^*_{x} M$ is a smooth function of $s,t$ and $x$, where defined.

Finally we say a family of piecewise smooth paths $\gamma_\lambda$ depends smoothly on a parameter $\lambda$ if it can be written as
$\gamma_{1,\lambda} \circ \ldots \circ \gamma_{N,\lambda}$ for a finite number of smooth paths $\gamma_{j,\lambda} $ depending smoothly on $\lambda$.
It will be convenient to endow the set of paths in $\mathcal{P}_\mathcal{N}$ with a topology. Here we choose the final topology with respect to all smooth maps
$I \to \mathcal{P}_\mathcal{N}$, i.e. the finest topology on $\mathcal{P}_\mathcal{N}$ such that any smooth map $I \to \mathcal{P}_\mathcal{N}$ is continuous.
Hence, a map $F: \mathcal{P}_\mathcal{N} \to T$ into an arbitrary topological space is continuous iff the map $F \circ \iota$ is continuous for any smooth
map $\iota: I \to \mathcal{P}_\mathcal{N}$.

\begin{definition} \label{defconnection}
 A {\sl unitary connection} on the trivial bundle $\mathcal{N} \times \calH \to \mathcal{N}$ is a map $S: \mathcal{P}_\mathcal{N} \to U(\calH)$ from $\mathcal{P}_\mathcal{N} $ into the unitary group $U(\calH)$ of $\calH$ together with a family $(\calH^\infty_g)_{g \in \mathcal{N}}$ of dense subsets  $\calH^\infty_g \subset \calH$
 such that
 \begin{enumerate}
  \item In case $\gamma \in \mathcal{P}_\mathcal{N}(g_2,g_1), \gamma' \in \mathcal{P}_\mathcal{N}(g_3,g_2)$ we have $S(\gamma' \circ \gamma ) = S(\gamma')S(\gamma )$.
  \item $S(\gamma^{-}) = S(\gamma)^{-1} = S(\gamma)^*$ for any $\gamma \in \mathcal{P}_\mathcal{N}$.
  \item 
  The map $$\mathcal{P}_\mathcal{N} \to U(\calH), \gamma \mapsto S(\gamma)$$ is strongly continuous.   
  \item for any $\gamma \in \mathcal{P}_\mathcal{N}(g',g)$ we have $S(\gamma) \calH^\infty_{g} \subset \calH^\infty_{g'}$.
  \item for any smooth family of curves
  $\gamma_\lambda$ with $\lambda$ in a parameter manifold $I$ the map
  $$
   I \to \calH, \quad \lambda \mapsto   S(\gamma_\lambda) v
  $$
  is smooth for any $v \in \calH^\infty_g$ in case all the paths have the same starting point $g$.
  \item for any $\gamma \in \mathcal{P}_{\mathcal{N}}(g',g)$ and  $v \in \calH^\infty_g$ the derivative $\frac{\der}{\der t} S(\gamma_{t,a}) v$ depends only on the tangent vector
   $\dot \gamma(t)$ and in a linear fashion. In other words there exists for each $g \in \mathcal{N}$ a (strongly continuous) operator valued distribution $T_g$
   taking values in the operators from $\calH^\infty_g$ to $\calH$ such that
   \begin{align}
 \rmi \frac{\der}{\der t} S(\gamma_{t,s}) v = T_{\gamma(t)}( \dot \gamma(t) ) S(\gamma_{t,s})  v
\end{align}
for any $v \in \calH^\infty_{\gamma(s)}$ for any $\gamma \in \mathcal{P}_\mathcal{N}^\infty$.
 \end{enumerate}
 We say the connection has central holonomy if $S(\gamma) \in U(1) \cdot \mathbf{1}_\calH$ for every $\gamma \in \mathcal{P}_\mathcal{N}(g,g)$.
\end{definition}

In case $\mathcal{N}$ is invariant under the action of $\mathrm{Diff}_{0,\compp}(M)$ we can introduce the associated moduli space $\mathcal{N}_\mathrm{mod}$. Given a unitary connection as above that satisfies in addition $S(\tilde \gamma) = S(\gamma)$ if $\tilde \gamma(t) = \phi^*_t \gamma(t)$, $\phi_t \in \mathrm{Diff}_{0,\compp}(M)$ being a smooth path of diffeomorphisms, then we say $S$ defines a unitary connection on the moduli space $\mathcal{N}_\mathrm{mod}$.

The (possibly unbounded) operator-valued distribution $T_g$ is called the connection one form.
 Given a dense set  $\calH^\infty_g$ and a family $T_g(h)$ defined on $\calH^\infty_g$ we also say that this family generates the connection if there is a unique connection that has $T_g(h)$ as a connection one form. We will not discuss in this paper conditions on such operator-valued distributions to generate connections but merely just postulate the existence of a connection.
Given a smooth curve $\gamma$ finding the parallel transport map $S(\gamma)$ corresponds to finding the solution of the initial-value problem for the non-autonomous evolution equation
\begin{align}
 \rmi \frac{\der}{\der t} v(t) = A(t) v(t), \quad v(a) = v,
\end{align}
where $v(t) = S(\gamma_{t,a}) v$ and $A(t) = T_{g(t)}( \dot \gamma(t))$ for a dense set of vectors $v$. By continuity the unitary map $S(\gamma)$ will then extend to all of $\calH$.
Sufficient conditions for a family of self-adjoint operators $A(t)$ to generate a connection were found in increasing generality by Kato (\cite{MR0058861, MR0086986, MR0279626}) and  other authors (we refer to the survey \cite{MR1944168} and the results in \cite{MR2493564} which seem particularly well-suited for the problem at hand).

The notion of a connection is intimately related to the notion of unitary evolution system.
Given a topological space $X$ we say a family of bounded operators $U(s,t) , s,t \in X$ on $\calH$ forms a unitary evolution system if
\begin{itemize}
 \item $U(s,t)$ is unitary for every $s,t \in X$.
 \item the map $X \times X \to \mathcal{B}(\calH), (s,t) \mapsto U(s,t)$ is strongly continuous.
 \item $U(s,r) \circ U(r,t) = U(s,t)$ and $U(s,s)= \mathbf{1}$ for all $r,s,t \in X$.
\end{itemize}
We say the family $U(s,t) , s,t \in X$ on $\calH$ forms a projective unitary evolution system if there exists a 
continuous function $\sigma: X \times X \times X \to U(1)$ such that
\begin{itemize}
  \item $U(s,t)$ is unitary for every $s,t \in X$.
 \item the map $X \times X \to \mathcal{B}(\calH), (s,t) \mapsto U(s,t)$ is strongly continuous.
 \item $U(s,r) \circ U(r,t) = \sigma(s,r,t) U(s,t)$ and $U(s,s)= \mathbf{1}$ for all $r,s,t \in X$.
\end{itemize}
We will refer to the function $\sigma$ as the Schwinger cocycle. The above means that the family $U(t,s)$ forms a projective evolution system in the sense that the identity $U(s,r) \circ U(r,t) = U(s,t)$ holds if we consider the operators as maps on the projective space $\mathbb{P}\calH$. Alternatively one can also view $U(s,t)$ as elements in $U(\calH)/U(1)$.
Given a connection with central holonomy this defines a projective unitary evolution system 
$U(g',g)$ on $\mathcal{N}$ by 
$$
 U(g',g) = [S(\gamma)] \in U(\calH)/U(1)
$$
for any $\gamma \in \mathcal{P}_\mathcal{N}(g',g)$. Note that by the assumption of the holonomy being central $U(g',g)$ is well defined as its class does not depend on the chosen path.
We will see later that the stress energy tensor of the Klein-Gordon field is derived from the connection with central holonomy.

We will use the following notation that encodes that scattering matrix $S(g',g)$ from the metric $g$ to the metric $g'$ is well defined as soon as a path from $g$ to $g'$ is specified.

\begin{conv}
In case $g(t), t \in [a,b]$ is a family of piecewise smooth metrics in $\mathcal{N}$ with $g(a) = g$ we will also write
$S(g(t),g)$ for $S(\gamma)$, where $\gamma$ is the path $t \mapsto g(t)$. We will write
$S(g,g(t))$ for $S(\gamma)^*$. 
\end{conv}

\begin{definition} \label{defstress}
Assume we are given a quantum field theory on $M$ defined for all metrics on the space $\mathcal{N}$
and realised on the same Hilbert space $\mathcal{H}$, as defined before.
We say that this quantum field theory has a stress energy tensor if there exists a unitary connection $S: \mathcal{P}_\mathcal{N} \to U(\calH)$ with central holonomy such that the following conditions are satisfied.

\begin{stressaxioms}
\stressitem{defstress:covariance}{Covariance}
Let $\gamma \in  \mathcal{P}_\mathcal{N}(g',g)$ and assume that $\phi_t \in \mathrm{Diff}_{0,\compp}(M)$ is a smooth path of diffeomorphisms such that the new path
$\tilde \gamma(t) = \phi^*_t \gamma(t)$ is in $\mathcal{P}_\mathcal{N}$. Then
\begin{align*}
  S(\tilde \gamma) = S(\gamma).
\end{align*}{\vspace{0.1cm}}
\stressitem{defstress:locality}{Locality}
Assume that $\calO \subset M$ is an open subset and $\gamma \in \mathcal{P}_\mathcal{N}(g',g)$ is a piecewise smooth path such that $\supp(\gamma(t)-g) \subset \calO$ for all $t \in [a,b]$. Then 
\begin{enumerate}
   \item \label{loc1} $\calA_{\gamma(t)}(\calO) = \calA_g(\calO)$, i.e. $\calA_{\gamma(t)}(\calO)$  is independent of $t$ and 
  \item \label{loc2} $S(\gamma) \in \calA_g(\calO)$.
\end{enumerate}${}^{}$
  {\vspace{0.1cm}}
\stressitem{defstress:causality}{Causality}
  \begin{enumerate}
 \item \label{caus1} If $\mathrm{supp}(g'-g) \cap J^-_g(\calO) = \emptyset$ then $\calA_{g'}(\calO) = \calA_{g}(\calO)$.
 \item \label{caus2} If for all $t$ in the parameter range of $g(t)$ we have $\mathrm{supp}(g(t)-g) \cap J^+_g(\calO) = \emptyset$ then $\calA_{g(t)}(\calO) = S(g(t),g) \calA_{g}(\calO) S(g(t),g)^*$.
 \item \label{caus3} Suppose that $h_1(t),h_2(t),h_3(t) \in C^\infty_0(M,\mathrm{Sym}^2 T^*M), t \in [a,b]$ are smooth paths of compactly supported smooth sections. Assume that $g+h_1(t)+h_2(t)+h_3(t), g+ h_2(t)+h_3(t),g+h_2(t),g+h_1(t)+h_2(t),g \in \mathcal{N}$ for all $t \in [a,b]$
and $h_1(a) = h_2(a) = h_3(a) = 0$.
 Assume further that $\left(J^-_{g+h_2(t)}\left( \mathrm{supp}\;h_1(t) \right)\right) \cap \mathrm{supp} h_3(t) = \emptyset$ for any $t \in [a,b]$. Then we have
  \begin{align}
   S(g,g+h_1(t) + h_2(t) + h_3(t)) \nonumber\\=  S(g,g + h_2(t) + h_3(t)) S(g+h_2(t),g+h_1(t)+h_2(t))
  \end{align}
  modulo a factor in $U(1)$.
 \end{enumerate}
\end{stressaxioms}
\end{definition}

We will call $S(\gamma)$  the {\sl scattering matrix} as it implements the metric perturbation in the Hilbert space relative to the unperturbed evolution along the path $\gamma \in \mathcal{P}_\mathcal{N}(g',g)$.

\begin{rem}
The condition $\left(J^-_{g+h_2(t)}\left( \mathrm{supp}\;h_1(t) \right)\right) \cap \mathrm{supp} h_3(t) = \emptyset$ assures that the supports of $h_1$ and $h_3$ can be separated by a Cauchy surface with respect to the metric $g+h_2$ (see Prop. \ref{sepper}).
 The equation
 $$
  S(g,g+h_1 + h_2 + h_3) = S(g,g + h_2 + h_3) S(g+h_2,g+h_1+h_2) \mod U(1)
  $$
  can also be written as
 $$
  S(g,g+h_1 + h_2 + h_3) = S(g,g + h_2 + h_3) S(g,g+h_2)^* S(g,g+h_1+h_2) \mod U(1).
 $$
 Causality conditions of this form are at the heart of causal perturbation theory where test functions $h_k$ play the role of localised interactions.
 In this paper we require this relation to hold only projectively, i.e. up to a phase factor in $U(1)$, although one might also be tempted to enforce this condition exactly.
 It was pointed out to the author by Klaus Fredenhagen that the problem of fixing these phases has been addressed in \cite{MR4226454} for perturbations of Minkowski space. Without considering continuity and differentiability, cohomological arguments show that it is possible to remove these phases in this context.
 \end{rem}

\subsection{Discussion of the conditions}

The above are ``exponentiated" versions of expected properties of a (quantum) stress energy tensor.
The infinitesimal characterisation  of covariance under diffeomorphisms corresponds directly to
$T_g$ being divergence free, i.e. $T_g(\der_s v)=0$ for any compactly supported co-vector field $v$. In local coordinates this means formally $\nabla_j T_g^{jk}=0$. Here $\der_s v$ is the symmetric exterior derivative given in local coordinates by $(\der_s v)_{jk} = \nabla_j v_k + \nabla_k v_j$. 

In case $T_g(h)$ is self-adjoint or possesses a reasonable functional calculus we expect that the spectral projections of $T_g(h)$ are contained in $\calA_g(\calO)$ whenever we have the inclusion $\mathrm{supp}(h) \subset \calO$. In other words $T_g$ should be a local field.
This is the infinitesimal version of the exponentiated operators $S_g$ being affiliated to local algebras.
Part of the causality statement can be understood as requiring the generators $T_g(h_1),T_g(h_2)$ to be commuting operators whenever the supports of $h_1$ and $h_2$ are not causally related.

The locality assumption  $\calA_{g(t)}(\calO) = \calA_{g(0)}(\calO)$ in case $\mathrm{supp}(g(t)-g(0)) \subset \calO$  shows independence of $t$ of the algebra $\calA_{g(t)}(\calO)$. This is stated here for technical reasons as an independent property, but it should really be seen as a consequence of $S(\gamma) \in \calA_g(\calO)$ and other much milder assumptions. 
The main technical complication here is that the formula for the change of algebras depends on the causal support properties of the metric perturbation relative to the support properties of the observable. On the level of fields this leads to time-ordering considerations. 
Accordingly there are several angles how one can understand the relations between the two parts {\sl(\ref{loc1})} and {\sl(\ref{loc2})} of the locality
axiom \eqref{defstress:locality} and we will discuss two of them here.  The first is a physics viewpoint that is based on the existence and locality of certain time-ordered products and we will discuss this in a non-rigorous manner as a plausibility argument. The second point of view will show that a relatively weak excision property can be used to circumvent these time-ordering problems and thereby allows to derive the locality property {\sl(\ref{loc1})} from the second locality property {\sl(\ref{loc2})}.

We start with the first non-rigorous argument that uses the existence of certain time-ordered products affiliated to the local algebras. Here we assume that the local algebras are von Neumann algebras affiliated to a field operator $\Phi(\cdot)$.
The infinitesimal variation of the field operator is given in terms of a time-ordered product with the stress energy tensor
\begin{align}
 \rmi \frac{\der}{\der t}\Phi_{g(t)}(\cdot) = \rmi \dot  \Phi_{g(t)}(\cdot) = T_{g(t)}(\dot g(t)) \Phi_{g(t)}(\cdot)  - \mathtt{T}(\Phi_{g(t)}(\cdot)  T_{g(t)}(\dot g(t))).
\end{align}
Here we write $\dot f = \frac{\der}{\der t} f(t)$ to  abbreviate the derivative with respect to the parameter $t$.
The time ordered product $\mathtt{T}$ has the meaning implied by parts {\sl (\ref{caus1})} and {\sl (\ref{caus2})} of the causality axiom
\eqref{defstress:causality}  in Definition \ref{defstress}:
$$
 \mathtt{T}(\Phi_{g(t)}(x)  T^{jk}_{g(t)}(y)) = \begin{cases} \Phi_{g(t)}(x)  T^{jk}_{g(t)}(y) & x \in J^+(y) \\ T^{jk}_{g(t)}(y)  \Phi_{g(t)}(x) & x \notin J^+(y) \end{cases},
$$
understood in an appropriate regularised sense, as these are distributions. It is thus assumed that such a time-ordered product can be made sense of.
This infinitesimal change also will affect the stress energy tensor. We think of it as an expression of the local field and the metric. Its variation will then be
\begin{align}
 \frac{\der}{\der t} T_{g(t)}(\dot g(t)) = F_{g(t)}(\dot g(t),\Phi_{g(t)}) + T_{g(t)}(\ddot g(t)),
\end{align}
given that the stress energy commutes with itself. Here $F_{g(t)}(\dot g,\Phi_{g(t)})$ is a field depending on $g(t),\dot g(t)$ and has been constructed locally from the fields $\Phi_{g(t)}$. This gives the following first order system
\begin{align}
  \rmi \frac{\der}{\der t}\Phi_{g(t)}(\cdot) &= \rmi \dot  \Phi_{g(t)}(\cdot) = T_{g(t)}(\dot g) \Phi_{g(t)}(\cdot)  - \mathtt{T}(\Phi_{g(t)}(\cdot)  T_{g(t)}(\dot g(t)))\\
   \frac{\der}{\der t} T_{g(t)}(\dot g) &= F_{g(t)}(\dot g(t),\Phi_{g(t)})+  T_{g(t)}(\ddot g(t))= G(g(t),\dot g(t),\ddot g(t), \Phi_{g(t)}),
\end{align}
where $G$ is a local functional of $g(t),\dot g(t),\ddot g(t)$ and the field $\Phi_{g(t)}$. The construction of $G$ from the field may contain possible renormalisations as for example Wick-ordering in the case of free fields. However, both the time-ordered product and the functional $G$ are local in the following sense: If $g(t)-g(0)$ is compactly supported in $\calO$ then both
$G(g(t),\dot g(t),\ddot g(t), \Phi_{g(t)})$ and $\dot  \Phi_{g(t)}(\cdot)$ are affiliated with $\calA_{g(t)}(\calO)$.
For notational simplicity we write
\begin{align}
 H\left( g(t),\dot g(t),\ddot g(t), \Phi, T \right)=   \left(\begin{matrix} -\rmi ( T \Phi  - \mathtt{T}(\Phi T)) \\ G(g(t),\dot g(t),\ddot g(t), \Phi) \end{matrix} \right).
\end{align}
Assume now that this initial-value problem can be solved for small $t>0$ by a fixed point iteration of the type
\begin{align}
 \left(\begin{matrix} \Phi_{k+1,g(t)} \\ T_{k+1,g(t)}(\dot g(t)) \end{matrix}\right) = \left(\begin{matrix}  \Phi_{g(0)} \\ T_{g(0)}(\dot g(0)) \end{matrix} \right) + \int_{0}^t H\left( g(s),\dot g(s),\ddot g(s), \Phi_{k,g(s)}, T_{k,g(s)}(\dot g(s)) \right)\der s.
\end{align}
with start values $(\Phi_{0,g(t)}, T_{0,g(t)}(\dot g(t))) = (\Phi_{g(0)}, T_{g(0)}(\dot g(0)))$.
We then see that in each step of the iteration we remain affiliated to $\calA_{g(0)}(\calO)$ in case the field $\Phi_{g(0)}$ is localised in $\calO$.
In the case of the Klein-Gordon field this argument is plagued by technical difficulties because of the unboundedness of the field operators. In case of the Dirac field the field operators are bounded operators and it may be easier to establish the well-posedness in the space of bounded operators as well as weak-$*$-convergence of the fixed point iteration. 

Back to a rigorous setting the condition $\calA_{g(t)}(\calO) = \calA_{g(0)}(\calO)$ can also be replaced by considerably weaker and more plausible assumptions.
One example of a weak physically plausible assumption is that the local algebras are additive and have an excision-type property:
Given any acausal spacelike hypersurface $C$ whose closure is contained in an open set $\calO$ a possible weak excision property would be that for every $g \in \mathcal{M}$ there exists a $C^1$-open neighborhood $\calV$ of $C$ in $\calO$ such that $\calA_{\tilde g}(\calO) = \calA_{\tilde g}(\calO \setminus \overline{\calV})$ for every $\tilde g \in \mathcal{M}$ in an open neighborhood of $g$. This weak excision property is physically very plausible. It means information cannot be trapped in a thin set $\calV$ but has to propagate into $\calO$ and this thin set can be chosen uniformly at least for small perturbations of the metric. This weak excision property together with additivity allows us indeed to conclude that
$\calA_{g+h}(\calO) = \calA_g(\calO)$ whenever $h$ is supported in $\calV$. Given a spacelike hypersurface $C$ with closure in $\calO$ we can cover the set $\calO \setminus \overline{\calV}$ by open sets $\calU_\alpha$ such that either $J^+(\calV) \cap \calU_\alpha = \emptyset$ or $J^-(\calV) \cap \calU_\alpha = \emptyset$. In the former case
$\calA_{g+h}(\calU_\alpha) = \calA_{g}(\calU_\alpha)$ whereas in the latter case
$\calA_{g+h}(\calU_\alpha) = S(g+h,g) \calA_{g}(\calU_\alpha) S(g+h,g)^*$. Since the operator $S(g+h,g)$ is affiliated to 
$\calA_g(\calV) \subset \calA_g(\calO) =  \calA_g(\calO \setminus \overline{\calV})$, this shows that each of the algebras is contained in $\calA_g(\calO)$. Then, by additivity $\calA_{g+h}(\calO) = \calA_g(\calO)$.

This then implies $\calA_{g+h}(\calO) = \calA_g(\calO)$ for general $h$ supported in $\calO$. To see this we note that, by compactness of the unit interval, it is sufficient to show this for $h$ that are sufficiently small in the $C^1$-norm so that $g + \chi h \in \mathcal{M}$ for any smooth cut-off function $\chi \in C^\infty_{0}(M)$ with $0 \leq \chi \leq 1$. Any such $h$ can then be decomposed into a finite sum of metric perturbations $h = \sum_{j=1}^N h_j$ where each of the $h_j$ is supported in a set of the form $\calV_j$ as above.
Such a decomposition can for example be achieved by using a foliation of $M$ by Cauchy surfaces $\Sigma_t$, and choosing as achronal sets $C_t = \Sigma_t \cap \mathrm{supp}(h)$. Then the associated sets $\calV_t$ form an open cover of the compact set $\mathrm{supp}(h)$ from which we can extract a finite sub-cover $(\calV_{t_j})_{j=1,\ldots,N}$. Now one uses a partition of unity $\chi_j$ subordinate to this cover to construct $h_j = \chi_j h$ and use that $\calA_{\tilde g + h_j}(\calO) = \calA_{\tilde g}(\calO)$.

\subsection{Discussion of possible applications}

In contrast to covariant quantum field theory the above structures allow directly for formulations of
spectrum conditions and analyticity. We give examples of possible approaches below.

\begin{definition}
A vector $v \in \calH$ is called a $C^k$-vector for the metric $g \in \mathcal{M}$ if for any smooth family $(\gamma_\lambda)_{\lambda \in Q}$ of paths emanating from $g$ the family
 $S(\gamma_\lambda)v$ is in $C^k(Q,\calH)$. Similarly, we call a vector $v \in \calH$ analytic if for any linear path $g(t) = g + t h$
 the family $S(g(t))v$ is real analytic.
\end{definition}

In the following denote by $\overline{V^+}$ the closed forward light cone in $T^*M$, i.e. the set of future pointing causal covectors together with the zero vector.

\begin{definition} \label{sdfjsdfljksndflknsdF}
 A theory satisfies the state-independent positive spectrum condition if for any compactly supported covector field $v$ taking values in $\overline{V^+}$
 we have that $T_g(v \otimes v)$ is semibounded below.
 \end{definition}
 
 The condition in Definition \ref{sdfjsdfljksndflknsdF} is a state-independent QEI-type condition.
This is stronger than the generalized, state-dependent lower bounds that are
available in some models. In particular, for the non-minimally coupled scalar
field, state-independent lower bounds may fail, while generalized
state-dependent QEIs can still be nontrivial \cite{FewsterOsterbrinkQEI}; see
also \cite{FewsterLecturesQEI} for a general introduction to QEIs.

Our analysis of Hadamard states for parameter-dependent systems as Lagrangian
distributions in Appendix \ref{Appendparam} suggests another possible
definition of a spectrum condition, close to a microlocal spectrum condition
for the stress-energy tensor. Morally, the wavefront set of
$$
  \langle \psi,S(\gamma)\varphi\rangle
$$
should be contained in the dual cone to the set of metric variations for which
the infinitesimal stress-energy generator is positive. 

Thus, if the
stress-energy tensor satisfies a dominant, weak, or null energy condition, the
corresponding scattering matrix should have wavefront set in the polar cone of
the corresponding positive metric variations. Of course, this cannot be taken
literally on the full infinite-dimensional space $\mathcal M$ without
specifying an appropriate calculus on that space. However, the condition has a
meaningful finite-dimensional formulation for parameter spaces $I$
parametrising families of metrics.

Let $\gamma:I\to\mathcal M$ be a smooth family of metrics.
For $\lambda\in I$, let $C_{\gamma,\lambda}\subset T_\lambda I$
be the cone consisting of all $\mathcal X\in T_\lambda I$ for
which
$$
 \der\gamma_\lambda(\mathcal X)
 =
 \sum_{j=1}^N (\alpha_j \otimes \beta_j+\beta_j\otimes\alpha_j),
$$
where the $\alpha_j$ and $\beta_j$ are compactly supported
smooth covector fields taking values in $\overline{V^+_{g_\lambda}}$.

We define
$$
 C_{\gamma,\lambda}^{\circ}
 =
 \left\{
   \tau\in T_\lambda^*I:
   \tau(\mathcal X)\leq0
   \text{ for every }
   \mathcal X\in C_{\gamma,\lambda}
 \right\},
$$
and set
$$
 C_\gamma^\circ
 =
 \bigcup_{\lambda\in I}
 C_{\gamma,\lambda}^\circ
 \subset T^*I.
$$

\begin{definition}\label{sjaalksdoihwlekEFSW}
A theory satisfies the dominant microlocal spectrum condition if, for every
finite-dimensional parameter space $I$, every smooth family
$\gamma:I\to\mathcal M$, and all vectors $\psi,\varphi\in\calH$, the continuous function
$$
  I \to \C,\quad \lambda\mapsto \langle \psi,S(\gamma(\lambda))\varphi\rangle
$$
satisfies
$$
 \WF\bigl(\langle\psi,S(\gamma)\varphi\rangle\bigr)
 \subset C_\gamma^\circ .
$$
Here $S(\gamma(\lambda))$ is defined using any smooth family
$\Gamma_\lambda$ of paths from a fixed reference metric to
$\gamma(\lambda)$. By smoothness of the connection and central
holonomy, any two such choices differ by a smooth
$U(1)$-valued function. Since multiplication by a smooth
nowhere-vanishing function preserves the wavefront set, the condition is independent of these choices.
\end{definition}
It is clear that such a condition then does also not depend on the choice of background metric and the choice of path. Other possible conditions can be formulated by changing the cone of metrics.
In this paper we will however not focus on these aspects. We will instead show that the time-slice axiom and the implementability of Killing flows in the Hilbert space of the theory are consequences of the existence of a stress energy tensor. The first result is that in a quantum field theory with a stress energy tensor the local time-slice property holds (Theorem \ref{lts}) and local isometries can be implemented unitarily in the theory (Theorem \ref{li}).

We also show that the Klein-Gordon field is an example and expect the same for other fields.
In particular Theorem \ref{KGExamplTh} will state that the Klein-Gordon field with mass $m \geq 0$ defined with respect to any pure quasifree Hadamard state is a quantum field theory with stress energy tensor. The corresponding connection has central holonomy and therefore also defines a projective evolution system on $\mathcal{M}$.
This involves the proof of implementability of the classical time evolution in Fock space and its differentiability properties. This is done in detail in Section \ref{implsecdiff} and formulated as Theorem \ref{fptheorem}.

\section{Properties of theories with stress energy tensors}

%

In what follows we assume that we are given a quantum field theory defined on all metrics in $\mathcal{M}$ with a stress energy tensor as defined before. 

\subsection{The local time-slice property} \label{ltssec}

Here we demonstrate that covariance, causality, and locality in fact imply a local time-slice property. Recall that the future/past Cauchy developments $D^\pm(C)$ of an achronal set $C \subset M$ are defined as the set of points $p \in M$ such that every past/future inextendible causal curve emanating from $p$ intersects $C$. The Cauchy development $D(C) = D^+(C) \cup D^-(C)$ is also called the domain of dependence.
If $C$ is a spacelike hypersurface then $D(C)$ is open in $M$ and globally hyperbolic with Cauchy surface $C$.
We refer to \cite{MR0424186} and \cite{ON} for the definitions and basic theory.

\begin{lemma} \label{sevenlemma}
 Assume that $\Sigma \subset M$ is any spacelike Cauchy hypersurface and $C \subset \Sigma$ is an open subset.  Let $\calO$ be any open neighborhood of $\overline{C}$ in $M$. Then, for any open subset $\calU$ of compact closure contained in $D^+(C)$ or $D^-(C)$ there exists an open subset $\calO_1$ with compact closure in $\calO$ such that
 $\calA_g(\calU) \subset \calA_g(\calO_1) \subset \calA_g(\calO)$.
\end{lemma}
\begin{proof}
Suppose as required that $\calU$ is of compact closure in $D^\pm(C)$.
We then construct a smooth path of metrics $g''(s)$ and $g'(s)$, $s \in [0,s_0]$ in $\mathcal{M}$ and a smooth family of diffeomorphisms $\phi_s \in \mathrm{Diff}_{0,\compp}(M)$
with the properties (see Fig. \ref{causaldomain}),
\begin{itemize}
 \item $g''(0) = g'(0)=g$,
 \item for all $s \in [0,s_0]$ we have $g''(s)= \phi^*_s g'(s)$,
 \item for all $s \in [0,s_0]$ we have $\mathrm{supp}(g'(s)-g) \subset \calO$,
 \item $\mathrm{supp}(g''-g) \cap J^\mp(\calU) = \emptyset$,
 \item $\phi(\calU) \subset \calO$,
\end{itemize} 
where we write for short $g''=g''(s_0), g'=g'(s_0)$ and $\phi = \phi_{s_0}$.

\begin{figure}[h!] 
\centering
\begin{tikzpicture}[scale=0.8]
 \fill[gray,opacity=0.4] (3.52, 3.2) ellipse (0.8 and 1.28);
 
  \filldraw[thin,fill=gray!80,opacity=.2,dashed] (3.4,5.2) .. controls (-2.0,-0.2) and  (0.,-0.1) .. (1.0,-.3) .. controls (7,-0.6) and  (11,-0.9) .. (3.4,5.2);
  \fill[gray,opacity=0.2] (3.2,0) ellipse (7.2 and 0.96);
 \draw[thick,opacity=0.8] (-3.36,0.0) ..  controls (2.4,-0.32) and  (4,-1.12) .. (9.28,.0);
 \filldraw[thin,fill=gray!20,dashed,opacity=0.2] (-3.36,0.0) -- (2.96, 6.32) -- (9.28,0.0);

  \fill[gray,opacity=0.3] (3.52,0.32) ellipse (1.5 and 0.58);
   \draw[ultra thin, ->,>=stealth] (2.9,3.3) -- (2.7,1.6);
    \draw[ultra thin, ->,>=stealth] (4.0,3.3) -- (4.1,1.6);
    \draw[ultra thin, ->,>=stealth] (2.7,1.5) -- (2.6,0.5);
    \draw[ultra thin, ->,>=stealth]  (4.1,1.5) -- (4.2,0.5);
    \draw[ultra thin, ->,>=stealth] (3.42,3.6) -- (3.42,1.6);
     \draw[ultra thin, ->,>=stealth] (3.42,1.5) -- (3.42,0.7);
    \node at (-0.7, 1.5) {$D^+(C)$};
     \node at (-0.7,-0.5) {$C$};
     \node at (-0.8,0.4) {$\calO$};
     \node at (5.12, 1.5) {$\supp Z$};
     \node at (3.52, 0.32) {$\phi(\calU)$};
     \node at (3.8, 3.7) {$\calU$};
\end{tikzpicture}
\caption{Illustration of the proof of Lemma \ref{sevenlemma}}\label{causaldomain}
\end{figure}

We first explain how to construct the required family of metrics and the diffeomorphism. We will only consider the case $\overline{\calU} \subset  D^+(C)$, as the construction for $\overline{\calU} \subset  D^-(C)$ is the same with time-orientation reversed.

Choose a global Cauchy temporal function
$t:M\to\mathbb R$ with $\Sigma=t^{-1}(0)$. Since
$$
 A=J^-(\overline{\calU})\cap D^+(C)
$$
is compact and $D(C)$ is open, the construction in Lemma
\ref{vectorfieldZ} can be localized to a relatively compact
open flow tube $W\Subset D(C)$. We may choose $Z$, $s_0$, and
$\calV$ so that the flow is vertical with respect to $t$ on
$W$, $\supp Z\subset W$,
$$
 \phi_{s_0}(A)\subset\calV\Subset\calO,
 \quad
 \bigcup_{r\in[0,s_0]}
 \phi_{-r}(\overline{\calV})\subset W.
$$
We thus find a vector field $Z$, compactly supported in
$D^+(C)$, with time $s$ flow $\phi_s$, an $s_0>0$, and an open subset
$\calV$ with compact closure in $\calO$ such that
\begin{itemize}
  \item[(a)] $Z$ is either zero or past-directed timelike.
  \item[(b)] $\phi_{s_0}(J^-(\overline{\calU}) \cap D^+(C)) \subset \calV$.
  \item[(c)] the global Cauchy temporal function $t$ chosen
above is temporal for the metrics $\phi_s^*g$ on $\calV$
for all $s\in[-s_0,s_0]$.
  \item[(d)] there exists a vector field which is timelike for $g$ and whose
  restriction to $\calV$ is timelike for the metrics $\phi_s^*g$ for all
  $s\in[-s_0,s_0]$.
\end{itemize}

By construction the metric
$\tilde g(s)=(\phi_s^{-1})^*g$ equals to $g$ outside the support of $Z$.
Next we choose a compactly supported cutoff function
$\chi\in C^\infty_0(\calV)$ that equals to one near the subset
$\phi_{s_0}(J^-(\overline{\calU})\cap D^+(C))$ and is supported in an
open neighborhood $\calO_1$ of this set with compact closure in $\calV$.
This function can be constructed in such a way that $0\leq\chi\leq1$.

Since $\supp\chi\subset\calV$, conditions (c) and (d), together with
Proposition \ref{prop-app-1}, show for every $s\in[0,s_0]$ that
$$
 g'(s)=(1-\chi)g+\chi\tilde g(s)
$$
is Lorentzian and that $\der t$ is timelike for $g'(s)^{-1}$ on
$\supp\chi$. Outside $\supp\chi$ we have $g'(s)=g$, so these properties
hold globally. Since $g'(s)-g$ has compact support, Proposition
\ref{prop-app-2} shows that $g'(s)$ is globally hyperbolic for every
$s\in[0,s_0]$.

By construction $g'(s) -g = -\chi g + \chi \tilde g(s)$ is supported in $\calO_1$. Furthermore, we set 
$g''(s) = \phi^*_s g'(s)$ and since
$$
 g''(s)=(1-\phi_s^*\chi)\,\phi_s^*g+(\phi_s^*\chi)\,g
$$
we see that $g''(s_0) - g=0$ when $x \in J^-(\overline{\calU})$ or $x \notin \supp(Z)$.

We will now demonstrate that this will imply the statement of the lemma 
with any relatively compact open neighborhood $\calO_1$ of $\phi(\overline{\calU}) \cup \mathrm{supp}(g'-g)$.
Indeed, in case $\calU$ has compact closure in $D^+(C)$ we have
\begin{align}
  \calA_g(\calU) =  \calA_{g''}(\calU) = \calA_{g'}(\phi(\calU)) \subset  \calA_{g'}(\calO_1) = \calA_{g}(\calO_1),
\end{align}
where the last equality follows from locality \eqref{defstress:locality}{\sl (1)}, since the path from $g$ to
$g'$ is supported in $\calO_1$.

In case $\calU$ has compact closure in $D^-(C)$, apply the
time-reversed version of the above construction, and denote
the resulting objects by $Z$, $\phi_s$, $s_0$, $W$, $\chi$, $t$,
$g'(s)$, and $g''(s)$. The family $g'(s)$ defines a path from
$g$ to $g'$, which we denote by $\gamma$, and
$\widetilde\gamma(s)=\phi_s^*\gamma(s)$ is a path from $g$ to $g''$.
Put
$$
 A_-=J^+(\overline{\calU})\cap D^-(C),\qquad
 \chi_s=(\phi_{s_0-s})^*\chi,
$$
and define
$$
 \widehat\gamma(s)
 =
 \phi_s^*\bigl((1-\chi_s)g+\chi_s\phi_{-s}^*g\bigr).
$$

The supports of the functions $\chi_s$ are contained in $W$.
On $W$ the flow is vertical with respect to the global Cauchy
temporal function $t$. Hence Propositions
\ref{prop-app-1} and \ref{prop-app-2} show that
$\widehat\gamma$ is a path in $\mathcal M$.

Moreover,
$$
 \widehat\gamma(0)=g,\qquad \widehat\gamma(s_0)=g'',
$$
and, since for $x \in A_-$ we have
$$
 (\phi_s^*\chi_s)(x)=\chi(\phi_{s_0}(x))=1
$$
we conclude that
$$
 \supp(\widehat\gamma(s)-g)\cap J^+(\overline{\calU})
 =\emptyset
$$
for every $s\in[0,s_0]$. Hence causality
\eqref{defstress:causality}{\sl(2)} and central holonomy give
$$
 \calA_g(\calU)
 =
 S(\widehat\gamma)^*\calA_{g''}(\calU)S(\widehat\gamma)
 =
 S(\widetilde\gamma)^*\calA_{g''}(\calU)S(\widetilde\gamma).
$$

By covariance \eqref{defstress:covariance} we have $S(\gamma) = S(\tilde \gamma)$. Since the path $\gamma$ is supported in $\calO_1$ locality \eqref{defstress:locality}{\sl (2)} then states that $S(\gamma) \in \calA_{g}(\calO_1)$.
All together we then have
\begin{align}
  \calA_g(\calU) &= S(\tilde \gamma)^* \calA_{g''}(\calU) S(\tilde \gamma)= S(\gamma)^*  \calA_{g'}(\phi(\calU)) S(\gamma) \nonumber \\
  &\subset S(\gamma)^*  \calA_{g'}(\calO_1) S(\gamma) \subset \calA_{g}(\calO_1).\nonumber
\end{align}

\end{proof}

\begin{theorem} \label{lts}
Assume that $\Sigma \subset M$ is any spacelike Cauchy hypersurface and $C \subset \Sigma$ is an open subset.  Let $\calO$ be any open neighborhood of $\overline{C}$ in $M$. Then
 $\calA_g(\calU) \subset \calA_g(\calO)$ for any open subset $\calU$ of compact closure contained in $D(C)$.
\end{theorem}
 \begin{figure}[h!] 
\centering
\begin{tikzpicture}[scale=0.5]
 \fill[gray,opacity=0.4] (3.52, 0) ellipse (1.8 and 2.28);
 \node at (3.52, 0) {$\calU$};
  \fill[gray,opacity=0.2] (3.2,0) ellipse (7.2 and 0.96);
 \draw[thick,opacity=0.9] (-3.36,0.0) ..  controls (2.4,-0.32) and  (4,-1.12) .. (9.28,.0);
 \draw[thick,opacity=0.9] (-3.36,0.0) ..  controls (2.4,-4.72) and  (4,-5.12) .. (9.28,.0);
 \filldraw[thin,fill=gray!40,opacity=0.5,draw=none] (-3.36,0.0) .. controls (2.4,-3.6) and  (4,-4.0) .. (9.28,.0) .. controls (4,-5.8) and (2.4,-5.4)   ..(-3.36,0.0);
  \filldraw[thin,fill=gray!20,dashed,opacity=0.2] (-3.36,0.0) -- (2.96, 6.32) -- (9.28,0.0);
  \filldraw[thin,fill=gray!20,dashed,opacity=0.2] (-3.36,0.0) -- (2.96, -6.32) -- (9.28,0.0);
  \node at (-0.3,0.25) {$C$};
   \node at (1.3,-2.8) {$C_{t_-}$};
  \node at (6.4,0.4) {$\calO$};
 \node at (3.52, 0) {$\calU$};
\end{tikzpicture}

\caption{Illustration of the proof of Theorem \ref{lts}}\label{causaldomaina}
\end{figure}
 
\begin{proof}
 Since $D(C)$ is globally hyperbolic we can choose a global time function so that the level sets $C_t$ are Cauchy hypersurfaces. Since the closure of $\calU$ is compact there exists a sufficiently small $t_-<0$ such that
 $\overline{\calU}$ is a compact subset of the interior of $D^+(C_{t_-})$. We now apply the previous Lemma to find an open subset $\calO_2$ of compact closure within a neighborhood of $C_{t_-}$ contained in $D^-(C)$ so that $\calA_g(\calU) \subset \calA_g(\calO_2)$.
Since $\calO_2$ is a subset of compact closure in $D^-(C)$ we can again apply the previous lemma to conclude the proof.
\end{proof}

\subsection{Local implementability of Killing flows by unitaries} \label{implsec}

One important property of the stress energy tensor in Minkowski space is that suitable integrals over Cauchy hypersurfaces yield locally generators of the space-time translations. 
In a spacetime with a global Killing field the current $T_g(Z)$ with coordinates $T_g^{jk} Z_j$ is divergence free. At least on a formal level one expects that the operator defined by an integral 
$\int_\Sigma T_g^{jk} Z_j n_k \der x$
over any Cauchy hypersurface with future pointing normal covector field $n_k$ generates the Killing flow on the Hilbert space. There are essentially two problems with such a formula. The first is that the Cauchy surface may be non-compact and the integral may not make sense even as a quadratic form on a large enough set of vectors. The second problem is that the 
operator distributional current $T_g(Z)$ may not have a well defined restriction to the Cauchy hypersurface. It is relatively easy to get around the latter problem, as one can simply replace integration over one Cauchy surface by a smeared out version which averages over many Cauchy surfaces. In mathematical terms this is achieved by pairing $T_g(Z)$ with the derivative of a function that equals one in the far future and that vanishes in the far past. In the case of a spatially compact spacetime the derivative of such a function is compactly supported and the distributional pairing then defines an operator. In case the spacetime is not spatially compact one can still expect a local version of this construction as the parts contributing to the integral over a Cauchy hypersurfaces at large spatial separation of the region $\calO$ of interest commute with the algebra of that region. We may therefore modify the integral and integral only over the parts that are causally related to $\calO$.
This localisation is exactly what the following construction achieves and it shows that under our hypotheses the Killing flow may indeed be locally implemented.

We assume that $\calO$ is an open subset with compact closure so that the closure of the set $J^+_g(\calO) \cap J^-_g(\calO)$ is compact. We assume that $\calU$ is an open set with compact closure containing $\overline{J^+_g(\calO) \cap J^-_g(\calO)}$.
We assume that the metric on $\calU$ admits a Killing field $Z$. No timelikeness assumption on $Z$ is made; the construction applies to any
Killing field defined on $\calU$.
We will now modify the flow outside a neighborhood of $\calO$. For this we choose an open set $\calW$ such that
the closure of $\calO$ is contained in $\calW$ and so that $\calU$ can be covered by the following open sets
$\calW, \calU_+, \calU_-,\calU_0$ (see Fig. \ref{figcones}) such that
\begin{itemize}
 \item[(i)] the closures of $\calU_\pm$ and of $J^\mp(\calO)$ do not intersect, and
 \item[(ii)] the closures of $\calU_0$ and $J(\calO) =J^+(\calO) \cup J^-(\calO)$ do not intersect.
 \item[(iii)] the closures of $\calU_+$ and $\calU_-$ do not intersect.
\end{itemize}
\begin{figure}[H] 
\centering
\begin{tikzpicture}
\begin{axis}[
    axis lines=none,         
    xlabel={},
    ylabel={},
    grid=none,               
    xmin=-5, xmax=5,    
    ymin=-5, ymax=5,       
    width=8cm,
    height=8cm,
    domain=-3.5:3.5,
    axis line style={thick},
    enlargelimits
]

\node at (600,520) {$\calO$};
\node at (600,680) {$\mathcal{W}$};
\node at (600,950) {$\calU_+$};
\node at (600,100) {$\calU_-$};
\node at (110,510) {$\calU_0$};
\node at (1100,510) {$\calU_0$};
\node at (1070,10) {$\calU$};

\addplot[fill=gray, opacity=0.4, thick, samples=100] ({cos(deg(x))}, {sqrt(0.75)*sin(deg(x))}); 

\addplot[fill=gray, opacity=0.2, thick, samples=100] ({4*cos(deg(x))}, {4*sqrt(0.75)*sin(deg(x))}); 
  
\addplot[fill=gray, opacity=0.1, thick, samples=100] ({6*cos(deg(x))}, {6*sqrt(0.75)*sin(deg(x))});

\addplot[gray, dashed, thin, opacity=0.3,domain=1:4.6] {x-1};  
 \addplot[gray, dashed, thin, opacity=0.3, domain=1:4.6] {-x+1};  
\addplot[gray, dashed, thin, opacity=0.3, domain=-4.5:-1]  {x+1};  
 \addplot[gray, dashed, thin, opacity=0.3, domain=-4.5:-1] {-x-1};

\addplot[gray, dashed, thick, domain=4:5.7] {x-4};  
 \addplot[gray, dashed, thick, domain=4:5.7] {-x+4}; 
 \addplot[gray, dashed, thick, domain=-5.6:-4] {x+4}; 
 \addplot[gray, dashed, thick, domain=-5.6:-4]  {-x-4};

\end{axis}
\end{tikzpicture}
\caption{The sets $\calO,\calU,\mathcal{W}$ and the cover $\calU_\pm,\calU_0,\calW$.}\label{figcones}
\end{figure}

In case $\calU$ is considerably larger than $\calO$ we can also choose $\calW$ larger. Furthermore, in case the Killing field is globally defined and $M$ is spatially compact then the set $\calU_0$ might be empty.
Using a suitable cutoff function that equals one near $\calW$ we can now modify the Killing field $Z$ to a vector field $\tilde Z$ that is compactly supported in $\calU$ and equals $Z$ near $\calW$.
We denote by $\phi_t$ the one-parameter group of diffeomorphisms associated to $\tilde Z$.

Since $\tilde Z$ is not Killing the metric $g_t = \phi_t^* g = g + h(t)$ differs from $g$ by a symmetric tensor $h(t)$ and we can write
\begin{align}
h(t) = h_+(t) + h_-(t) + h_0(t),
\end{align} 
so that
$$
 \mathrm{supp}\,h_\pm(t) \subset \calU_\pm, \quad  \mathrm{supp}\,h_0(t) \subset \calU_0, \quad \phi_{\pm t}(\mathrm{supp}\,h_0(t)) \subset \calU_0
$$
for all $t$ in a sufficiently small interval  about $0$. Such a splitting can be achieved for example using a suitable partition of unity.
Further, for $T>0$ sufficiently small we have
$$
 g+h(t), g+ h_+(t), g + h_+(t) + h_0(t),  g + h_-(t),  g + h_-(t) + h_0(t) \in \mathcal{M}.
$$
for $t \in [-T,T]$ and we fix such a $T$.
 
\begin{definition} \label{killdef}
 The family of unitary operators $U_Z(t), t \in [-T,T]$ associated to the Killing field is defined as
 \begin{align}
  U_Z(t) = S(g+h_+(t) + h_0(t),g).
 \end{align}
\end{definition} 

A priori this family depends on the choice of modification $\tilde Z$ and on the splitting via the choice of cutoff functions.
We will see below that different choices of splitting result in a sense in equivalent operators. To analyse this we observe the following lemma.

\begin{lemma}  \label{indeplemmaflow}
 Assume that $k_0(t),k_0'(t)$ are smooth families of symmetric two tensors that are for each $t \in [-T,T]$ compactly supported in $\calU_0$, and such that
 $g+h_+(t) + k_0(t)$ and  $g+h_+(t) + k_0'(t)$ are in $\mathcal{M}$ for all $t \in [-T,T]$.
 Then
  $$
   S(g+h_+(t) + k_0(t),g) S(g+h_+(t) + k_0'(t),g)^*
  $$
  commutes with any element in $\calA_g(\calO)$ for all $t \in [-T,T]$.
\end{lemma}
\begin{proof}
 The expression equals 
 $$
  S(g+h_+(t) + k_0(t), g+h_+(t) + k_0'(t)).
 $$
 This can be written as in terms of equivalence classes modulo $U(1)$-factors as
 $$
  S(g+h_+(t) + k_0(t), g+h_+(t) + k_0'(t)) \sim S(g + h_+(t) + k_0(t), g + h_+(t)) S(g + h_+(t) , g+h_+(t) + k_0'(t))
 $$
 For each $t \in [-T,T]$ the right hand side is contained in
 $\calA_{g+h_+(t)}(\calU_0)$ by locality \eqref{defstress:locality},{\sl (\ref{loc2})}. It therefore commutes with every element in
 $\calA_{g+h_+(t)}(\calO)$. Since  $\calA_{g+h_+(t)}(\calO) = \calA_{g}(\calO)$ this shows the lemma.
\end{proof}

The operators $S(g+h_\pm(t) + k_0,g)$ are actually inverses of each other modulo operators that commute with $\calA_g(\calO)$
as the following lemma indicates.

\begin{lemma} \label{elevenlemma}
 Assume that $k_0(t)$ is a smooth family of symmetric two-tensors that is compactly supported in $\calU_0$ such that
 $g+h_+(t) + k_0(t)$ and $g+h_-(t) + k_0(t)$ are in $\mathcal{M}$ for all $t \in [-T,T]$ and such that $\phi_t(\supp k_0(t))$ is contained in $\calU_0$ for all $t \in [-T,T]$.
 Then there exist families of unitary operators $V(t)$ and $W(t)$ in $\calA_g(\calU_0)$ such that
  \begin{align}
   S(g+h_-(t) + k_0(t),g) V(t) S(g+h_+(t) + k_0(t),g) W(t) \sim \mathbf{1}
  \end{align}
  for all $t \in [-T,T]$.
\end{lemma}

\begin{proof}
Let $r_0(t)= \phi_{-t}^* (k_0(t) - h_0(t))$. Then $r_0(t)$ is supported in $\calU_0$ and 
$$
 \phi_{t}^* ( g +  r_0(t) ) = g + h_+(t) + h_-(t) + k_0(t).
$$
Hence, by covariance \eqref{defstress:covariance} and causality \eqref{defstress:causality}{\sl (3)},
\begin{align*}
 S(g + r_0(t),g) &= S(g+h_+(t) + h_-(t) + k_0(t),g) \\ &\sim S(g+h_-(t)  + k_0(t),g + k_0(t)) S(g+h_+(t)  + k_0(t),g) .
\end{align*}
Therefore,
\begin{align*}
  S(g+h_-(t)  + k_0(t),g) S(g,g+k_0(t))  S(g+h_+(t)  + k_0(t),g) S(g + r_0(t),g)^*  \sim \mathbf{1}.
\end{align*}
This proves the Lemma as $S(g,g+k_0(t))$ and $S(g + r_0(t),g)$ are in $\calA_g(\calU_0)$ by locality \eqref{defstress:locality}
{\sl(2)}.
\end{proof}

\begin{proposition}
 Let $U_Z(t)$ and $\tilde U_Z(t)$ be two families of unitary operators associated to the Killing field $Z$ on $\calU$ in the sense of Definition \ref{killdef} with respect to possibly different extensions of the Killing field, possibly different choices of open sets $\calU_\pm,\calU_0$, and possibly different decompositions $h = h_+ + h_- + h_0$.
 Then, for sufficiently small $T>0$ the operators $U_Z(t) \tilde U_Z(t)^*$ commute with any operator in $\calA_g(\calO)$ for any $t \in [-T,T]$.
 In particular, for any $A \in \calA_g(\calO)$ the action $\alpha_t(A) = U_Z(t)^* A U_Z(t), t \in [-T,T]$ is independent of the choices above.
\end{proposition} 
\begin{proof}
 From Lemma \ref{indeplemmaflow} we already know that different splittings and choices of open sets $\calU_0,\calU_\pm$
 result in the same unitary modulo possible factors commuting with $\calA_g(\calO)$.
 We therefore can assume that we have used the same cut-off functions for the splitting and the same open sets $\calU_0,\calU_\pm$.
 We thus can assume that the Killing field has been modified in two different ways to a compactly supported vector field in $\calU$ with corresponding flows  $\phi_t$ and $\tilde \phi_t$. These flows coincide on $\calW \cap \tilde \calW$ for $t \in [-T,T]$ for sufficiently small $T>0$. 
 This defines a diffeomorphism $\psi_t = \phi_t^{-1} \tilde \phi_t$ which acts trivially outside a compact set $K$ in $\calU_0 \cup \calU_+ \cup \calU_-$, i.e. $\supp \psi_t \subset K$.
 We then have the related decompositions of the pull-backs of the metric
 $$
  \psi_t^*(g + h_+(t) + h_0(t) + h_-(t)) =  g + \tilde h_+(t) + \tilde h_0(t) + \tilde h_-(t).
 $$
 We first show the theorem in three special cases.\vspace{0.2cm}\\
\noindent
\underline{Case 1: $\psi_t$ acts trivially outside of $\calU_0$:}
 In this case 
 $g + h_+(t) + h_0(t)$ differs from $g + \tilde h_+(t) + \tilde h_0(t)$ by a two-tensor supported in $\calU_0$. The statement of the theorem then follows immediately from Lemma \ref{indeplemmaflow}.
\vspace{0.2cm}\\
\noindent
\underline{Case 2: $\psi_t$ acts trivially outside of $\calU_-$:}
In this case $$\left( \tilde h_+(t) - h_+(t) \right) + \left( \tilde h_0(t) - h_0(t) \right) + \left( \tilde h_-(t) - h_-(t) \right)$$ is supported in $\calU_-$ and therefore
$\tilde h_+(t) - h_+(t)=0$. It follows that
 $$g + h_+(t) + h_0(t) = g + \tilde h_+(t) + \tilde h_0(t) + k_0(t)$$ with $k_0(t)$ supported in $\calU_0$ for sufficiently small $|t|$.  We can again use Lemma \ref{indeplemmaflow} to conclude the statement of the theorem.
 \vspace{0.2cm}\\
\noindent
\underline{Case 3: $\psi_t$ acts trivially outside of $\calU_+$:}
In this case $$\left( \tilde h_+(t) - h_+(t) \right) + \left( \tilde h_0(t) - h_0(t) \right) + \left( \tilde h_-(t) - h_-(t) \right)$$ is supported in $\calU_+$.
Comparing supports we see that $\psi_t^* h_-(t) = h_-(t) = \tilde h_-(t)$.
Therefore,
 $$
 \psi_t^*(g + h_+(t) + h_0(t))  =  g + \tilde h_+(t) + \tilde h_0(t)
 $$
 and, using covariance \eqref{defstress:covariance}, we obtain directly $U_Z(t) \sim \tilde U_Z(t)$.
 \vspace{0.2cm}\\
\noindent
\underline{General case:}
 The general case is now obtained by writing the diffeomorphism $\psi_t$ as a composition of three diffeomorphisms of the above types. This can always be done for $t \in [-T,T]$ and sufficiently small $T$. To see this choose relatively compact open sets $\calV_-\Subset\calU_-, \calV_0\Subset\calU_0, \calV_+\Subset\calU_+$
such that $K\subset \calV_-\cup \calV_0 \cup \calV_+ $. Now decompose $K = K_0 \cup K_+ \cup K_-$ into closed sets $K_0 \subset \calV_0,K_+ \subset \calV_+,K_- \subset \calV_-$ 
and choose smooth functions $(\chi_\bullet)_{\bullet \in \{0,+,-\}} \in C^\infty_0(M)$ with
$$
  \supp\chi_\bullet\subset\calU_\bullet,
  \quad
  \chi_\bullet=1
  \quad\text{in an open neighbourhood of }K_\bullet .
$$
For any compactly supported time-dependent vector field $X_t$, the flow
equation generates a two-parameter family of diffeomorphisms $\phi_{(t,s)}$.
We say that a one-parameter family $\psi_t$ is generated by $X_t$ if
$\psi_t=\phi_{(t,0)}$. After decreasing $T$, we may write $\psi_t$ in this
form with $X_t$ supported in $K$. Define $\psi_{-,t}$ to be the family
generated by $\chi_-X_t$. Since $\chi_-=1$ in an open neighbourhood of
$K_-$, and $K_-$ is compact, $\psi_{-,t}$ agrees with $\psi_t$ on a
neighbourhood of $K_-$ for $|t|\leq T$, after possibly decreasing $T$.
Therefore $\psi_t^{(1)}:=\psi_{-,t}^{-1}\circ\psi_t$
is the identity near $K_-$ and $
  \supp \psi_t^{(1)}\subset K\setminus K_-
  \subset \calV_0\cup \calV_+$.
We now repeat the same construction for $\psi_t^{(1)}$. After decreasing
$T$ again if necessary we think of $\psi_t^{(1)}$ as associated to the time-dependent vector-field $X_t^{(1)}$ and define the diffeomorphism
$\psi_{0,t}$ as defined by the time-dependent vector field $\chi_0 X_t^{(1)}$.
Then $\psi_{0,t}$ is supported in $\calU_0$, and it agrees with
$\psi_t^{(1)}$ on an open neighbourhood of $K_0$ for $T$ sufficiently small. Hence
$\psi_t^{(2)}:=\psi_{0,t}^{-1}\circ\psi_t^{(1)}$
is the identity near $K_0$. Consequently
$
  \supp \psi_t^{(2)}\subset K\setminus(K_-\cup K_0)
  \subset \calV_+ \subset \calU_+ .
$
Set $\psi_{+,t}:=\psi_t^{(2)}$.
Then $\psi_{+,t}$ is supported in $\calU_+$, and by construction
$$
  \psi_t
  =
  \psi_{-,t}\circ\psi_{0,t}\circ\psi_{+,t}.
$$
All three families depend smoothly on $t$, and are equal to the identity at $t=0$.
Applying the three special cases successively to these factors proves the claim.
 \end{proof} 

\begin{proposition}
 Let $U_Z(t)$ be families of unitary operators associated to the Killing field $Z$ on $\calU$.
 Then, for sufficiently small $T>0$ there exist families of unitary operators $V(t,s)$ and $W(t,s)$ with $s,t \in [-T,T]$, commuting with 
 any operator in $\calA_g(\calO)$, such that
 \begin{align}
   U_Z(t) V(t,s) U_Z(s) = W(t,s) U_Z(t+s)
 \end{align}
 In particular we have for any $A \in \calA_g(\calO)$ the group property $\alpha_s(\alpha_t(A)) = \alpha_{t+s}(A)$.
 \end{proposition}
\begin{proof}
 Throughout we keep in mind that
 by causality \eqref{defstress:causality}{\sl (1)}
 $$
  \calA_{g}(\calO) = \calA_{g+h_+(t)}(\calO) = \calA_{g+h_+(t)+h_0(t)}(\calO).
 $$
 Since 
 $$
  \phi_t^*(g + h_+(s) + h_0(s) + h_-(s))= g + h_+(t+s) + h_0(t+s) + h_-(t+s)
 $$
 we have
 $$
   g + h_+(t) + h_0(t) + h_-(t) + \phi_t^*(h_+(s)) + \phi_t^*(h_0(s)) +  \phi_t^*(h_-(s)) = g + h_+(t+s) + h_0(t+s) + h_-(t+s).
 $$
 Comparing supports we can see that for sufficiently small $s,t \in [-\frac{1}{2} T,\frac{1}{2} T]$ the differences
 $h_\pm(t+s) - \phi_t^*(h_\pm(s)) - h_\pm(t)$ are supported in $\calU_0$. We can therefore write 
\begin{align} \label{groupwithexzczxc}
   \phi_t^*(g+h_+(s) + h_0(s)) = g + h_+(t+s)+ h_-(t) + k_0(t,s), 
 \end{align}
 where $k_0(t,s)$ is supported in $\calU_0$ for $s,t \in [-\frac{1}{2} T,\frac{1}{2} T]$.
 Similarly 
 \begin{align} \label{groupwithminus}
   \phi_t^*(g+h_-(s) + h_0(s)) = g + h_-(t+s)+ h_+(t) + \tilde k_0(t,s)
 \end{align}
with $\tilde k_0(t,s)$ supported in $\calU_0$.
 
Causality implies for all $s \in [-T,T]$ that
 $$
   \mathbf{1} \sim S(g + h_+(s) + h_0(s) ,g)  S(g + h_-(s) + h_0(s) ,g+h_0(s)).
 $$

 By covariance \eqref{defstress:covariance} and causality \eqref{defstress:causality}{\sl (3)}
 \begin{align*}
  S(g + h_+(s) + h_0(s),g) = S(g +  h_+(s+t) + h_-(t) + k_0(t,s),g)\\ \sim
  S(g + h_-(t) + k_0(t,s) , g + k_0(t,s)) S(g +  h_+(s+t)+ k_0(t,s),g) .
 \end{align*}

 Thus,
 \begin{align*}
  S(g + h_-(t) +k_0(t,s),g+k_0(t,s))^* S(g + h_+(s) + h_0(s),g)   \sim S(g +  h_+(s+t)+ k_0(t,s),g).
 \end{align*}

Using \eqref{groupwithminus} with $s=-t$ we see that $\phi_{-t}^*(g+h_-(t) +k_0(t,s))$ is of the form $g+h_+(-t) + m_0(t,s)$ with $m_0(t,s)$ supported in
$\calU_0$ for sufficiently $t,s \in [0,T]$ and sufficiently small $T>0$. We have used here that there is a compact subset of $\calU_0$ that supports $h_0(s)$  and $k_0(t,s)$ for all $t,s \in [0,T]$ as long as $T$ is sufficiently small. 
Using again covariance \eqref{defstress:covariance}, we obtain
 \begin{align*}
  S(g &+ h_-(t) +k_0(t,s),g+k_0(t,s)) = S(g+h_+(-t) + m_0(t,s),g+k_0(t,s))\\ & \sim S(g+h_+(-t) + m_0(t,s),g+h_+(-t) + h_0(-t))S(g+h_+(-t) + h_0(-t),g)S(g,g+k_0(t,s)).
 \end{align*}
We now observe that $V_1(t,s) = S(g+h_+(-t) + m_0(t,s),g+h_+(-t) + h_0(-t))$ is affiliated to $\calA_{g+h_+(-t) + h_0(-t)}(\calU_0)$ and therefore commutes
with $\calA_{g+h_+(-t) + h_0(-t)}(\calO) = \calA_{g}(\calO)$. Note also that $V_2(t,s) = S(g,g+k_0(t,s))$ commutes with $\calA_{g}(\calO)$ for the same reason.
Similarly 
\begin{align*}
 S(g &+  h_+(s+t)+ k_0(t,s),g) \\&= S(g +  h_+(s+t)+ k_0(t,s), g +  h_+(s+t)+ h_0(t+s)) S(g +  h_+(s+t)+ h_0(t+s),g)
 \end{align*}
and again $V_3(t,s) = S(g +  h_+(s+t)+ k_0(t,s), g +  h_+(s+t)+ h_0(t+s))$ commutes with $\calA_g(\calO)$.
Summarizing
\begin{align*}
  V_2^*(t,s) & S(g + h_+(-t) + h_0(-t),g)^* V_1^*(t,s) S(g + h_+(s) + h_0(s),g) \\ & \sim V_3(t,s) S(g + h_+(t+s) + h_0(t+s),g) ,
\end{align*}
in other words
\begin{align} \label{sdfsdfdfgswk}
 V_2^*(t,s) U_Z(-t)^* V_1^*(t,s) U_Z(s) \sim V_3(t,s) U_Z(t+s) .
\end{align}
In case $s=0$ we obtain
$$
  U_Z(-t)^*  \sim V_2(t,0) V_3(t,0) U_Z(t) V_1(t,0).
$$
Re-inserting this into \eqref{sdfsdfdfgswk} gives
$$
 U_Z(t) V_1(t,0) V_1^*(t,s) U_Z(s) \sim  V^*_3(t,0)  V^*_2(t,0)  V_2(t,s) V_3(t,s)  U_Z(t+s).
$$
Setting
$$
 V(t,s)=V_1(t,0)V_1^*(t,s),\qquad
 W(t,s)=V_3^*(t,0)V_2^*(t,0)V_2(t,s)V_3(t,s),
$$
and absorbing the scalar phase into $W(t,s)$, we obtain the stated
equality after decreasing $T$ if necessary. Moreover,
$\overline{\calU_0}$ remains spacelike separated from
$\overline{\phi_t(\calO)}$ for $|t|\leq T$. Hence the correction
operators commute with both $\calA_g(\calO)$ and
$\calA_g(\phi_t(\calO))$, and the group property follows.
 \end{proof}
 
We can now state that locally the Killing flow along $Z$ can be implemented by unitaries. 

\begin{theorem} \label{li}
 In a quantum field theory with stress energy tensor
 we have with $\alpha_t$ defined above there exists a $T>0$ such that $\alpha_t(\calA_g(\calO)) = \calA_g(\phi_t(\calO))$ for $t \in [-T,T]$.
\end{theorem}
\begin{proof}
 By causality \eqref{defstress:causality}{\sl (2)},
 \begin{align*}
  \calA_{g+h_-(t)+h_0(t)}(\calO) &= S(g+h_-(t)+h_0(t),g) \calA_{g}(\calO)S^*(g+h_-(t)+h_0(t),g),\\
   \calA_{g+h_+(t)+h_0(t)}(\calO) &=  \calA_{g}(\calO).
 \end{align*}
 By covariance \eqref{defstress:covariance}
 \begin{align*}
  \calA_g(\phi_t(\calO)) &=  \calA_{g+h_+(t)+h_-(t)+h_0(t)}(\calO) = \calA_{g+h_-(t)+h_0(t)}(\calO) \\&= S(g+h_-(t)+h_0(t),g) \calA_g(\calO) S(g+h_-(t)+h_0(t),g)^*.
  \end{align*}
  
By Lemma \ref{elevenlemma} we have $W(t) S(g+h_-(t) + h_0(t),g) V(t) \sim  U_Z^*(t)$.
Since $\overline{\calU_0}$ is spacelike separated from $\overline{\calO}$ by continuity of the flow $\overline{\calU_0}$ is also spacelike separated from $\overline{\phi_t(\calO)}$ for sufficiently small $T>0$ and any $t \in [-T,T]$. Hence $V(t), W(t) \in \calA_g(\calU_0)$ commute with both $\calA_g(\calO)$ and $\calA_g(\phi_t(\calO))$.
Therefore,
  \begin{align*}
  \calA_g(\phi_t(\calO)) &= W(t) \calA_g(\phi_t(\calO)) W(t)^* \\&= W(t) S(g+h_-(t)+h_0(t),g) \calA_g(\calO)  S(g+h_-(t)+h_0(t),g)^*W(t)^*\\ &= W(t) S(g+h_-(t)+h_0(t),g) V(t)  \calA_g(\calO) V(t)^* S(g+h_-(t)+h_0(t),g)^* W(t)^*\\&=
 U_Z^*(t)  \calA_g(\calO) U_Z(t),
  \end{align*}
  which proves the claim.
 \end{proof}

\begin{rem}
The size of $T$ depends on the choice of cutoff modifying $Z$ to $\tilde Z$ and on the cover $\calW, \calU_\pm, \calU_0$. In favourable cases - for example when $Z$ is a global complete Killing field on a spatially compact spacetime (so that $\calU_0$ can be taken empty) - one can construct $U_Z(t)$ for arbitrarily large $|t|$. We do not pursue here the question of how to quantify the maximal $T$ in terms of the support of $\mathcal L_Z g$ and the geometry of $\calO$.
\end{rem}

\section{The example of the Klein-Gordon field} \label{KGSection}

The basic free fields, the Dirac field and the Klein-Gordon field, have been constructed algebraically by Dimock (\cite{MR0594301, MR0637032}). This construction provides a convenient starting point to discuss free fields on curved backgrounds, in particular when supplemented with the representation theory of CCR- and CAR-algebras. The construction of free fields can be understood as a two stage process. The first is the construction of the local field algebras as abstract $*$-algebras. This construction can be carried out for the massive and massless scalar field in any globally hyperbolic spacetime. It can also be carried out for the Dirac field in case a spin structure is specified (\cite{MR0637032}). In the second stage one looks for a representation of this abstract algebra as an algebra of operators on some Hilbert space. Such a representation always exists but usually requires a choice of reference state. Some form of selection criterion for such a choice of physically reasonable states is needed. In Minkowski spacetimes a criterion for the Minkowski vacuum is that of invariance under the Poincar\'e group and the fact that it leads to a positive energy representation. For curved spacetimes a wealth of reasons have been accumulated why the Hadamard condition provides a very good selection criterion from both a physics and mathematics point of view.
The Hadamard condition has been shown by Radzikowski (\cite{MR1400751}) to be the covariant microlocal counterpart of the spectrum condition. As proved by Verch (\cite{MR1266061}) different choices of Hadamard states lead to quasi-equivalent representations of the local algebras corresponding to regions with compact closure.

We will focus here on the Klein-Gordon field and will show that once a pure quasifree Hadamard state is chosen on a reference spacetime the theory admits a stress energy tensor.

Let $m\geq 0$ and denote by $P= \Box + m^2$ the Klein-Gordon operator. We have the retarded and advanced fundamental solutions $G_\pm$ that are regarded as continuous maps $C_0^\infty(M) \to C^\infty(M)$. The distributional kernels are then in $\mathcal{D}'(M \times M)$ and will be denoted by the same symbols.
The difference $G= G_+ - G_-$ then has antisymmetric distributional integral kernel, i.e. $G(f \otimes g) = - G(g \otimes f)$.
This coincides, up to a factor depending on the sign convention, with the Pauli-Jordan distribution and is sometimes also called the commutator distribution as it will describe the non-equal time commutator relation of the field.

We will denote by $\ker(P)$ the space of real-valued smooth solutions of the Klein-Gordon equation with spacelike compact support.
Given any Cauchy hypersurface $\Sigma$ the Cauchy data map $R_\Sigma: \ker(P) \to C^\infty_0(\Sigma,\R) \oplus C^\infty_0(\Sigma,\R)$ assigns to a solution its Cauchy data on $\Sigma$. This means
\begin{align}
 R_\Sigma u = (u |_{\Sigma}, (\partial_n u)|_\Sigma),
\end{align}
where $\partial_n$ is the future pointing unit normal vector field to $\Sigma$.
The map $R_\Sigma$ is an isomorphism, and for any pair $(\Sigma',\Sigma)$ of Cauchy hypersurfaces and the map $R_{\Sigma'} \circ (R_\Sigma)^{-1}: C^\infty_0(\Sigma,\R) \oplus C^\infty_0(\Sigma,\R) \to C^\infty_0(\Sigma',\R) \oplus C^\infty_0(\Sigma',\R)$ is a properly supported Fourier integral operator of order zero. We refer to \cite{MR1362544} for the construction of the solution operators as Fourier integral operators and note that the fact that the operators have proper support follows immediately from finite propagation speed for solutions of normally hyperbolic differential equations. The construction of this Fourier integral operator is also briefly sketched in Appendix \ref{Appendparam} where additional regularity under parameter dependence is established.

The restriction $R_\Sigma u$ is in fact defined on distributional solutions $u \in \mathcal{D}'(M)$ of the equation $P u =0$ and the corresponding map
$R_\Sigma : \mathrm{ker}\left( P : \mathcal{D}'(M) \to  \mathcal{D}'(M) \right) \to \mathcal{D}'(\Sigma) \oplus \mathcal{D}'(\Sigma)$ is invertible.
We will denote the inverse again by $R_\Sigma^{-1}$ when the choice of domain is not important.
Distributions and functions are identified on $\Sigma$ using the metric volume density. 
In what follows we will also consider situations in which the metric depends on an external parameter. Since the operators $P, R_\Sigma, R_\Sigma^{-1}$
as well as $\mathrm{ker}(P)$ depend on the metric we may also write $P_g, R_{g,\Sigma}, R_{g,\Sigma}^{-1}$ to denote this dependence explicitly.

We also have the restriction map $r_\Sigma: C^\infty_0(M) \to C^\infty_0(\Sigma) \oplus C^\infty_0(\Sigma)$ that is defined on all smooth functions rather than solutions. Of course $R_\Sigma = (r_\Sigma) |_{\mathrm{ker}(P)}$.
The adjoint $r_\Sigma^*$ therefore maps $\mathcal{D}'(\Sigma) \oplus \mathcal{D}'(\Sigma)$ to $\mathcal{D}'(M)$ in a continuous fashion. Similarly, $r_\Sigma^*: \mathcal{E}'(\Sigma) \oplus \mathcal{E}'(\Sigma) \to \mathcal{E}'(M)$.
The solution operator $R_\Sigma^{-1}: \mathcal{D}'(\Sigma) \oplus \mathcal{D}'(\Sigma) \to \mathcal{D}'(M)$
can also be expressed as
\begin{align} \label{solop}
 R_\Sigma^{-1} =  G \circ r_\Sigma^* \circ \left( \begin{matrix} 0 & 1\\ -1 &0 \end{matrix} \right).
\end{align}

For each Cauchy surface $\Sigma$ the Cauchy data space $C^\infty_0(\Sigma,\R) \oplus C^\infty_0(\Sigma,\R)$ carries a symplectic structure $\sigma_\Sigma$ given by

$$
 \sigma_\Sigma(\left( \begin{matrix} h \\  \dot h \end{matrix} \right),\left( \begin{matrix} f \\  \dot f \end{matrix} \right)) = \int_\Sigma \left(  f \dot h - h \dot f  \right) \der \mathrm{Vol}_\Sigma = \langle    \left( \begin{matrix} 0 & 1\\ -1 &0 \end{matrix}  \right)\left( \begin{matrix} h \\  \dot h \end{matrix} \right),  \left( \begin{matrix} f \\  \dot f \end{matrix} \right) \rangle_{L^2(\Sigma) \oplus L^2(\Sigma)}.
$$
The map $r^\dagger_\Sigma = r_\Sigma^* \left( \begin{matrix} 0 & 1\\ -1 &0 \end{matrix} \right)$ is the adjoint of the map
$r_\Sigma$ if one uses the symplectic form for the dual pairing on the space of initial data.
It follows from Green's identities that
the symplectic structure $\sigma$ induced on $\ker(P)$ by $R_\Sigma$ is independent of $\Sigma$ and is also characterised by
$$
 \sigma(u_1,u_2) = G(f_1 \otimes f_2) 
$$
if $u_1 = G f_1, u_2 = G f_2$. Since the kernel of the map $G$ is exactly $P C^\infty_0(M)$ this identifies the symplectic space $\ker(P)$ of solutions with the quotient $C_0^\infty(M,\R) / (P C_0^\infty(M,\R) )$ equipped with the symplectic form $G$.

The field algebra is now the abstract unital $*$-algebra generated by symbols $\Phi(f), f \in C^\infty_0(M)$ and the following relations
\begin{align}
 f &\mapsto \Phi(f)  \textrm{ is complex linear.}\\
 \Phi(f) \Phi(h) - \Phi(h) \Phi(f) &= - \rmi G(f \otimes h) \mathbf{1},\\
 \Phi(P f) &=0,\\
 \Phi(f)^* &= \Phi(\overline{f}).
\end{align}

By means of the map $G: C^\infty_0(M,\R) \to \mathrm{ker}(P)$ this is identified with the abstract CCR algebra of the symplectic vector space $(\mathrm{ker}(P), \sigma)$. For convenience we use here the non-exponentiated form of the CCR-algebra, which is not a $C^*$-algebra but merely a $*$-algebra, whose $*$-representations will be by unbounded operators $\Phi(f)$. We will be interested here in Fock space representations where the field operators are essentially self-adjoint and can then generate the exponentiated version of the CCR-algebra by $B(f) = e^{\rmi \Phi(f)}$. 
For an open set $\calO \subset M$ one defines the local von Neumann algebra $\calA(\calO)$ of the Klein-Gordon field as the von Neumann completion of the algebra generated by the $B(f)$ with $\supp f \subset \calO$. By von Neumann's bicommutant theorem this can also be defined as a commutant $$\calA(\calO) = (\{\Phi(f) \mid \supp f \subset \calO\}')'$$ of $\{\Phi(f) \mid \supp f \subset \calO \}'$, where
the commutant $\{\phi(f) \mid \supp f \subset \calO \}'$ is defined as the space of bounded operators on Fock space that commute with all the self-adjoint operators $\Phi(f)$ with $\supp f \subset \calO$.
Since the fields satisfy the Klein-Gordon equation it follows that the algebras automatically satisfy the local time-slice axiom. Namely, if $C \subset M$ is an achronal set and $\calO$ any open set containing the closure of $C$, then $\calA(\calO)$ contains $\calA(D(C))$.

\subsection{Pure quasifree states and Fock representations} \label{Fockspace}
A standard way to construct an irreducible representation of this CCR-algebra is by choosing a complex structure $J$ on a suitable completion of  $\mathrm{ker}(P)$ that allows to write the symplectic form as the imaginary part of a complex inner product. One way to phrase this is in terms of pure one-particle structures.
This consists of a continuous injective map
 $\kappa: \mathrm{ker}(P) \to H_\kappa$ into a real Hilbert space $H_\kappa$
with a compatible complex structure $J: H_\kappa \to H_\kappa$ such that

\begin{itemize}
 \item $\kappa$ has dense range,
 \item $\sigma(f,g) = -\langle J \kappa(f), \kappa(g) \rangle$.
\end{itemize}

Recall that $J$ is called compatible with the Hilbert space inner product $\langle\cdot, \cdot \rangle$ if 
$J$ is a skew-adjoint isometry $J = -J^* = -J^{-1}$ (sometimes this is referred to as a Kähler structure).
We can identify the $H_\kappa$ with the completion of $\ker P$ with respect to the real inner product induced by $\kappa$ and we will therefore in the following think of $\ker P$ as a subset of $H_\kappa$.
The complex structure allows us to view $H_\kappa$ as a complex Hilbert space with the same norm. 
The complex inner product $\langle\cdot,\cdot \rangle_\C$ is given by
\begin{align}
 \langle f,g \rangle_\C = \langle f,g \rangle + \rmi \langle J f,g \rangle.
\end{align}
This complex Hilbert space is the so-called {\sl one-particle} Hilbert space. 
A concrete way to identify the complex Hilbert space as a complex subspace of the complexification 
$H_\kappa \otimes_\R \C$ is to use the map $v \mapsto \mathrm{pr}_{+\rmi}(v \otimes 1)$, where
$\mathrm{pr}_{+\rmi}$ is the projection onto the $+\rmi$ eigenspace $H_J$ of $J: H_\kappa \otimes_\R \C \to H_\kappa \otimes_\R \C$ under the splitting $H_\kappa \otimes_\R \C = H_J \oplus \overline{H_J}$.
Since $\mathrm{pr}_{\pm\rmi} = \frac{1}{2} \left( 1 \mp \rmi J \right)$ the real subspace $H_\kappa \otimes_\R 1$
of $H_J \oplus \overline{H_J}$ is precisely the set of elements of the form $v \oplus \overline{v}$.
Note that the complex isomorphism $H_\kappa \to H_J, \quad v \mapsto \mathrm{pr}_{+\rmi}(v \otimes 1)$ is not norm preserving but instead $\|  \mathrm{pr}_{+\rmi}(v \otimes 1) \| = \frac{1}{\sqrt{2}} \| v\|$.

Putting everything together we have a real linear map $\mathrm{ker} P \to H_J, v \mapsto \mathrm{pr}_{+\rmi}(v \otimes 1)$. This of course can be extended complex linearly to a complex linear map $p: \mathrm{ker} P \otimes_\R \C \to H_J$.
Hence, the map $p \circ G$ is an $H_J$-valued distribution on $M$. Similarly $f \mapsto \overline{(\mathbf{1}-p) \circ G f}$ defines a conjugate linear map to $H_J$.

The representation of the field algebra is on the symmetric (bosonic) Fock space
$\mathcal{F}(H_J)$ which is the Hilbert space completion of  $\bigoplus_{k=0}^\infty \bigotimes_S^k H_J$.
The algebraic direct sum $\bigoplus_{k=0}^\infty \hat \bigotimes_S^k H_J$ of the completed symmetric tensor products is called the finite particle subspace and it will be denoted by $\mathcal{F}_\fp(H_J)$. Obviously the finite particle subspace is a dense subspace in the symmetric Fock space.
It will also be useful to consider the subspace $\mathcal{F}_\sfp(H_J)$ which is the algebraic direct sum of the symmetric tensor products without completion, i.e. this subspace consists of finite linear combinations of simple tensor products of vectors in $H_J$.

Given $v \in H_J$ one has the standard creation and annihilation operators $a^*(v)$ and $a(v)$ (see \cite{MR0493420}*{X.7}), which are unbounded operators defined on the finite particle subspace satisfying the canonical commutation relations
\begin{align}
 [a(v), a^*(w)] = \langle v, w \rangle \mathbf{1}, \quad [a(v), a(w)] = [a^*(v), a^*(w)] = 0,
\end{align}
The operator
\begin{align}
  \Phi_J(f) =a(\overline{(\mathbf{1}-p) G f}) + a^*(p(G f))= a(p (G \overline{f})) + a^*(p(G f)), \quad f \in C^\infty_0(M,\C). 
 \end{align}
is defined on the finite particle subspace. Note the absence of the usual factor of $\frac{1}{\sqrt{2}}$ as in \cite{MR0493420} which would appear if an isometric map from $H_\kappa \to H_J$ had been chosen instead of $p$.

The map $f \mapsto \Phi_J(f)$ is clearly a complex linear map from the test function space $C^\infty_0(M,\C)$ to the linear operators $\mathcal{F}_\fp(H_J) \to \mathcal{F}_\fp(H_J)$. If $f \in C^\infty_0(M,\R)$ this simplifies to 
$\Phi_J(f) =a(p \circ G (f)) + a^*(p \circ G (f)).$ One then checks 
\begin{align}
 [\Phi_J(f_1), \Phi_J(f_2)] = -\rmi G(f_1 \otimes f_2) \mathbf{1}
\end{align}
on $\mathcal{F}_\fp(H_J)$
and these operators therefore define a representation of the CCR relation. This procedure is referred to as Segal quantisation. 

The two-point function is the vacuum expectation value $\langle \Omega, \Phi(\cdot) \Phi(\cdot) \Omega \rangle$ of $\Phi(\cdot) \Phi(\cdot)$ and is then given by
\begin{align}
 \omega_J(f \otimes h) = \langle p(G(\overline{f})), p(G(h)) \rangle.
\end{align}
Given two pure quasifree states constructed from the one particle structures as above a theorem by Shale (\cite{MR0137504}) states that a necessary and sufficient condition for the two representations to be unitarily equivalent is that the inner products $\langle \cdot,\cdot \rangle_{\kappa_j}, j=1,2$ induce the same topology on
$\ker(P)$ and that the corresponding positive operator implementing the equivalence is a Hilbert-Schmidt perturbation of the identity.  If this condition is satisfied there exists a unitary map $U_{\kappa_2,\kappa_1} : \mathcal{F}(H_{J}) \to \mathcal{F}(H_{J})$ so that
\begin{align}
 \Phi_{\kappa_2}(f) =U_{\kappa_2,\kappa_1}    \Phi_{\kappa_1}(f) U_{\kappa_2,\kappa_1}^*.
\end{align}
This unitary operator is unique modulo a phase factor in $U(1)$ but can be fixed by requiring that $\langle \Omega_{\kappa_2},U_{\kappa_2,\kappa_1} \Omega_{\kappa_1} \rangle \geq 0$. 
A more detailed description and the explicit form of the implementer can be found in Appendix \ref{Shalendix}.

\subsection{Quasifree pure Hadamard states}
A pure quasifree state as above is called a quasifree pure Hadamard state if the wavefront set of the two-point function is constrained to a subset of $V_- \times V_+$, where $V_\pm$ are the closed future/past lightcones in the cotangent bundle, i.e. the set of  future/past-directed causal covectors. The two-point distribution $\omega_J$ of a Hadamard state is known to be a Fourier integral operator whose canonical relation is the null-geodesic relation.
We summarise the argument here for the sake of completeness. As shown by Duistermaat and Hörmander (\cite{MR0388464}) the operator $G$ is a Fourier integral operator on $M$.  Its distributional kernel satisfies
\begin{align}
 \mathrm{WF}'(G) =\{(x,\xi,x',\xi') \in N \times N \mid (x,\xi) \sim (x',\xi') \},
\end{align}
where $N$ is the set of non-zero lightlike co-vectors and $(x,\xi) \sim (x',\xi')$ means that $(x,\xi)$ and $(x',\xi')$ are in the same orbit of the geodesic flow. Since $\mathrm{WF}'(G)$ has two components distinguished by time-orientation there is a microlocal splitting
$G = S_+ + S_-$, where 
\begin{align}
 \mathrm{WF}'(S_\pm) =\{(x,\xi,x',\xi') \in N_\mp \times N_\mp \mid (x,\xi) \sim (x',\xi') \},
\end{align}
which is unique modulo smooth kernels.
Since any Hadamard state achieves precisely such a splitting it follows that $\rmi S_+$ coincides with a Hadamard two-point function modulo smooth kernels. In particular this shows that $\omega_J$ is the integral kernel of a Fourier integral operator with canonical relation
$\{(x,\xi,x',\xi') \in N_- \times N_- \mid (x,\xi) \sim (x',\xi') \}$.
This restriction to the Cauchy hypersurface 
then defines a pseudodifferential operator. More precisely, under the identification $R_{\Sigma}$ of $\ker P$ with the Cauchy data space $C^\infty_0(\Sigma,\R) \oplus C^\infty_0(\Sigma,\R)$
the inner product on $C^\infty_0(\Sigma,\R) \oplus C^\infty_0(\Sigma,\R)$ is characterised by
\begin{align}
 \langle (f,\dot f), (f,\dot f) \rangle_\kappa = \langle (f,\dot f) , A (f,\dot f) \rangle_{L^2(\Sigma) \oplus L^2(\Sigma)}. 
\end{align}
where 
\begin{align}
 A = \left(  \begin{matrix} A_{00} & A_{01} \\  A_{10} & A_{11}  \end{matrix} \right)
\end{align} is a two-by-two matrix of pseudodifferential operators. Indeed,  the operator $A$ is obtained by restriction of the Hadamard two-point function to the Cauchy hypersurface as the real part of
\begin{align}
  S_J = 2 \left( \begin{matrix} 0 & 1\\ -1 & 0 \end{matrix} \right) R_\Sigma \circ \omega_J \circ R_\Sigma^*  \left( \begin{matrix} 0 & 1\\ -1 & 0 \end{matrix} \right).
\end{align}
The restriction $R_\Sigma = \left( \begin{matrix} r_0 \\ r_1 \end{matrix} \right)$, if composed with operators that are microlocally supported away from the normal bundle of $\Sigma$, behaves like a Fourier integral operator, with $r_0$ of order $\frac{1}{4}$ and $r_1$ of order $\frac{5}{4}$ (see \cite{MR4178908}*{Lemma 8.3} and \cite{MR1362544}*{Section 5.1, p. 113} for a precise statement).
Counting orders we obtain 
\begin{align}
  S_J =  \left(  \begin{matrix} S_{00} & S_{01} \\  S_{10} & S_{11}  \end{matrix} \right),
\end{align}
where $S_{00}$ is a pseudodifferential operator of order $1$, $S_{11}$ is a pseudodifferential operator of order $-1$, and $S_{01}, S_{10}$ are of order $0$.
The principal symbols can be directly computed from the principal symbol of  $\omega_J$, which was computed by Duistermaat and Hörmander (\cite{MR0388464}*{Theorem 6.6.1}). Its restriction to the diagonal equals $\sqrt{2 \pi}$, up to a factor that depends on the choice of half-density on the Lagrangian submanifold $\Lambda$.
For the principal symbol $\sigma_{S_{00}}(\xi)$ of $S_{00}$ at $\xi \in T^*\Sigma$  one obtains $(g^{-1}_\Sigma(\xi,\xi))^{\frac{1}{2}}$, for  
$\sigma_{S_{11}}(\xi)$ one obtains  $(g^{-1}_\Sigma(\xi,\xi))^{-\frac{1}{2}}$, where $g_\Sigma = -g |_{\Sigma}$ is the induced Riemannian metric on $\Sigma$ and $g^{-1}_\Sigma$ its dual.
In particular $S_{00}$ and $S_{11}$ are elliptic.
Similarly, one obtains $\rmi, -\rmi$ for the principal symbols of $S_{10}$ and $S_{01}$ respectively. At the principal symbol level the latter do not contribute to $A_{01}$ and $A_{10}$.
Therefore, $A_{00}$ is an elliptic pseudodifferential operator of order $1$, $A_{11}$ is an elliptic pseudodifferential operator of order $-1$, and $A_{01}, A_{10}$ are of order $-1$. We refer to \cite{MR4178908}*{Sections 6,7,8} where the restrictions and the principal symbols are computed explicitly using the Fourier integral operator representation.

Since the construction of the Feynman parametrix for the wave equation can be carried out in the polyhomogeneous symbol class, the pseudodifferential operator $A$ will also have a polyhomogeneous symbol. The full symbol of $A$, modulo smoothing operators, is locally determined as it is obtained by solving a transport equation locally, followed by restriction to the hypersurface. 
This can also be seen directly since the full symbols of these pseudodifferential operators are determined algebraically, which can also be used to construct Hadamard states directly  by means of pseudodifferential calculus (\cite{MR1421547, MR3148100}).

\subsubsection{Frequency splitting as an example}
Before we continue it is instructive to see the example of the construction of the ground state and in the case of ultra-static spacetimes in which this construction corresponds to a frequency splitting procedure. We will thus assume, only in this subsection, that $M = \R_t \times \Sigma_x$ and the metric has the form
$\der t^2 - h$, where $h$ is a time-independent Riemannian metric on $\Sigma$. We will assume also for simplicity that $m>0$ and that $\Sigma$ is compact.
The above linear algebra construction with the complex structure formalises the process of frequency splitting in great generality. 
 Then every element in the complexification $\mathrm{ker}(P) \otimes_\R \C$ of $\mathrm{ker}(P)$ has a generalized Fourier expansion of the form
\begin{align}
 u(t,x) = \sum_{j=0}^\infty a_j \Phi_j(x) e^{\rmi \omega_j t} + b_j \Phi_j(x) e^{-\rmi \omega_j t},
\end{align}
where $\Phi_j$ are an orthonormal basis of eigenfunctions of the Laplace operator on $\Sigma$ with eigenvalues $\mu_j$ and $\omega_j = \sqrt{\mu_j + m^2}$. A standard choice of complex structure is then the map defined by linear extension of the rule
\begin{align}
 J \Phi_j(x) e^{\pm \rmi \omega_j t} = \mp \rmi \Phi_j(x) e^{\pm \rmi \omega_j t}.
\end{align}
The splitting $\ker(P)\otimes_\R \C = \mathcal{W} \oplus \overline{\mathcal{W}}$ is then the splitting into positive and negative frequency subspaces.
A different choice of sign here is possible and leads to a different notion of positivity for frequency and the opposite sign convention in the wavefront set condition for Hadamard states.

One checks that $J$ commutes with conjugation and thus leaves the space of real-valued solutions invariant.
Namely, $$J \left(\Phi_j(x) \cos(\omega_j t)\right) =  \Phi_j \sin(\omega_j t), \quad J \left( \Phi_j(x) \sin(\omega_j t) \right) =  -\Phi_j \cos(\omega_j t).$$
Therefore, the inner product satisfies
\begin{align}
 \langle \Phi_j(x) \cos(\omega_j t) , \Phi_k(x) \cos(\omega_k t) \rangle &= \delta_{jk} \omega_j,\\
  \langle \Phi_j(x) \omega_j^{-1}\sin(\omega_j t) , \Phi_k(x) \omega_k^{-1}\sin(\omega_k t) \rangle &= \delta_{jk} \omega_j^{-1},\\
  \langle \Phi_j(x) \cos(\omega_j t) , \Phi_k(x) \omega_k^{-1}\sin(\omega_k t) \rangle &= 0.
\end{align}
On the level of Cauchy data we can write the map $J$ as 
\begin{align}
 J =  \left( \begin{matrix} 0 & -(-\Delta+m^2)^{-\frac{1}{2}} \\   (-\Delta+m^2)^{\frac{1}{2}}  & 0\end{matrix}\right).
\end{align}
By spectral calculus this map also makes sense if $\Sigma$ is non-compact, as long as it is metrically complete, and $m>0$. 
In this case the operator $A$ is given by
\begin{align}
 A= \left( \begin{matrix} (-\Delta+m^2)^\frac{1}{2} & 0\\ 0 & (-\Delta+m^2)^{-\frac{1}{2}} \end{matrix} \right)
\end{align}
and the projection operator is
\begin{align}
 p = \left( \begin{matrix} \frac{1}{2} & \frac{\rmi}{2} (-\Delta+m^2)^{-\frac{1}{2}} \\ -\frac{\rmi}{2} (-\Delta+m^2)^{\frac{1}{2}} & \frac{1}{2} \end{matrix}\right)
\end{align}
The completion of the Cauchy data space with respect to the real inner product equals $H^{\frac{1}{2}}(\Sigma,\R) \oplus H^{-\frac{1}{2}}(\Sigma,\R)$ if $\Sigma$ is compact or, more generally, of bounded geometry. The complex Hilbert space $H_J$ is therefore the range of the projection, i.e. the complex subspace of 
$H^{\frac{1}{2}}(\Sigma,\C) \oplus H^{-\frac{1}{2}}(\Sigma,\C)$ consisting of vectors of the form
$(f, -\rmi (-\Delta+m^2)^{\frac{1}{2}} f), f \in H^{\frac{1}{2}}(\Sigma,\C)$.

This space is isomorphic to $H^{\frac{1}{2}}(\Sigma,\C)$ and can thus be identified with the Hilbert space $L^2(\Sigma,\C)$
via the map 
$$
 f \mapsto (-\Delta+m^2)^{\frac{1}{4}} f.
$$
Therefore the corresponding irreducible one-particle Hilbert space structure is isomorphic to the one obtained from the map
\begin{align}
 \kappa: (f,g) \mapsto (-\Delta+m^2)^{\frac{1}{4}} f + \rmi (-\Delta+m^2)^{-\frac{1}{4}} g
\end{align}
from the Cauchy data space $C^\infty_0(\Sigma,\R) \oplus C^\infty_0(\Sigma,\R)$ to $L^2(\Sigma,\C)$, if the latter space is regarded as a real Hilbert space with complex structure given by multiplication by $\rmi$.

It can be checked that the corresponding state is a quasifree pure Hadamard state.
One can use this construction and employ a deformation argument by Fulling-Narcowich and Wald (\cite{MR0641893}) to show the existence of a quasifree and pure Hadamard state for any globally hyperbolic spacetime.

\subsection{Parameter dependent Hadamard states}

In the following we will assume that $I$ is a compact smooth manifold (with or without boundary) and the Lorentzian metric $g_s$ smoothly depends on a parameter $s \in I$ within $\mathcal{M}$ so that
the metrics $g_s$ all coincide outside a compact subset $K \subset M$. As shown in Section \ref{Appendparam} the associated family of Pauli-Jordan distributions, regarded as a distribution on $I \times M \times M$ satisfies
\begin{align}
 \mathrm{WF}'(G) &\subset \{(s,\tau_s(x',\xi',x,\xi),x',\xi',x,\xi) \in T^*I \times N \times N \mid (x,\xi) \sim_s (x',\xi') \} \nonumber \\ &= \Lambda \subset T^*(I \times M \times M).
\end{align}
where the notations $N, (x,\xi) \sim_s (x',\xi')$ and the function $\tau_s(x,\xi,x',\xi')$ are defined below. The set 
$N$ is the set of non-zero lightlike co-vectors and $(x,\xi) \sim_s (x',\xi')$ means that $(x,\xi^\#)$ and $(x',(\xi')^\#)$ are in the same orbit of the geodesic flow with respect to the metric $g_s$. This means there exists a null-geodesic $\gamma$ connecting the points $x$ and $x'$ with tangent vectors $\dot \gamma$ equal to $\xi^\#$ at $x$ and equal to $(\xi')^\#$ at $x'$. Then one defines the $T_s^*I$-valued function
$$
 \tau_s(x,\xi,x',\xi') = -\frac{1}{2} \int (\partial_s g_s)(\dot \gamma(t), \dot \gamma(t)) \der t.
$$
It is shown in Section \ref{Appendparam} that this is indeed a Lagrangian submanifold of $T^*I \times T^*M \times T^*M$ equipped with a twisted symplectic form, and we refer to this section for details.

Note here that this statement is local and sufficiently small compactly supported changes of the metric satisfy the assumptions of Section \ref{Appendparam}.
The precise statement is
\begin{align}
 G \in \mathrm{I}^{-\frac{3}{2}-\frac{\mathrm{dim}(I)}{4}}(I \times M \times M, \Lambda').
\end{align}

There again is a microlocal splitting
$G = S_+ + S_-$, where 
\begin{align}
 \mathrm{WF}'(S_\pm) \subset \Lambda_\pm,
\end{align}
which is unique modulo smooth kernels so that
\begin{align}
 S_\pm \in \mathrm{I}^{-\frac{3}{2}-\frac{\mathrm{dim}(I)}{4}}(I \times M \times M, \Lambda_\pm').
\end{align}
A family of Hadamard states that equals $\rmi S_+$ modulo smooth kernels will be called a parameter-dependent Hadamard state $\omega_s$.
If $\Sigma$ is a Cauchy hypersurface that does not intersect the compact set $K \subset M$ then 
the restriction of $\omega_s$ to $\Sigma$ then defines a smooth family of pseudodifferential operators in the sense that the full symbol in each coordinate chart depends smoothly on $s$ with symbolic estimates uniform in $s$. This follows immediately from the explicit description of the Fourier integral operators $S_\pm$ in Appendix \ref{Appendparam}.
The main example we have in mind here is a family of Hadamard states that has been constructed near a Cauchy surface $\Sigma_0$ that does not intersect $K$ and that is extended to 
$M$ by solving the wave equation on $M$ with respect to the metric $g_s$. As shown in Appendix \ref{Appendparam} this gives a smooth family of Hadamard states in the above sense.
The previous analysis of Hadamard states remains valid with $A_s,p_s,J_s$ being families of pseudodifferential operators that depend smoothly on $s$. 
For the parameter-dependent Hadamard states constructed in this way, the
singular part is independent of $s$ near $\Sigma$. Hence
$A_{s'}-A_s$, $p_{s'}-p_s$, and $J_{s'}-J_s$ are smooth kernels depending
smoothly on the parameters.

\subsection{Continuity properties of quasifree pure Hadamard states}

We record some observations about Hadamard states and their regularity properties. In the remainder of this subsection we assume that we are given a one-particle structure $\kappa$ associated to a quasifree pure Hadamard state. For notational simplicity we identify $H_\kappa$ with the completion of $C^\infty_0(\Sigma,\R) \oplus C^\infty_0(\Sigma,\R)$ with respect to the 
 norm $\| \cdot \|_\kappa$ and can therefore suppress $\kappa$ in the notations. As before let $J: H_\kappa \to H_\kappa$ be the associated complex structure and we denote the complex linear map extending $J$ to $H_\kappa \otimes_\R \C \to H_\kappa \otimes_\R \C$ by the same letter.
We use the projection $\mathrm{pr}_{+\rmi}: H_\kappa \otimes_\R \C \to H_J$ onto the $+\rmi$ eigenspace
of $J$, i.e. $J \otimes 1 = \rmi \left( 2\, \mathrm{pr}_{+\rmi} - 1 \right)$ and the map
\begin{align}
 p: C^\infty_0(\Sigma,\R) \oplus C^\infty_0(\Sigma,\R) \to H_J, \quad p = \mathrm{pr}_{+\rmi} |_{C^\infty_0(\Sigma,\R) \oplus C^\infty_0(\Sigma,\R)}.
\end{align}

The following Lemma is essentially the statement of \cite{MR1464689}*{Prop. 3.5} (see also \cite{MR1266061}). We give a self-contained presentation and proof since we will need the explicit description later.
\begin{lemma}
Assume $\kappa$ is a one particle structure associated to a quasifree pure Hadamard state.
 For a fixed (smooth and spacelike) Cauchy hypersurface $\Sigma \subset M$
 let $A$ and $H_\kappa$ be as above. Let $H_{\kappa}$ be the completion of $C^\infty_0(\Sigma,\R) \oplus C^\infty_0(\Sigma,\R)$ with respect to the 
 norm $\| \cdot \|_{\kappa}$. Then the inclusion map $C^\infty_0(\Sigma) \oplus C^\infty_0(\Sigma)$ to $H_{\kappa}$ extends continuously to give injective maps 
 \begin{align}
  H^{\frac{1}{2}}_\compp(\Sigma,\R) \oplus H^{-\frac{1}{2}}_\compp(\Sigma,\R) \hookrightarrow H_{\kappa}  \hookrightarrow H^{\frac{1}{2}}_\loc(\Sigma,\R) \oplus H^{-\frac{1}{2}}_\loc(\Sigma,\R).
\end{align}
 In particular we can understand the space $H_J$ as a subspace of the space 
$H^{\frac{1}{2}}_\loc(\Sigma) \oplus H^{-\frac{1}{2}}_\loc(\Sigma)$.
\end{lemma}

\begin{proof}
 Because of general mapping properties of pseudodifferential operators $A$ defines a continuous map
 $$
  H^{\ell}_\compp(\Sigma,\R) \oplus H^{\ell-1}_\compp(\Sigma,\R) \to H^{\ell-1}_\loc(\Sigma,\R) \oplus H^{\ell}_\loc(\Sigma,\R)
 $$
 for all $\ell \in \R$. Then, for $\ell \geq 0$ this allows to define the norm $\| \cdot \|_{\kappa}$ on $H^{\ell+\frac{1}{2}}_\compp(\Sigma,\R) \oplus H^{\ell-\frac{1}{2}}_\compp(\Sigma,\R)$ and the norm is continuous in the respective topology. 
 This is indeed a norm because of the hypoellipticity of $A$: In case $\langle f, A f \rangle =0, f \in H^{\ell+\frac{1}{2}}_\compp(\Sigma,\R) \oplus H^{\ell-\frac{1}{2}}_\compp(\Sigma,\R) $ it follows from the Cauchy-Schwarz inequality that $A f =0$ in the weak sense. 
 If we pick any properly supported elliptic pseudodifferential operator $q$ of order $2$ then the operator
 $$
  \left(  \begin{matrix} 1 & 0 \\  0 & q \end{matrix} \right) A
 $$
 is an elliptic pseudodifferential operator of order $1$ as it has lower triangular principal symbol.
 By elliptic regularity $f \in C^\infty_0(\Sigma,\R) \oplus C^\infty_0(\Sigma,\R)$ and therefore $f=0$.
 It follows that the closure of
  $C^\infty_0(\Sigma,\R) \oplus C^\infty_0(\Sigma,\R)$ is the same as the closure of $H^{\frac{1}{2}}_\compp(\Sigma,\R) \oplus H^{-\frac{1}{2}}_\compp(\Sigma,\R)$. 
 Any element in $H_{\kappa}$ defines a linear functional by $f \mapsto \langle J f, \cdot \rangle =-\sigma(f,\cdot)$
 which is then continuous on $H^{\frac{1}{2}}_\compp(\Sigma,\R) \oplus H^{-\frac{1}{2}}_\compp(\Sigma,\R)$.
 Since the dual of $H^{\ell}_\compp(\Sigma,\R)$ with respect to the $L^2$-pairing is $H^{-\ell}_\loc(\Sigma,\R)$  this shows, given the explicit form of $\sigma$, that we have an injection
 \begin{align*}
  H_{\kappa}  \hookrightarrow H^{\frac{1}{2}}_\loc(\Sigma,\R) \oplus H^{-\frac{1}{2}}_\loc(\Sigma,\R),
\end{align*}
where the element $v \in H_{\kappa}$ is mapped to the distribution
$
 \langle J v, \left(  \begin{matrix} 0 & -1 \\  1 & 0 \end{matrix} \right) \cdot \rangle_\kappa.
$
\end{proof}

\begin{rem}
The conclusion of the lemma holds for any inner product induced by a matrix of elliptic pseudodifferential operators of the form
$$
 A = \left(  \begin{matrix} A_{00} & A_{01} \\  A_{10} & A_{11}  \end{matrix} \right)
$$ 
if $A_{00}$ is an elliptic pseudodifferential operator of order $1$, $A_{11}$ is an elliptic pseudodifferential operator of order $-1$, and $A_{01}, A_{10}$ are of order $-1$. Indeed, if $H_A$ denotes the closure of $C^\infty_0(\Sigma) \oplus C^\infty_0(\Sigma)$ in the norm $\|f\|_A = \langle f, A f \rangle^\frac12$ then we have an injection
$H^{\frac{1}{2}}_\compp(\Sigma,\R) \oplus H^{-\frac{1}{2}}_\compp(\Sigma,\R) \to H_A$ by the same argument as in the proof above.
For any $f \in H_A$ the functional
$\langle A f, \cdot \rangle$ is continuous as a functional on $H^{\frac{1}{2}}_\compp(\Sigma,\R) \oplus H^{-\frac{1}{2}}_\compp(\Sigma,\R)$ and therefore has a unique representative in $H^{-\frac{1}{2}}_\loc(\Sigma,\R) \oplus H^{\frac{1}{2}}_\loc(\Sigma,\R)$.
Since then $A f \in H^{-\frac{1}{2}}_\loc(\Sigma,\R) \oplus H^{\frac{1}{2}}_\loc(\Sigma,\R)$ 
and consequently
$$
  \left(  \begin{matrix} 1 & 0 \\  0 & q \end{matrix} \right) A f \in H^{-\frac{1}{2}}_\loc(\Sigma,\R) \oplus H^{-\frac{1}{2}}_\loc(\Sigma,\R)
$$
the continuous inclusion of $H_A$
in $H^{\frac{1}{2}}_\loc(\Sigma,\R) \oplus H^{-\frac{1}{2}}_\loc(\Sigma,\R)$ is a direct consequence of the local elliptic regularity estimates.
The statement above therefore holds for any Hadamard state for the Klein-Gordon field on the space of real-valued functions.
\end{rem}

As a consequence elements in $H_{\kappa}$ are distributions, and therefore, via, $R_\Sigma$, correspond to distributional solutions $u$ of $P u =0$.
Another consequence is that any continuous map
$C^\infty_0(\Sigma,\R) \oplus C^\infty_0(\Sigma,\R)  \to H_{\kappa}$
has a unique distributional kernel.

For later purposes we will also define a scale of dense subspaces in $H_{\kappa}$ and $H_J$ as follows.
\begin{definition}
 For any $\ell \geq 0$ we define
 $$
  H_{\kappa}^\ell = \{ f + J g \mid f,g \in H^{\frac{1}{2}+\ell}_\compp(\Sigma,\R) \oplus H^{-\frac{1}{2}+\ell}_\compp(\Sigma,\R) \}.
 $$
 and its image $H_J^\ell$ in $H_J$ given by
 $$
  H_J^\ell = \mathrm{pr_{+\rmi}} \left( H^{\frac{1}{2}+\ell}_\compp(\Sigma) \oplus H^{-\frac{1}{2}+\ell}_\compp(\Sigma) \right) \subset H^{\frac{1}{2}+\ell}_\loc(\Sigma) \oplus H^{-\frac{1}{2}+\ell}_\loc(\Sigma).
 $$
 We also write 
 $$
  H_{\kappa}^\infty = \{ f + J g \mid f,g \in C^\infty_0(\Sigma,\R) \oplus C^\infty_0(\Sigma,\R) \}.
 $$
 and its image $H_J^\infty$ in $H_J$ given by
 $$
  H_J^\infty = \mathrm{pr_{+\rmi}} \left( C^\infty_0(\Sigma) \oplus C^\infty_0(\Sigma) \right) \subset   C^\infty(\Sigma) \oplus C^\infty(\Sigma).
 $$
\end{definition}

The notation does not imply that $H_J^0 = H_J$ or $H_\kappa^0 = H_\kappa$, but each of the above subspaces is dense.
Of course the map $p$ maps into $H_J^\infty$.

In the following assume that $W : C^\infty_0(\Sigma) \oplus C^\infty_0(\Sigma) \to C^\infty_0(\Sigma) \oplus C^\infty_0(\Sigma)$
is a Fourier integral operator with properly supported kernel. We assume furthermore that $W$ is a finite sum of Fourier integral operators whose canonical relation is the graph of an invertible canonical map.
Taking into account the different Sobolev weightings on the initial data space we call
$$
 W = \left( \begin{matrix} W_{11} & W_{12}\\ W_{21} & W_{22} \end{matrix}\right)
$$
of order $m$ if $W_{11},W_{22}$ are of order $m$, $W_{12}$ is of order $m-1$, and $W_{21}$ is of order $m+1$.
If $W$ has order $m \geq 0$ then, by the mapping properties of Fourier integral operators (\cite{MR0388464}*{Theorem 4.4.4}), it
extends continuously as maps
\begin{align}
W &: H^{\ell+\frac{1}{2}+m}_\compp(\Sigma) \oplus H^{\ell-\frac{1}{2}+m}_\compp(\Sigma) \to H^{\ell+\frac{1}{2}}_\compp(\Sigma) \oplus H^{\ell-\frac{1}{2}}_\compp(\Sigma)\\
W &: H^{\ell+\frac{1}{2}+m}_\loc(\Sigma) \oplus H^{\ell-\frac{1}{2}+m}_\loc(\Sigma) \to H^{\ell+\frac{1}{2}}_\loc(\Sigma) \oplus H^{\ell-\frac{1}{2}}_\loc(\Sigma) 
\end{align}
which will be denoted by the same letter, mildly abusing notations.

\begin{lemma} \label{ucontlem}
Assume $W, W^{-1}$ are, as above, properly supported zero order Fourier integral operators that are inverses to one another. Suppose that $W, W^{-1}$ are finite sums of Fourier integral operators whose canonical relations are the graph of an invertible canonical map. Assume further that $W - \mathbf{1}$ has compactly supported integral kernel. Then 
\begin{enumerate}
 \item also $W^{-1} - \mathbf{1}$ has compactly supported integral kernel.
 \item $W$ and $W^{-1}$ continuously map $H_{\kappa}$ to itself.
 \item $W$ and $W^{-1}$ continuously map $H^\ell_{\kappa}$ to itself for any $\ell \in [0,\infty]$.
 \end{enumerate}
\end{lemma}
\begin{proof}
 The first claim is immediate from $W^{-1} - \mathbf{1} =- W^{-1}(W-\mathbf{1}) = -(W-\mathbf{1})  W^{-1}$ and the fact that $W^{-1}$ is properly supported.
 The second and third  claim follow from the fact that any zero order Fourier integral operator whose kernel has compact support will give rise to a continuous map
 $$
  H^{\frac{1}{2}+\ell}_\loc(\Sigma) \oplus H^{-\frac{1}{2}+\ell}_\loc(\Sigma) \to H^{\frac{1}{2}+\ell}_\compp(\Sigma) \oplus H^{-\frac{1}{2}+\ell}_\compp(\Sigma).
 $$
\end{proof}

In much the same way one proves the following.
\begin{lemma} \label{diffulem}
 Assume that $W$ is a compactly supported Fourier integral operator of order $m \geq 0$. Suppose that $W$ is a finite sum of Fourier integral operators whose canonical relation is the graph of an invertible canonical map. 
 Then
 $W$ maps  $H_\kappa^{\ell+m} \to H^{\frac{1}{2}+\ell}_\compp(\Sigma) \oplus H^{-\frac{1}{2}+\ell}_\compp(\Sigma)$ continuously.
\end{lemma}

\subsection{Implementation of metric changes}

Assume $g_s$ is a smooth family of metrics in $\mathcal{M}$ depending on a parameter $s$ in a compact manifold $I$.
That is we assume there exists a compact subset $K \subset M$  and a smooth family of symmetric two-tensors $h_s$ with support in $I \times K$
such that $g_s = g_0 + h_s$, where $g_0$ is a fixed background metric.
We then have two families of symplectic maps $C_\pm(s) : \ker(P) \to \ker(P_s)$ defined by
\begin{align}
 C_\pm(s) = (R_{g_s,\Sigma_\pm})^{-1} \circ R_{g_0,\Sigma_\pm},
\end{align}
where $\Sigma_\pm$ is any Cauchy surface to the future/past of $K$. 
The maps $C_\pm(s)$ do not depend on the chosen Cauchy hypersurfaces. It should also be noted here that the existence of these maps is implied by the fact that the spacetimes $(M,g_s)$ and $(M,g_0)$ are identical to the future and to the past of the metric perturbation and the restriction to a Cauchy surface is merely a vehicle to facilitate the identification of solutions. The family of scattering maps is defined by $\Lambda(s) = C_-(s)^{-1} \circ C_+(s) : \ker P \to \ker P$.
We use the Cauchy data map $R_{\Sigma_-}$ to identify $\ker P$ with the space of Cauchy data $C^\infty_0(\Sigma_-,\R) \oplus C^\infty_0(\Sigma_-,\R)$.
We can also look at the Cauchy data evolution maps $V(g_s) : C^\infty_0(\Sigma_-,\R) \oplus C^\infty_0(\Sigma_-,\R) \to C^\infty_0(\Sigma_+,\R) \oplus C^\infty_0(\Sigma_+,\R)$ given by $R_{g_s,\Sigma_+} R_{g_s,\Sigma_-}^{-1}$.
One easily computes
\begin{align}
  W_s := R_{\Sigma_-} \Lambda(s) R_{\Sigma_-}^{-1} = V(g_s)^{-1} V(g_0)
\end{align}
This is a family of linear symplectic maps with each $W_s$ being a properly supported zero order Fourier integral operator. Upon adding the parameter we have from Theorem \ref{fiorepra}  that
\begin{align}
 W_s \in I^{-\frac{\mathrm{dim}(I)}{4}}(I \times \Sigma \times \Sigma,  \Lambda_\Sigma'),
\end{align}
with $W_s^{-1} - \mathbf{1}$ having compactly supported kernel (see Appendix \ref{Appendparam}).
Here
$\Lambda_\Sigma$
is obtained by pulling back $\Lambda \subset T^*(I \times M \times M)$ to $T^*(I \times \Sigma \times \Sigma)$.
It follows from Lemma \ref{diffulem} that $W_s-\mathbf{1}$ is a $C^k$-family of maps from $H_\kappa^{\ell + k} \to H^{\frac{1}{2}+\ell}_\compp(\Sigma,\R) \oplus H^{-\frac{1}{2}+\ell}_\compp(\Sigma,\R)$ for any $\ell \geq 0$.

Differentiating in $s$ therefore changes the order. For a fixed quasifree pure Hadamard state with symplectic structure $\sigma$ we  
also define the family of symplectic structures $J_s = W_s^{-1} J W_s$. 
This also gives rise to a family of projections $p_s$ and a parameter dependent quasifree pure Hadamard state.
For each $s \in I$ the map $W_s$ maps $\mathcal{H}_\kappa$ continuously to itself, however $W_s$ is not unitary.
Instead $W^*_s W_s$ is of the form $\mathrm{1} + r_s$, where $r_s$ is a smooth family of smooth kernels.

\subsection{Continuity properties of parameter dependent quasifree pure Hadamard states}

The operators $J$, and $p$ respectively, achieve a microlocal splitting of solutions of the equation $Pu=0$ into distributional solutions with wavefront sets in the forward and backward light cone.

\begin{proposition} \label{wvrestr}
 The $H_J$-valued distributions on $M$ defined by
 $p \circ R_\Sigma \circ G$ and the $\overline{H_J}$-valued distribution $(1-p) \circ R_\Sigma \circ G$ satisfy the relations
 \begin{align}
  \mathrm{WF}\left( p \circ R_\Sigma \circ G \right) &\subset  N_{+},\\
  \mathrm{WF}\left( (\mathbf{1}-p) \circ R_\Sigma \circ G \right) &\subset N_{-}.
 \end{align}
  \end{proposition}

 \begin{proof}
 This is the statement of \cite{MR1936535}*{Prop. 6.1} and its complex conjugate.
 \end{proof}

In the following we will construct pseudodifferential operators $Q_\pm$ on $\Sigma$ that facilitate the microlocal splitting on the level of Cauchy data.
The construction of $Q_\pm$ is local near $\Sigma$ and will not depend on $s$.
We therefore suppress $s$ momentarily and replace $M$ by a globally hyperbolic open neighborhood $M'$ of the Cauchy surface $\Sigma$ so that $\Sigma$ is now a Cauchy surface in $M' \subset M$.
Let $\eta \in C^\infty(M')$ be a fixed smooth function such that $\mathrm{supp}(\eta)$ is past compact and $\mathrm{supp}(1-\eta)$ is future compact. This $\eta=1$ in the far future and $\eta=0$ in the far past. Such a function can easily be constructed utilising a global time function. We assume that the derivative $\eta'$ is supported near the Cauchy surface $\Sigma$ and its support is future and past compact. The operators $P, R_\Sigma^{-1},G$ on $M'$ are independent of $s$ and they arise by restriction of the corresponding operators on $M$.
It is easy to see that the map 
\begin{align}
 P(\eta R_\Sigma^{-1}): C^\infty_0(\Sigma) \oplus C^\infty_0(\Sigma) \to C^\infty_0(M')
 \end{align} 
 is a right-inverse of the operator 
 \begin{align}
  R_\Sigma \circ G: C^\infty_0(M') \to C^\infty_0(\Sigma) \oplus C^\infty_0(\Sigma).
 \end{align} 
We now fix properly supported pseudodifferential operators  $Q_\pm$ on $M'$ with microsupport away from $N_\mp$. 
In our application we are interested in the case when either $Q_+ + Q_- = \mathbf{1}$ or $Q_+ + Q_- = \chi$ for a given smooth cut-off function $\chi$. 
Such microlocal partitions can be constructed from the corresponding partition in phase space and a quantisation map that maps to properly supported operators. 
We then define the operators
\begin{align} \label{microprojectors}
  q_\pm : C^\infty_0(\Sigma) \oplus  C^\infty_0(\Sigma) \to C^\infty_0(\Sigma) \oplus C^\infty_0(\Sigma), \quad q_\pm = 
  R_\Sigma \circ G \circ Q_\pm \circ P(\eta R_\Sigma^{-1}).
\end{align}
It follows that $q_\pm$ are properly supported pseudodifferential operators. In case $Q_+ + Q_- = \mathbf{1}$
one has $q_+ + q_- = \mathbf{1}$. If $K' \Subset \Sigma$ is a fixed compact subset then we can find a compactly supported cutoff function $\chi \in C^\infty_0(M)$ such that $Q_+ + Q_- = \chi$ and $(q_+ + q_-) f = f$ for any $f$ supported in $K'$.

\begin{lemma} \label{twentyonelemma}
 The maps $p \circ q_+$ and $(\mathbf{1}-p) \circ q_-$ extend continuously to maps
 $$
  H^\ell_\compp(\Sigma) \oplus H^{\ell-1}_\compp(\Sigma) \to H_\kappa \otimes_\R \C.
 $$ 
 for any $\ell \in \R$.
\end{lemma}

\begin{proof}
Note that
\begin{align}
 p q_+ &= p \circ R_\Sigma \circ G \circ Q_+ \circ P(\eta R_\Sigma^{-1}),\\
 (\mathbf{1}-p) q_- &= (\mathbf{1}-p) \circ R_\Sigma \circ G \circ Q_- \circ P(\eta R_\Sigma^{-1}).
\end{align}
 The formal transpose operators $Q_\pm^t$ have their microsupport away from $N_\pm$.
 Then Prop. \ref{wvrestr} shows that the maps $Q_+^t (p \circ R_\Sigma \circ G) = p \circ R_\Sigma \circ G \circ Q_+$ and $(\mathbf{1}-p) \circ R_\Sigma \circ G \circ Q_-$
 extend to $H^\ell_\compp(\Sigma) \oplus H^{\ell-1}_\compp(\Sigma)$. The Lemma now follows as $P(\eta R_\Sigma^{-1})$ is a continuous map from 
 $H^\ell_\compp(\Sigma) \oplus H^{\ell-1}_\compp(\Sigma)$ to $H^{\ell-2}_\compp(M)$. Here we use that
 $R_\Sigma^{-1}$ is a continuous map from $H^\ell_\compp(\Sigma) \oplus H^{\ell-1}_\compp(\Sigma)$ to
 $H^{\ell}_\loc(M)$.
\end{proof}

\begin{lemma} \label{twentytwolemma}
The maps  $q_+ \circ  p$ and $q_- \circ (\mathbf{1}-p)$ extend continuously to maps
 $$
 \calH_\kappa \otimes_\R \C \to C^\infty(\Sigma) \oplus C^\infty(\Sigma).
 $$ 
\end{lemma}
\begin{proof}
We show first that the solution operator $R_\Sigma^{-1}$ to the Cauchy problem maps $H_J$ continuously to the subset $H^\frac12_-(M)$ of $H^\frac12_\loc(M)$ consisting of distributions with the additional property that their wavefront set is contained in set  $N_-$ of past-directed null covectors.
We equip this space with the weakest locally convex topology so that a $v_j$ converges to $v$ in $H^\frac12_-(M)$ if and only if $v_j \to v$
in $H^\frac12_\loc(M)$ and $A v_j \to A v$ in $C^\infty$ for all compactly supported pseudodifferential operators $A$ with microsupport disjoint from $N_-$.
If $R_\Sigma^{-1}$ maps $H_J$  continuously into this space we can conclude by the definition of $q_+$ that map $q_+ p$ is continuous and to $C^\infty(\Sigma) \oplus C^\infty(\Sigma)$ as claimed.
A similar argument applies to $q_- \circ (\mathbf{1}-p)$. Namely, taking the complex conjugate $\overline{H}_J$ maps continuously to $H^\frac12_+(M)$, which is defined similarly with $N_-$ replaced by $N_+$.

Now let $f \in C^\infty_0(M)$ and $v \in H_J$. One computes 
\begin{align} \label{barbelqudsdf}
(R_\Sigma^{-1} v, f)_{L^2(M)} &= (G P \eta R_\Sigma^{-1} v, f)_{L^2(M)} = G(f, P \eta R_\Sigma^{-1} v)  \nonumber \\ &= \sigma_\C(R_\Sigma G f, R_\Sigma G P \eta R_\Sigma^{-1} v) =
\sigma_\C(R_\Sigma G f,v),
\end{align}
where $\sigma_\C$ is $\sigma$ extended bilinearly to the space of complex valued functions and it is understood as a pairing 
between $H^\frac{1}{2}_\compp(\Sigma) \oplus H^{-\frac{1}{2}}_\compp(\Sigma)$ and $H^\frac{1}{2}_\loc(\Sigma) \oplus H^{-\frac{1}{2}}_\loc(\Sigma)$.
It can also be understood as a continuous bilinear form on $H_J \oplus \overline{H_J}$ and we do not distinguish this notationally.
The form $\sigma_\C$ on $H_\kappa \otimes \C = H_J \oplus \overline{H_J}$ can be expressed in terms of the inner product on $H_J$ and is
$$
 \sigma_\C(v,w) = \frac{\rmi}{2} \left( \langle \overline v, w \rangle - \langle \overline w, v \rangle \right).
$$
In particular it follows that $\sigma_\C(w,v) = 0$ if $v,w \in H_J$. Hence Equ. \eqref{barbelqudsdf} becomes
\begin{align} 
 (R_\Sigma^{-1} v, f)_{L^2(M)} = \sigma_\C((\mathbf{1}-p)R_\Sigma G f, v) 
\end{align}

Assume now $A$ is a compactly supported pseudodifferential operator with microsupport disjoint from $N_-$.  Then the formal transpose $A^\mathtt{t}$ and the complex conjugate $\overline{A}$ have their microsupports disjoint from $N_+$. The operator $A$ applied to the distribution $R_\Sigma^{-1} v$ is defined by the linear map
\begin{align}\label{asfdasdfsdgse5}
 f \mapsto (R_\Sigma^{-1} v, A^\mathtt{t} f)_{L^2(M)} =  \sigma_\C((\mathbf{1}-p)R_\Sigma G A^\texttt{t} f, v).
 \end{align}

By Prop. \ref{wvrestr} the $\overline{H_J}$-valued distribution $(\mathbf{1} -p) R_\Sigma G$ has its wavefront set in $N_-$ and therefore $A (\mathbf{1} -p) R_\Sigma G$ regarded as an  $\overline{H_J}$-valued distribution is a smooth $\overline{H_J}$-valued function.
 It follows from \eqref{asfdasdfsdgse5} that map $v \mapsto A R_\Sigma^{-1} v$ is continuous from $H_J$ to $C^\infty(M)$.
Since $R_\Sigma^{-1}$ continuously maps $H^{\frac{1}{2}}_\loc(\Sigma) \oplus H^{-\frac{1}{2}}_\loc(\Sigma)$ to $H^\frac12_\loc(M)$ the statement is proved.
\end{proof}

\subsection{Implementation in Fock space and differentiability} \label{implsecdiff}

 We will now show that the scattering map $\Lambda(g',g) = C_-(g',g)^{-1} \circ C_+(g',g) : \ker P \to \ker P$ for any two metrics $g',g \in \mathcal{M}$ always satisfies Shale's criterion with respect to any pure quasifree Hadamard state and therefore can be implemented by a unitary in the Fock space.
 We will need a parameterised version of this and therefore assume as before that the metric $g'=g_s$ depends smoothly on a parameter $s \in I$ in a compact parameter manifold.
 As before $g_s-g$ will have support in the compact set $I \times K$, where $K \subset M$ is compact. We also assume that the background metric $g$ equals the metric for a distinguished point $o \in I$, i.e. $g = g_o$.
Given a quasifree pure Hadamard state this fixes an inclusion $\ker(P) \subset H_\kappa$ and a complex structure $J$. 
 The family of scattering maps $\Lambda(g_s, g)$ then gives rise to a family of complex structures $J_s$ and real inner products on $H_\kappa$ with the resulting symplectic structure being independent of $s$.

The following theorem shows implementability of the family of scattering maps $\Lambda(g_s, g)$ on the Fock space $\mathcal{F}(H_J)$. As before we denote by $\mathcal{F}_\fp(H_J)$
the subspace of finite particle vectors and by $\mathcal{F}_{\sfp}(H_J^\infty) \subset \mathcal{F}_\fp(H_J)$ the subset of vectors in the algebraic direct sum $\oplus_{k=0}^\infty \otimes^k_S H_J^\infty$, where the tensor product is also the algebraic tensor product. Its elements are therefore finite linear combinations of simple tensor products of vectors in $H_J^\infty$.
Similarly, we denote by $\mathcal{F}_{\sfp}(H_J^\ell) \subset \mathcal{F}_\fp(H_J)$ the set $\oplus_{k=0}^\infty \otimes^k_S H_J^\ell$. These subspaces are all dense in the Hilbert space $\mathcal{F}(H_J)$.

\begin{theorem}\label{fptheorem}
 The family $\Lambda(g_s, g) J - J \Lambda(g_s, g)$ is a smooth family of Hilbert-Schmidt operators on $H_\kappa$. Hence, the scattering map $\Lambda(g_s, g)$ can be implemented for each $s \in I$ on the Fock space $\mathcal{F}(H_J)$ to define a family of unitary maps $U_s : \mathcal{F}(H_J) \to \mathcal{F}(H_J)$. Then for each vector $v \in \mathcal{F}_{\sfp}(H_J^\infty)$ the family $U_s v$ is a smooth function on $I$ if $U_s$ is chosen as the standard implementer.
\end{theorem}

\begin{proof}
 We will use  the map $R_{\Sigma_-}$ to identify the space $H_\kappa$ with a set of functions on $\Sigma$ as before:
$$
  H^\frac{1}{2}_\compp(\Sigma_-,\R) \oplus H^{-\frac{1}{2}}_\compp(\Sigma_-,\R) \hookrightarrow H_\kappa \hookrightarrow H^\frac{1}{2}_\loc(\Sigma_-,\R) \oplus H^{-\frac{1}{2}}_\loc(\Sigma_-,\R).
 $$
The inner product of the state is given by
\begin{align}
 \langle (f,\dot f), (f,\dot f) \rangle_\kappa =\langle (f,\dot f), A (f,\dot f) \rangle_{L^2(\Sigma_\pm) \oplus L^2(\Sigma_\pm)},
\end{align}
where $A$ is a pseudodifferential operator. Let $W(g_s,g) = R_{\Sigma_-} \Lambda(g_s,g) R_{\Sigma_-}^{-1} = V(g_s)^{-1} V(g)$.
We will now simply write $W_s$ for $W(g_s,g)$ in this proof.
We have already seen that the norm $\| W_s \cdot \|_\kappa$ is equivalent to the norm $\| \cdot \|_\kappa$ for each $s \in I$.
Indeed, the map $W_s$ is a properly supported zero order Fourier integral operator and, by finite propagation speed, $W_s-\mathbf{1}$ has compactly supported kernel.
Thus $W_s$ extends continuously to a map $H_\kappa \to H_\kappa$ with continuous dependence on $s$. This is, for the same reason, also true for the inverse $W^{-1}_s$. Hence, the topologies induced by the two inner products coincide with uniform bounds in the parameter $s \in I$.
To show that $W_s$ can be implemented to a family of unitary operators on Fock space we  show that
\begin{align}
 p W_s - W_s p  = p (W_s -\mathbf{1}) - (W_s-\mathbf{1}) p = p (W_s-\mathbf{1}) (\mathbf{1}-p) - (\mathbf{1}-p)(W_s-\mathbf{1}) p
\end{align}
is a trace-class operator $H_\kappa \otimes \C \to H_\kappa\otimes \C$. 
The kernel of $W_s-\mathbf{1}$, $W^{-1}_s-\mathbf{1}$ has support in some compact set $K' \times K', K' \Subset \Sigma$ independent of $s \in I$.
Using a microlocal splitting, as before, we define the compactly supported pseudodifferential operators $q_\pm$ so that $(q_+ + q_-) f =f$ for initial data $f$ supported in $K'$. 
Then we have
\begin{align}
 p (W_s-\mathbf{1}) (\mathbf{1}-p) =  p (q_+ + q_-)  (W_s-\mathbf{1}) (q_+ + q_-) (\mathbf{1}-p).
\end{align}
As a consequence of the propagation of singularities theorem the operators
$$
 q_+  (W_s-\mathbf{1})  q_-, \quad q_-  (W_s-\mathbf{1})  q_+
$$ 
have smooth integral kernels with compact support that are also smooth in $s$. They are therefore smooth families of continuous maps from
$H^\frac{1}{2}_\loc(\Sigma_-) \oplus H^{-\frac{1}{2}}_\loc(\Sigma_-)$ to
$H^\ell_\compp(K) \oplus H^\ell_\compp(K)$ for any $\ell$. It is hence a smooth family of nuclear maps, and therefore a smooth family of trace-class operators
$H_\kappa \otimes \C \to H_\kappa\otimes \C$.
To show that $p (W_s-\mathbf{1}) (\mathbf{1}-p)$ is a smooth family of trace-class operators it is therefore enough to prove that 
the operators $p q_+(W_s-\mathbf{1})  ,  (W_s-\mathbf{1}) q_-(\mathbf{1}-p)$ are smooth families of trace-class operators.
By Lemma \ref{twentyonelemma} the map $p q_+$ is continuous from $H^{\ell}(K) \oplus H^{\ell-1}(K)$ to $H_\kappa \otimes \C$ for any $\ell$.
We can think of $K$ as a subset of a compact Riemannian manifold $N$ and use the Laplace operator $\Delta_N$ on
that compact Riemannian manifold to write for any integer $\ell>0$
$$
 p q_+ = p q_+ (\Delta_N +\mathbf{1})^{\ell+k}  (\Delta_N +\mathbf{1})^{-\ell-k}.
$$
Since $p q_+ (\Delta_N +\mathbf{1})^{\ell+k}$ is continuous for any $\ell>0,k>0$ and $(\Delta_N +\mathbf{1})^{-\ell-k} (W_s-\mathbf{1}) $ is a $C^k$ family of trace-class operators
on $H^{\frac{1}{2}}_\compp(K) \oplus H^{-\frac{1}{2}}_\compp(K) $ for sufficiently large 
$\ell$ this shows that the composition $p q_+ (W_s-\mathbf{1}) $ is a smooth family of trace-class operators.
For the trace-class property of $(W_s-\mathbf{1})  q_-(\mathbf{1}-p)$ we use Lemma \ref{twentytwolemma} in a similar way.
Namely, $$(W_s-\mathbf{1})  q_-(\mathbf{1}-p) = (W_s-\mathbf{1})  (\Delta_N +\mathbf{1})^{-\ell-k}  (\Delta_N +\mathbf{1})^{\ell+k} q_-(\mathbf{1}-p)$$ where $(\Delta_N +\mathbf{1})^{\ell+k} q_{-}(\mathbf{1}-p)$ is bounded and 
$(W_s-\mathbf{1})  (\Delta_N +\mathbf{1})^{-\ell-k}$ is a $C^k$ family of trace-class operators for $\ell$ sufficiently large.
The statement for the operator $(\mathbf{1}-p) (W_s-\mathbf{1}) p$ follows in a similar manner.

Since we would also like to prove differentiability we will give more precise statements following the Appendix \ref{Shalendix}. Using the splitting $H_\kappa \otimes \C = H_J \oplus \overline{H_J}$ we have the decomposition
\begin{align}
W_s = \left( \begin{matrix} q_s & \overline{r}_s \\ r_s & \overline{q}_s \end{matrix} \right).
\end{align}
Here $q_s = p W_s p$ and $r_s = (\mathbf{1}-p) W_s p = (\mathbf{1}-p)(W_s-\mathbf{1}) p$.
We recall that $W_s-\mathbf{1}$ is a $C^k$ function of $s$ taking values in the operators from $H^{\frac{1}{2}+ \ell +k}_\loc(\Sigma) \oplus H^{-\frac{1}{2}+ \ell+k}_\loc(\Sigma)$
to $H^{\frac{1}{2}+ \ell}_\compp(\Sigma) \oplus H^{-\frac{1}{2}+ \ell}_\compp(\Sigma)$ for any $\ell \in \R$. The map $W_s$ is invertible and symplectic on $H_\kappa$ and therefore
satisfies $W_s J W_s^* J =  J W_s^* J W_s= -\mathbf{1}$. Thus, 
\begin{align} \label{grfnklasasp}
 W_s^{-1} = \left( \begin{matrix} q_s^* & -r_s^* \\ -\overline{r_s}^* & \overline{q}_s^* \end{matrix} \right).
\end{align}
Since the previous arguments also apply to the inverse $W_s^{-1}$ we have already shown that
\begin{itemize}
 \item[(i)] $\overline{r_s}^*$ and $r_s$ are smooth families of Hilbert-Schmidt operators on $H_J$,
 \item[(ii)]  $\overline{r_s}^*$ and $r_s$ are smooth families of maps $H_J$ to $\overline{H_J^\ell}$ for any $\ell \geq 0$.
 \item[(iii)]  $q_s^*$ and $q_s$ are $C^k$-functions of $s$ as maps $H_J^{\ell +k}$ to $H_J^{\ell}$ for any $\ell \geq 0$.
\end{itemize}
From \eqref{grfnklasasp} one obtains
\begin{align}
 q_s^* q_s - r_s^* r_s = \mathbf{1}
\end{align}
and therefore $q_s$ is invertible as a map from $H_J$ to $H_J$ for every $s$ with inverse given by
\begin{align}
 q_s^{-1} =  \left( \mathbf{1} + r_s^* r_s \right)^{-1} q^*_s= q^*_s -  \left( \mathbf{1} + r_s^* r_s \right)^{-1} r_s^* r_s q_s^*.
\end{align}
The family $\left( \mathbf{1} + r_s^* r_s \right)$ is a $C^k$ family of maps from $H_J^\ell$ to $H_J^\ell$. Since the inverse from $H_J \to H_J$
satisfies
\begin{align}
 \left( \mathbf{1} + r_s^* r_s \right)^{-1} = \mathbf{1} - r_s^* r_s \left( \mathbf{1} + r_s^* r_s \right)^{-1} 
\end{align}
it also continuously maps $H_J^\ell \to H_J^\ell$ for any $\ell \geq 0$. We conclude that
$q_s^{-1}$ is a $C^k$ family of maps from $H_J^{\ell +k}$  to $H_J^{\ell}$.
All together $q_s^{-1}$ is a $C^k$ family of maps from $H_J^{\ell +k}$  to $H_J^{\ell}$. The same argument applies to $(q_s^{-1})^*$.

As in Appendix \ref{Shalendix} we write
$$
 K_s = \overline{r_s q_s^{-1}}, \quad L_s = - q_s^{-1} \overline{r}_s.
$$

The operators $K_s, L_s$ are then smooth symmetric families of Hilbert-Schmidt operators.
For the canonical implementer we have 
\begin{align}
 U_s v = \left(\mathrm{det}(1 - K^*_s K_s) \right)^\frac{1}{4} e^{-\frac{1}{2} a^*(K_s)} \Gamma( (q_s^{-1})^*) e^{-\frac{1}{2} a(L_s)} v.
\end{align}
The family of vectors $e^{-\frac{1}{2} a(L_s)} v$ is a finite linear combination of $s$-independent vectors in $\mathcal{F}_{\sfp}(H_J)$ with smooth functions in $s$ as coefficients.
Thus, $\Gamma( (q_s^{-1})^*) e^{-\frac{1}{2} a(L_s)} v$ is a smooth function of $s$ taking values in $\mathcal{F}_{\sfp}(H_J)$. The statement of smoothness now follows from Prop. \ref{smoothytwo}.
\end{proof}

This theorem is in line with physics intuition that the total number of particle anti-particle pairs being created out of a quasifree pure Hadamard state by a compactly supported metric perturbation is finite. This translates to the unitary implementability as captured by Shale's criterion. 

\subsection{The Klein-Gordon field is a quantum field theory with stress energy tensor} \label{exkgex}

As before we assume that $\mathcal{M}$ is a moduli space of globally hyperbolic metrics on a fixed background spacetime $(M,g_0)$.
We assume that for the Klein-Gordon field for mass $m \geq 0$ a pure quasifree Hadamard state has been fixed.
Denote by $\calH = \mathcal{F}(H_J)$ the corresponding Hilbert space and let $\calA_{g_0}(\calO)$ be the local algebras acting on $\calH$.
Using past-identification we can also represent the local algebras for the Klein-Gordon field $\calA_{g}(\calO)$ on the same Hilbert space
$\calH$. 
This Hilbert space is the symmetric Fock space over the one particle Hilbert space $H_J$.
We can fix a Cauchy surface $\Sigma$ and the isomorphisms will allow us to think of 
$H_J$ as a subspace of $H^{\frac{1}{2}}_\loc(\Sigma) \oplus H^{-\frac{1}{2}}_\loc(\Sigma)$.
We define the dense subspace 
\begin{align}
 \calH^\infty_{g_0} = \mathcal{F}_{\sfp}(H_J^\infty)
\end{align}
consisting of finite linear combinations of simple vectors in the symmetric tensor product of $H_J^\infty$.

For each path $\gamma: [a,b] \to \mathcal{M}$ in $\mathcal{M}$ with $\gamma(a)=g$ and $\gamma(b)=g'$ we can now construct a scattering matrix $S(\gamma(t))$ as
\begin{align}
 S(\gamma) = S_c(g',g_0) S_c(g,g_0)^*,
\end{align}
where $S_c(g',g_0)$ and $S_c(g,g_0)$ are the canonical
implementers of $\Lambda(g',g_0)$ and $\Lambda(g,g_0)$,
respectively. We define
\begin{align}
 \calH^\infty_{g} =  S(\gamma) \calH^\infty_{g_0} 
\end{align}
where $\gamma$ is any path connecting $g_0$ to $g$. Since different choices of paths implement the same map $\Lambda(g',g)$ the equivalence class of
$S(\gamma)$ in $U(\calH)/U(1)$ is path independent. Hence, the definition of the space above does not depend on the choice of path. The holonomy is trivial and therefore central.

\begin{theorem} \label{KGExamplTh}
 The above defines a theory with stress energy tensor: the scattering matrix with the dense set of smooth vectors $\calH^\infty_g$ defines a unitary connection in the sense of Definition \ref{defconnection} that satisfies all the requirements of  Definition \ref{defstress}.
 \end{theorem}
 \begin{proof}
 We split the proof into several parts. We can rely on well-established properties of the field algebra of the Klein-Gordon field such as additivity of the net (\cite{MR0158666}*{Th. 1, (3), Equ. (3.12)}), Einstein causality, and the time-slice property. These can be inferred directly from properties of the CCR algebra and their representations as laid out in \cite{MR0158666} and the well known properties of the solutions of the Klein-Gordon equation.
 \\
 
\underline{Check that $S$ defines a unitary connection:}
 Since the canonical implementer $S_c(g',g_0)$ is unitary we have $S(g,g)=\mathbf{1}$ but also
 path additivity:
 $$
  S(g'',g) = S_c(g'',g_0) S_c(g,g_0)^* = S_c(g'',g_0) S_c(g',g_0)^*  S_c(g',g_0) S_c(g,g_0)^* = S(g'',g') S(g',g).
 $$
 We will now show strong continuity for smooth paths $\gamma(s)$. Smooth families of paths result in smooth families of symplectic maps
 of the form
 \begin{align}
  W_s = \left( \begin{matrix} q_s & \overline{r}_s \\ r_s & \overline{q}_s \end{matrix} \right).
 \end{align}
 As we have seen in the proof of Theorem \ref{fptheorem} the maps $q_s$ are strongly continuous in $s$ and the  $r_s$ are smooth families of Hilbert-Schmidt operators. The statement now follows from Prop. \ref{strongprop}.
 This shows properties (1) - (5) of Definition \ref{defconnection}. The first order condition (6) needs to be considered separately and is in fact a computation 
 of the derivative. Let therefore $I\subset\mathbb R^N$ be a neighbourhood of $0$, and $h: I \to \Gamma^\infty_0(S^2T^*M)$ be a smooth family of symmetric two-tensors with $h(0)=0$, compactly supported between the Cauchy surfaces $\Sigma_-$ and $\Sigma_+$.
We set
$$
  g_s=g+h(s),
  \qquad
  P_s=P_{g_s},
  \qquad
  W_s=W(g_s,g),
$$
and write
$$
  \dot h=
  \left.\frac{\der}{\der s}\right|_{s=0}h(s),
  \qquad
  \dot P=
  \left.\frac{\der}{\der s}\right|_{s=0}P_s .
$$
We first compute the derivative of the classical scattering map $W_s$ at
$s=0$.

Let $u_s$ be the solution of $P_su_s=0$
with fixed Cauchy data on the past Cauchy surface $\Sigma_-$. Differentiating
at $s=0$ gives
$$
  P_g\dot u+\dot P u=0,
  \qquad
  u=u_0.
$$
Since the Cauchy data on $\Sigma_-$ are fixed, $\dot u$ has vanishing Cauchy
data on $\Sigma_-$ and therefore
$$
  \dot u=-G_{g,\mathrm{ret}}\dot P u.
$$
It follows that the derivative of the forward Cauchy evolution $V(g_s)$ from
$\Sigma_-$ to $\Sigma_+$ is
$$
  \dot V R_{\Sigma_-}u
  =
  -R_{\Sigma_+}G_{g,\mathrm{ret}}\dot P u .
$$
Since
$$
  W_s=V(g_s)^{-1}V(g),
$$
we get
$$
  \dot W=-V(g)^{-1}\dot V.
$$
Thus
$$
  \dot W R_{\Sigma_-}u
  =
  V(g)^{-1}R_{\Sigma_+}G_{g,\mathrm{ret}}\dot P u .
$$
The homogeneous solution whose future Cauchy data agree with
$G_{g,\mathrm{ret}}\dot P u$ is $G_g\dot P u$, because the advanced term
vanishes to the future of $\supp(\dot P u)$. Hence $
  \dot W R_{\Sigma_-}u
  =
  R_{\Sigma_-}G_g\dot P u.$
In particular, for $u=G_gf$,
$$
  \dot W R_{\Sigma_-}G_gf
  =
  R_{\Sigma_-}G_g\dot P G_gf .
$$

We now differentiate the canonical implementer. As before write
$$
  W_s=
  \begin{pmatrix}
    q_s &   \overline{r_s}\\
   r_s & \overline{q_s}
  \end{pmatrix},
  \qquad
  q_0=\mathbf 1,
  \qquad
  r_0=0.
$$
Let
$$
  \dot q=
  \left.\frac{\der}{\der s}\right|_{s=0}q_s,
  \qquad
  \dot r=
  \left.\frac{\der}{\der s}\right|_{s=0}r_s .
$$
The canonical implementer is
$$
  U_s
  =
  \left(\det(1-K_s^*K_s)\right)^{1/4}
  e^{-\frac12a^*(K_s)}
  \Gamma((q_s^{-1})^*)
  e^{-\frac12a(L_s)},
$$
where
$$
  K_s=\overline{r_sq_s^{-1}},
  \qquad
  L_s=-q_s^{-1}\overline{r_s}.
$$
Since $q_0=\mathbf 1$ and $r_0=0$, we have
$$
  K_0=0,
  \qquad
  L_0=0,
$$
and therefore
$$
  \dot K=\overline{\dot r},
  \qquad
  \dot L=-\overline{\dot r}.
$$
Moreover, the determinant factor has vanishing first derivative at $s=0$,
because $K_s^*K_s$ is second order in $s$. Hence, on smooth simple finite particle subspace $\mathcal{F}_{\sfp}(H_J^\infty)$,
$$
  \left.\frac{\der}{\der s}\right|_{s=0}U_s
  =
  -\frac12a^*(\dot K)
  +
  \der\Gamma(-\dot q^*)
  -
  \frac12a(\dot L).
$$
Equivalently,
$$
  \left.\frac{\der}{\der s}\right|_{s=0}U_s
  =
  -\frac12a^*(\overline{\dot r})
  +
  \der\Gamma(-\dot q^*)
  +
  \frac12a(\overline{\dot r}).
$$

This expression depends linearly only on the first variation $\dot W$, and hence only
on the first variation $\dot h$ of the metric, again in a linear fashion. This shows the first order condition.

\underline{Covariance \eqref{defstress:covariance}:}\\ This follows directly from the functorial properties of the Klein-Gordon operator, namely for any diffeomorphism $\phi: M \to M$ we have
\begin{align}
 P_{\phi^* g}( \phi^* f) = \phi^*(P_g f).
\end{align}
In particular, if $\phi$ is supported between two Cauchy hypersurfaces the Cauchy evolution maps between these surfaces for the metrics $g$ and $\phi^* g$ are identical.

\underline{Causality \eqref{defstress:causality}:}\\
We start with \eqref{defstress:causality}{\sl (3)}.

We fix $t \in [a,b]$ and construct $\Sigma_1, \Sigma_2, \Sigma_3$
depending on $t$. The assumptions and Prop.~\ref{sepper} imply the
existence of a Cauchy hypersurface $\Sigma_2$ with respect to the
metric $g+h_2(t)$ separating $\supp h_1(t)$ and $\supp h_3(t)$. Since
$\Sigma_2$ is disjoint from $\supp h_1(t) \cup \supp h_3(t)$, the four
metrics $g+h_2(t)$, $g+h_1(t)+h_2(t)$, $g+h_2(t)+h_3(t)$, 
$g+h_1(t)+h_2(t)+h_3(t)$ agree near $\Sigma_2$, and as they agree outside compact sets disjoint
from $\Sigma_2$ we have that $\Sigma_2$ is a Cauchy hypersurface for all
four metrics.
The formula for \eqref{defstress:causality}{\sl (3)} will be verified now for the symplectic map $W$.
We let
\begin{align}
V_1(h)&: C^\infty_0(\Sigma_1) \oplus C^\infty_0(\Sigma_1) \to C^\infty_0(\Sigma_2) \oplus C^\infty_0(\Sigma_2),\\
V_2(h)&: C^\infty_0(\Sigma_2) \oplus C^\infty_0(\Sigma_2) \to C^\infty_0(\Sigma_3) \oplus C^\infty_0(\Sigma_3),\\
V(h)&: C^\infty_0(\Sigma_1) \oplus C^\infty_0(\Sigma_1) \to C^\infty_0(\Sigma_3) \oplus C^\infty_0(\Sigma_3),
\end{align}
be the Cauchy evolution maps with respect to the metric $g+h$. We then have, of course, the factorisation $V(h) = V_2(h) \circ V_1(h)$ but also
\begin{align}
 V(h_1+h_2+h_3) &= V_2(h_2+h_3) \circ V_1(h_1+h_2),\\
 V(h_1+h_2) &= V_2(h_2) \circ V_1(h_1+h_2),\\
 V(h_2+h_3) &= V_2(h_2+h_3) \circ V_1(h_2).
\end{align}
With $W(h) = V(0)^{-1} V(h)$ this shows
\begin{align}
 W(h_1+h_2 +h_3) = W(h_2+h_3) W(h_2)^{-1} W(h_1+h_2).
\end{align}
Applying the canonical implementers in Fock space shows the relation for the scattering matrices up to a phase factor.

It remains to show the other two properties {\sl \eqref{caus1}} and
{\sl \eqref{caus2}} of causality \eqref{defstress:causality}. 
In this proof we briefly denote the Pauli-Jordan distribution, 
the difference between retarded and advanced fundamental solutions, by $G_g$ to emphasise the dependence on the metric $g$.
We note that
the map $R_\Sigma\circ G_g$, restricted to test functions supported in
$\calO$, only depends on the metric in the causal region between $\calO$
and $\Sigma$. If $\Sigma_-$ is chosen in the far past of the support of
the metric perturbation, this region intersects the support of the perturbation
only if the perturbation intersects $J^-_g(\calO)$. This proves
{\sl \eqref{caus1}}.

To prove {\sl \eqref{caus2}}, we choose instead a Cauchy surface $\Sigma_+$
in the far future of the support of the perturbation. If
$$
  \supp(g'-g)\cap J^+_g(\calO)=\emptyset,
$$
then for every $f\in C^\infty_0(\calO)$ the future Cauchy data
$$
  R_{\Sigma_+}G_g(f)
$$
are the same for the metrics $g$ and $g'$. Let $V(g)$, respectively
$V(g')$, denote the Cauchy evolution map from $\Sigma_-$ to $\Sigma_+$ for the
metrics $g$, respectively $g'$, cf. Section \ref{implsecdiff}. Then
$$
  V(g)R_{\Sigma_-}G_g(f)
  =
  R_{\Sigma_+}G_g(f)
  =
  R_{\Sigma_+}G_{g'}(f)
  =
  V(g')R_{\Sigma_-}G_{g'}(f).
$$
Hence
$$
  R_{\Sigma_-}G_{g'}(f)
  =
  V(g')^{-1}V(g)R_{\Sigma_-}G_g(f).
$$
The scattering matrix $S(g',g)$ implements the symplectic map
$V(g')^{-1}V(g)$. Therefore
$$
  \calA_{g'}(\calO)
  =
  S(g',g)\calA_g(\calO)S(g',g)^*,
$$
which is precisely {\sl \eqref{caus2}}.

\underline{Locality \eqref{defstress:locality}:}\\
We first prove $\eqref{loc2}$ under the additional assumption that the support
of the metric perturbation is sufficiently small. Let $h$ be compactly
supported in $\calO$, and assume that $\supp h$ is contained in $D(\Sigma_0)$
for some relatively compact open subset $\Sigma_0$ of a Cauchy surface
$\Sigma$, with $\overline{\Sigma_0}\subset\calO$. Let
$$
  W:C^\infty_0(\Sigma)\oplus C^\infty_0(\Sigma)
  \longrightarrow
  C^\infty_0(\Sigma)\oplus C^\infty_0(\Sigma)
$$
be the symplectic map induced by the change from $g$ to $g+h$. By finite
propagation speed, $W-\mathbf 1$ is supported in $\Sigma_0$. Equivalently,
$W$ acts trivially on
$$
  C^\infty_0(\overline{\Sigma_0}^{\,c})
  \oplus
  C^\infty_0(\overline{\Sigma_0}^{\,c}).
$$
Hence the implementer $S(g+h,g)$ commutes with all Weyl operators $B(v)$ with
$$
  v\in
  C^\infty_0(\overline{\Sigma_0}^{\,c})
  \oplus
  C^\infty_0(\overline{\Sigma_0}^{\,c}).
$$
By Haag duality, see \cite{MR0158666}*{Th. 1, (5), Equ. (3.14)}, it follows
that $S(g+h,g)$ belongs to the von Neumann algebra generated by the Cauchy data
supported in some open neighbourhood $\tilde\Sigma_0$ of $\overline{\Sigma_0}$
with $\overline{\tilde\Sigma_0}\subset\calO$. By the local time-slice property,
these Cauchy data are generated by fields $Gf$ with $f\in C^\infty_0(\calO)$.
Thus
$$
  S(g+h,g)\in\calA_g(\calO).
$$
This proves $\eqref{loc2}$ for perturbations with sufficiently small support.

We now prove $\eqref{loc1}$. Let $g'=g+h$ with $h$ compactly supported in
$\calO$. By isotony and additivity of the net, it is enough to prove the
assertion for relatively compact open subsets of $\calO$. Thus we may assume
without loss of generality that $\overline{\calO}$ is compact. We also reduce
to small changes of the metric: choose a piecewise smooth path in
$\mathcal M$ from $g$ to $g'$, and subdivide it if necessary. It is enough to
prove the assertion for one sufficiently small step. By compactness of the
curve parameter space we may assume that $g+h$ lies in a sufficiently small
$C^1$-neighbourhood $\calN_0$ of $g$ in $\mathcal M$, so that all pointwise convex combinations
$g+\chi h$, $0\leq \chi \leq1$, are globally hyperbolic metrics in
$\mathcal M$, by Propositions \ref{prop-app-1} and \ref{prop-app-2}.

We next reduce to perturbations supported in a sufficiently thin slab. Let
$M=\mathbb R_t\times\Sigma$ be a smooth temporal splitting for the metrics under consideration. For $t\in\mathbb R$ and $\epsilon>0$, set
$$
  C_{t,\epsilon}
  =
  \calO\cap\bigcup_{\tau\in(t-\epsilon,t+\epsilon)}\Sigma_\tau .
$$
The sets $C_{t,\epsilon}$ cover $\supp h$. By compactness, we may choose a
finite subcover and a partition of unity $(\chi_j)$ subordinate to it. We
choose the partition so that $\chi_j\geq0$ and $\sum_j\chi_j=1$ on $\supp h$.
Then the partial sums $\sum_{j\leq k}\chi_j$ take values in $[0,1]$. If the
claim is proved for perturbations supported in such thin slabs, then applying
it successively to the partial sums
$$
  g_k=g+\sum_{j\leq k}\chi_jh
$$
gives
$$
  \calA_g(\calO)=\calA_{g+h}(\calO).
$$
It remains to prove the thin-slab case.

Let $C_{t,\epsilon}$ be one of the slabs defined above, and set
$$
  K_{t,\epsilon}=\supp h\cap C_{t,\epsilon}.
$$
After decreasing $\epsilon>0$, and after shrinking the neighbourhood
$\calN_0$ of $g$ in $\mathcal M$ if necessary, we may assume that for every
metric $\tilde g\in\calN_0$,
$$
  J^-_{\tilde g}(K_{t,\epsilon})\cap\Sigma_{t-2\epsilon}
  \subset
  \calO\cap\Sigma_{t-2\epsilon}.
$$
We choose a relatively compact open neighbourhood
$\Sigma_0\subset\calO\cap\Sigma_{t-2\epsilon}$ of
$$
  \bigcup_{\tilde g\in\calN_0}
  \left(J^-_{\tilde g}(K_{t,\epsilon})\cap\Sigma_{t-2\epsilon}\right).
$$

\begin{figure}[h!]
\centering
\begin{tikzpicture}[scale=0.85]

  \def\bubblepath{(0,1.1) ellipse (6.4 and 3.5)}

  \filldraw[thin,fill=gray!10,draw=gray!50,opacity=0.7] \bubblepath;

  \begin{scope}
    \clip \bubblepath;
    \filldraw[fill=gray!35,draw=gray!60,opacity=0.5]
       (-0.7,1.0) -- (-3.6,4.7) -- (3.6,4.7) -- (0.7,1.0) -- cycle;
    \filldraw[fill=gray!35,draw=gray!60,opacity=0.5]
       (-0.7,1.0) -- (-4.0,-3.5) -- (4.0,-3.5) -- (0.7,1.0) -- cycle;
  \end{scope}

  \fill[gray,opacity=0.18] (-6.5,0.7) rectangle (6.5,1.3);
  \draw[thin,gray!70,dashed] (-6.5,0.7) -- (6.5,0.7);
  \draw[thin,gray!70,dashed] (-6.5,1.3) -- (6.5,1.3);

  \filldraw[thin,fill=gray!55,draw=gray!70,opacity=0.6]
     (0,1.05) ellipse (0.95 and 0.75);

  \draw[dashed,thick,gray!70] (-2.6,0.0) -- (0,2.6) -- (2.6,0.0);

  \fill[gray!90,opacity=0.9] (0,1.0) ellipse (0.5 and 0.28);

  \draw[line width=1.6pt] (-2.6,0.0) -- (2.6,0.0);

  \node at (-5.5,3.4) {$\calO$};
  \node at (0,1.0) {$K_\chi$};
  \node[anchor=west] at (0.0,1.5) {$\calO_0$};
  \node at (-5.0,1.7) {$\calO_+$};
  \node at (5.0,1.7) {$\calO_+$};
  \node at (-5.0,-0.) {$\calO_-$};
  \node at (5.0,-.0) {$\calO_-$};
  \node[anchor=north] at (-1.5,-0.05) {$\Sigma_0$};
  \node[anchor=west] at (6.55,1.0) {$C_{t,\epsilon}$};

\end{tikzpicture}
\caption{Illustration of the thin-slab argument. The perturbation is supported in the compact set $K_\chi$ inside the slab $C_{t,\epsilon}$, and the two cones depict the causal future and past of $K_\chi$ with respect to $g$. The region $\calO_0\subset D^+(\Sigma_0)$ contains $K_\chi$. The region $\calO$ is covered by $\calO_-$, $\calO_0$, $\calO_+$.}\label{thinslab}
\end{figure}
We choose a relatively compact open subset $\calO_0\subset\calO$ such that
$$
  K_{t,\epsilon}\subset\calO_0 \quad \textrm{and} \quad
  \calO_0\subset D^+_{\tilde g}(\Sigma_0)
$$
for all metrics $\tilde g\in\calN_0$. After decreasing $\epsilon$ and shrinking
$\calN_0$ if necessary, we also choose $\calO_0$ so that
$$
  \calO\cap J^+_{\tilde g}(K_{t,\epsilon})
  \cap J^-_{\tilde g}(K_{t,\epsilon})
  \subset \calO_0
$$
for all metrics $\tilde g\in\calN_0$. 

For every cutoff function $\chi$ subordinate to $C_{t,\epsilon}$, we have
$$
  K_\chi:=\supp(\chi h)\subset K_{t,\epsilon}.
$$
Define
$$
  \calO_-
  =
  \calO\cap\left(M\setminus J^+_g(K_\chi)\right),
  \qquad
  \calO_+
  =
  \calO\cap\left(M\setminus J^-_g(K_\chi)\right).
$$
By construction, the three open sets $\calO_-$, $\calO_0$, $\calO_+$ cover
$\calO$.

On the middle region $\calO_0$, the local time-slice property for the metric
$g+\chi h$ gives
$$
  \calA_{g+\chi h}(\calO_0)
  \subset
  \calA_{g+\chi h}(\calV),
$$
where $\calV\subset\calO$ is a sufficiently small neighbourhood of
$\Sigma_0$ in $\calO$ chosen so that
$\supp(\chi h)$ does not intersect $J^-_g(\calV).$
This is possible because $\Sigma_0\subset\Sigma_{t-2\epsilon}$ lies below the
thin slab containing $\supp(\chi h)$. Since $\calV$ lies in the past of $\supp(\chi h)$, causality
$\eqref{caus1}$ gives
$$
  \calA_{g+\chi h}(\calV)=\calA_g(\calV).
$$
Therefore
$$
  \calA_{g+\chi h}(\calO_0)
  \subset
  \calA_g(\calV)
  \subset
  \calA_g(\calO).
$$

On the region $\calO_-$, the perturbation is not in the causal past of
$\calO_-$. Hence causality $\eqref{caus1}$ gives
$$
  \calA_{g+\chi h}(\calO_-)
  =
  \calA_g(\calO_-)
  \subset
  \calA_g(\calO).
$$

On the region $\calO_+$, the perturbation is not in the causal future of
$\calO_+$. Hence causality $\eqref{caus2}$ gives
$$
  \calA_{g+\chi h}(\calO_+)
  =
  S(g+\chi h,g)\calA_g(\calO_+)S(g+\chi h,g)^*.
$$
By the small-support case of locality $\eqref{loc2}$ 
we already know that $S(g+\chi h,g)\in\calA_g(\calO).$
Since $
  \calA_g(\calO_+)\subset\calA_g(\calO),$
we conclude that $
  \calA_{g+\chi h}(\calO_+)\subset\calA_g(\calO).$

By additivity, the three inclusions imply $
  \calA_{g+\chi h}(\calO)\subset\calA_g(\calO).$
Applying the same argument with $g$ and $g+\chi h$ interchanged gives the
reverse inclusion. Hence
$$
  \calA_{g+\chi h}(\calO)=\calA_g(\calO).
$$
This proves the thin-slab case, and hence $\eqref{loc1}$.

It remains to remove the small-support assumption in $\eqref{loc2}$. Let $h$
be any compactly supported metric perturbation in $\calO$. Choose a finite
decomposition
$h=h_1+\cdots+h_N$
such that each $h_j$ has sufficiently small support in the sense used above,
and such that the partial sums
$g_k=g+h_1+\cdots+h_k$
belong to $\mathcal M$. Then $S(g+h,g) = S(g_N,g_{N-1})\cdots S(g_1,g_0).$
By the small-support case of $\eqref{loc2}$, each factor belongs to
$\calA_{g_{k-1}}(\calO)$. By $\eqref{loc1}$, already proved above,
$$
  \calA_{g_{k-1}}(\calO)=\calA_g(\calO).
$$
Hence all factors belong to $\calA_g(\calO)$, and therefore
$$
  S(g+h,g)\in\calA_g(\calO).
$$
This proves $\eqref{loc2}$ in general.
\end{proof}
 
\begin{rem}
 The above construction results in a scattering matrix $S(\gamma)$ that does not depend on the chosen path but only on the endpoints. The price to pay for this is the dependence on the reference metric $g_0$. It would be interesting to see if there are other more natural choices, in particular allowing for the causality relation to be satisfied exactly without the phase factor correction. The choice of phase factor corresponds to the choice of additive constants in the definition of the stress energy tensor.

It is instructive to see what the stress energy tensor is concretely for the canonical implementer.
Using
$$
  \dot W R_{\Sigma_-}G_gf
  =
  R_{\Sigma_-}G_g\dot P G_gf
$$
and the canonical commutation relations, the resulting quadratic operator is, up to scalar multiple of the identity,
the normally ordered expression obtained from the first variation of the
classical Klein-Gordon action in the direction $\dot h$. Therefore
$$
  \left.\frac{\der}{\der s}\right|_{s=0}U_s
  =
  -\rmi \; \tilde T_g(\dot h)
$$
on the smooth simple finite particle vectors. Here $\tilde T_g(\dot h)$ is the normally ordered stress energy tensor, normally ordered with respect to the Hadamard state.
This computation is carried out in the proof of \cite{MR2007173}*{Th. 4.3} in the weak sense of quadratic forms.
Now we had set $S(\gamma) = S_c(g',g_0) S_c(g,g_0)^*,$
with respect to a reference metric $g_0$. This shows that
$$
 \left. T_g(\dot h) = \rmi \frac{\der}{\der s }\right|_{s=0}(S)(g_s,g) = S_c(g,g_0) \tilde T_g(\dot h) S_c(g,g_0)^* + c_g(\dot h) \mathbf 1
$$
which is, up to the constant scalar $c_g(\dot h) 1$ determined from the cocycle, the classical stress energy tensor, normally ordered with respect to the Hadamard state obtained by evolving from the reference metric, using $S_c(g,g_0)$.
One can compute $c_g(\dot h)$ explicitly.
With respect to the decomposition
$H_J\oplus\overline{H_J}$, write
$$
 \Lambda(g,g_0)
 =
 \begin{pmatrix}
  q_g & \overline{r_g}\\
  r_g & \overline{q_g}
 \end{pmatrix},
$$
and define as before
$$
 K_g=\overline{r_gq_g^{-1}}.
$$
Then one finds, essentially differentiating Equ. \eqref{impl}, that
\begin{align}
 c_g(\dot h)
 =
 \frac{\rmi}{4}\operatorname{Tr}\left[
  \left(\mathbf{1}-K_g^*K_g\right)^{-1}
  \left(
   K_g^*\der K_g(\dot h)
   -\der K_g(\dot h)^*K_g
  \right)
 \right],
\end{align}
where $\der K_g(\dot h)$ is the differential at $g$ of the map
$
 \mathcal{M} \to \mathcal{L}_{\mathrm{HS}}(H_J),
 g\mapsto K_g,
$
in the direction $\dot h$. 
\end{rem}

\appendix

\section{Some properties of globally hyperbolic spacetimes}

By a spacetime we will mean a smooth Lorentzian manifold $(M,g)$ that is oriented and time-oriented.
A spacetime is called globally hyperbolic if it admits a Cauchy surface $\Sigma$, i.e. a smooth spacelike \footnote{We restrict ourselves here to smooth spacelike Cauchy surfaces.} hypersurface that is met by every inextendible timelike curve exactly once. 
We define curves to be piecewise $C^1$ maps from an interval to $M$. A curve $\gamma: (a,b) \to M$ is called inextendible if it has no endpoint, i.e.
neither $\lim_{s \to a} \gamma(s)$ nor $\lim_{s \to b} \gamma(s)$ exist.
A function $T: M \to \R$ is called a temporal function if $\der T(X)>0$ for each future directed causal vector in $M$. A temporal function is called a Cauchy temporal function if the image of $T \circ \gamma$ is all of $\R$ for any inextendible causal curve $\gamma$.

In the following we will assume that $(M,g)$ is a globally hyperbolic spacetime.
By results of Geroch (\cite{MR0270697}) and Bernal-Sanchez (\cite{MR2254187,MR2029362})  there is a smooth Cauchy temporal function $t$ that induces a splitting $M = \R_t \times \Sigma$ as a smooth manifold such that each level set $\Sigma_t = \{t\} \times \Sigma$ is a Cauchy hypersurface and the metric takes the simple form
\begin{align}
 g = N^2 \der t^2 -h_t,
\end{align}
where $h_t$ is a smooth family of Riemannian metrics on $\Sigma$, and $N$ is a smooth positive function on $M$. Moreover, given any spacelike smooth Cauchy hypersurface one can choose the above so that this Cauchy hypersurface is the zero set of the Cauchy temporal function, in other words gets identified with $\{0\} \times \Sigma$.
The manifold then splits into past and future of the Cauchy surface as $M = M^+ \cup M^-,  M^+ \cap M^- = \Sigma$ with
$M^\pm = J^\pm(\Sigma)$. 

If $A \subset M$ is a compact subset of a globally hyperbolic spacetime $M$ then the sets $J^\pm(A)$ and $J(A)= J^+(A) \cup J^-(A)$ are closed (see \cite{ON}*{Lemma 22}). We say that two closed subsets $A_1,A_2 \subset M$ are causally separated if $A_1 \cap J(A_2)= \emptyset$ (or equivalently $A_2 \cap J(A_1)= \emptyset$). 
We also note that $A_2 \cap J^\pm(A_1) = \emptyset$ if and only if $A_1 \cap J^\mp(A_2) = \emptyset$.

\begin{proposition} \label{sepper}
 Let $A_1,A_2$ be compact subsets of a globally hyperbolic spacetime $M$.
 In case $A_2 \cap J^+(A_1) = \emptyset$ the sets $A_1$ and $A_2$ can be separated by a Cauchy surface, i.e. there exists a Cauchy surface $\Sigma$ so that
 $A_1 \subset J^+(\Sigma)$ and $A_2 \subset J^-(\Sigma)$.
\end{proposition}
\begin{proof}
 Consider the open subset $M \setminus J^+(A_1)$ and note that $A_2$ is a compact subset in it. Next consider any Cauchy surface $\Sigma_-$ in $M$ that does not intersect 
 $J^+(A_1)$. Such a Cauchy surface exists because $A_1$ is compact and therefore for any foliation $\R_t \times \Sigma_0$ as above and $t$ sufficiently negative the Cauchy surface in $M$ associated to $\{t\} \times \Sigma_0$ does not intersect $A_1$. Choosing this Cauchy surface in the past of $A_1$ it will then not intersect $J^+(A_1)$. We now argue that $\Sigma_-$ is also a Cauchy surface in $M \setminus J^+(A_1)$. Consider an inextendible timelike curve in $M \setminus J^+(A_1)$. Such a curve is either inextendible as a curve in $M$ or an endpoint is in $J^+(A_1)$. Such an endpoint must be a future endpoint and therefore the curve intersects $\Sigma_-$ precisely once. 
Since $M \setminus J^+(A_1)$ admits a Cauchy surface it is globally hyperbolic and $A_2$ is a compact subset in it. We now apply the same argument to the spacetime $N= M \setminus J^+(A_1)$ with its time-orientation reversed and apply the argument to the compact subset $A_2 \subset N$. Since $J_N^-(A_2)= J^-(A_2) \subset N$, we conclude that
 there exists a Cauchy surface in $N \setminus J^-(A_2) = (M \setminus J^+(A_1)) \setminus J^-(A_2)$. 
 
 Finally we show that any Cauchy surface $\Sigma_N$ in $N = M \setminus J^+(A_1)$ is also a Cauchy surface in $M$. 
 Indeed, let $\gamma$ be an inextendible timelike
curve in $M$. Since $A_1$ is compact, $\gamma$ cannot be entirely
contained in $J^+(A_1)$. Hence $\gamma\cap N$ is a non-empty initial
segment of $\gamma$. This segment is inextendible as a timelike curve in
$N$. Therefore it
meets the Cauchy surface of $N$ exactly once. Consequently $\gamma$ itself
intersects $\Sigma_N$ exactly once.
The same two-step argument as above now also shows that any Cauchy surface $\Sigma$ in $(M \setminus J^+(A_1)) \setminus J^-(A_2)$ is also a Cauchy surface in $M$. The proof is finished. 
\end{proof}
 
The following is an expression of the fact that metric changes do not change the causality relations in causally separate regions.
 
\begin{proposition}
 Assume that two Lorentzian metrics $g',g$ coincide outside a compact set $K \subset M$ and assume that $(M,g)$ and $(M,g')$ are globally hyperbolic spacetimes with coinciding time-orientations outside of $K$. Assume $A \subset M$ is a compact set with $A \cap J_g^+(K) =A \cap J_{g'}^+(K) = \emptyset$. Then $J^-_{g'}(A) = J^-_g(A)$. 
\end{proposition}
\begin{proof}
 It is obviously enough to show the inclusion $J^-_g(A) \subset J^-_{g'}(A)$, since equality follows from interchanging the two metrics.
 Assume that $x \in J^-_g(A)$, i.e. there exists $y \in A$ and a past-directed causal curve from $y$ to $x$. Since $K \cap J_g^-(A) = \emptyset$ this curve does not intersect $K$ and hence is also causal for $g'$. It follows that $x \in J^-_{g'}(A)$.
\end{proof}

Unlike in Riemannian geometry convex linear combinations of Lorentzian metrics may fail to be Lorentzian even if they share a common time orientation. A counterexample are the two metrics
$$
  g_1 = \der t^2 + 4 \der t \der x +\der x^2, \quad g_2 = \der t^2 - 4 \der t \der x +\der x^2
$$
on $\R^2$, each of which has signature $(+,-)$ and determinant $-3$. The vector field $\partial_t$ is timelike for both metrics, hence defines a common time orientation. Nevertheless, $\frac{1}{2}(g_1+g_2)= \der t^2 + \der x^2$ has Riemannian signature.

There is however a natural condition that ensures the convex linear combinations remain Lorentzian. Once this condition is satisfied such metrics may be glued using smooth gluing functions.

\begin{proposition} \label{prop-app-1}
 Assume $M$ is a smooth manifold with Lorentzian metrics $g_1$ and $g_2$. Assume there exists a function $t: M \to \R$ which is temporal with respect to both $g_1$ and $g_2$. Assume in addition there exists a vector field $Z$ on $M$ which is timelike with respect to both metrics. Then, given any two smooth functions $\chi_1,\chi_2 \geq 0$ satisfying
$\chi_1+\chi_2>0$ everywhere we have that $g= \chi_1 g_1 + \chi_2 g_2$ is a Lorentzian metric and $t$ is a temporal function for $g$.
\end{proposition}
\begin{proof}
The assertion is immediate at points where one of the coefficients
vanishes. At points where both coefficients are positive, rescaling
reduces the argument to the case $\chi_1=\chi_2=1$, which we now assume without loss of generality. We only need to check that $g$ has Lorentzian signature on $T_xM$ for every $x$. Consider the subspace $V = \{X \in T_xM \mid \der t(X)=0 \}$. Then $g_1,g_2$ are negative definite on $V$ and hence so is $g$.
 For the timelike vector $Z_x$ at the point $x$ we have $g(Z_x,Z_x) = g_1(Z_x,Z_x)  + g_2(Z_x,Z_x)>0$. For dimensional reasons this implies that $g$ is non-degenerate of Lorentzian signature. 
 Since $\ker \der t$ is spacelike for $g$, the covector $\der t$ is
timelike for $g^{-1}$. Choosing the time-orientation for $g$ so that
$\der t$ is positive on future-directed causal vectors, $t$ is temporal
for $g$.
\end{proof}

\begin{proposition}\label{prop-app-2}
 Assume that $(M,g)$ is a globally hyperbolic spacetime
 and assume further that $h \in C^\infty_0(M; \mathrm{Sym}^2 T^*M)$ is such that $g + h$ is a Lorentzian metric on $M$.
 Assume that there exists a Cauchy temporal function $t: M \to \R$ with respect to the metric $g$ such that $\der t$ is timelike with respect to the metric $(g+h)^{-1}$.
 Then  
 $(M, g +  h)$ is globally hyperbolic.
\end{proposition}
\begin{proof}
The timelike covector $\der t$ determines the time-orientation on
$(M,g+h)$, which we will use.
We assume without loss of generality that $h$ is supported in a compact set  $K \Subset M$ contained in $[a,\infty) \times \Sigma$ for some $a>0$ so that  $\Sigma_0$ is in the past of $K$.
We will show that $\Sigma_0$ is a Cauchy surface for $(M,g+h)$ and therefore that $(M,g+h)$ is globally hyperbolic.
For any future directed piecewise smooth causal curve $\gamma$ with respect to  the metric $g'=g+h$ we have
 $\frac{\der}{\der s}  t(\gamma(s)) = \der t(\gamma'(s)) > 0$.  Therefore,
 $t \circ \gamma$ is a strictly increasing function of the parameter $s$ and it thus follows that $\Sigma_0$ is acausal.
 Since $\Sigma_0$ is acausal it is sufficient to show that any maximal affinely parametrised null-geodesic $\gamma$ intersects $\Sigma_0$ exactly once ( \cite{ON}*{54 Corollary}).
To control the parametrisation we choose an arbitrary Riemannian metric $\tilde g$ and define the function
$q(\xi) =( \tilde g^{-1}(\xi,\xi) )^{-\frac{1}{2}}= \|\xi \|_{\tilde g}^{-1}$. Denote by $N_{g'}$ the set of null covectors in $\dot T^* M$ with respect to the metric $g'$.
Then the geodesic vector field on $N_{g'}$ with respect to the metric $g'$
is proportional to the Hamiltonian vector field generated by 
$q(\xi) H_{g'}(\xi) = -\frac{1}{2}q(\xi) (g')^{-1}(\xi,\xi)$. Indeed, for the Poisson brackets we have
\begin{align}
 \{q H_{g'} , f\} = q \{H_{g'},f\} + H_{g'} \{q,f\} = q \{ H_{g'},f\}
\end{align}
since $H_{g'}(\xi)$ vanishes on $N_{g'}$.
The factor $q(\xi)$ turns this into a function on cotangent space that is homogeneous of degree one and therefore the corresponding Hamiltonian vector field is homogeneous of degree zero.
In local coordinates it is given by
\begin{align}
 X = -\|\xi \|_{\tilde g}^{-1} \left(  \frac{1}{2}\frac{\partial (g')^{jk}}{\partial  x^m}  \xi_j \xi_k \partial_{\xi_m} - (g')^{jk} \xi_j \partial_{x^k} \right).
\end{align}
Null geodesics are the orbits of this vector field. 
Since $\der t$ is timelike for $g'$, and since the time-orientation is chosen
so that $t$ increases on future-directed $g'$-causal curves, the quantity
$$
  \|\xi\|_{\tilde g}^{-1}(g')^{-1}(\der t,\xi)
$$
has a fixed positive sign on the future null cone. In particular, on the
compact set $K$, after restricting to the $\tilde g$-unit null covectors, it is
bounded below by a positive constant.

It follows that an orbit of $X$ cannot remain in $K$ for infinite parameter
time. Hence every inextendible null geodesic of $g'$ eventually leaves $K$ in both time directions. In particular, it leaves $K$ when followed to the past.
Outside $K$, the metrics $g'$ and $g$ coincide.
Therefore the geodesic agrees there with a null geodesic of $g$. Since
$\Sigma_0$ is a Cauchy surface for $g$, this past part intersects $\Sigma_0$.
Thus every inextendible null geodesic of $g'$ intersects $\Sigma_0$.

Since $t$ is strictly increasing along future-directed $g'$-causal curves,
$\Sigma_0$ is acausal, and the intersection with any inextendible null geodesic
is unique. By the cited criterion, $\Sigma_0$ is a Cauchy surface for $g'$.
 \end{proof}
 
 \begin{proposition}
  Assume that $(M,g)$ is a globally hyperbolic spacetime.
  Given any $h \in C^\infty_0(M; \mathrm{Sym}^2 T^*M)$ there exists $\epsilon>0$ such that
  $(M,g+ s h)$ is a globally hyperbolic spacetime for any $s \in [-\epsilon,\epsilon]$.
 \end{proposition}
 \begin{proof}
  Let $K = \mathrm{supp}(h)$.
  Since the set of matrices of signature $(+,-,\ldots,-)$ is an open subset in the set of matrices (also as a consequence of Prop. \ref{prop-app-1}) we can choose
  $\epsilon_1>0$ such that $g+ s h$ has Lorentzian signature for all $s \in [-\epsilon_1,\epsilon_1]$.
  Next we choose a Cauchy temporal function $t$ for $(M,g)$ and note that $g^{-1}(\der t, \der t) > \delta >0$ on $K$.
  It follows that there exists $0<\epsilon< \epsilon_1$ with $(g+s h)^{-1}(\der t, \der t) >0$ for all
  $s \in [-\epsilon,\epsilon]$. By Prop.  \ref{prop-app-2} this shows that such metrics are globally hyperbolic.
 \end{proof}

\begin{lemma} \label{vectorfieldZ}
 Assume that $(M,g)$ is a globally hyperbolic spacetime with smooth spacelike Cauchy surface $\Sigma \subset M$.
 Let $\calO$ be an open neighborhood of $\Sigma$ and $\calU$ an open subset with compact closure in $\mathrm{int}(D^+(\Sigma))$.
 Then there exists a vector field $Z$ on $M$ with compact support in $\mathrm{int}(D^+(\Sigma))$, a number $s_0>0$,
 and an open subset $\calV$ with compact closure in $\calO$ such that the following holds
 for $Z$ and its flow $\phi_s$:
 \begin{itemize}
  \item[(a)] for each $x \in M$ the vector $Z(x)$ is either zero or timelike and past-directed.
  \item[(b)] $\phi_{s_0}(J^-(\overline{\calU}) \cap D^+(\Sigma)) \subset \calV$.
  \item[(c)] there exists a Cauchy temporal function $t:M\to\R$ for $g$ whose restriction to $\calV$ is temporal for the metrics
  $\phi_s^*g$ for all $s\in[-s_0,s_0]$.
  \item[(d)] there exists a vector field which is timelike for $g$ and whose restriction to $\calV$ is timelike for the metrics
  $\phi_s^*g$ for all $s\in[-s_0,s_0]$.
 \end{itemize}
\end{lemma}

\begin{proof}
 Since $(M,g)$ is globally hyperbolic it is isometric to a product of the form
 $$(\R_t\times\Sigma,N^2\der t^2-h_t),$$
 where $h_t$ is a smooth family of metrics and $N$ is a positive function on
 $\R_t\times\Sigma$. This isometry can be chosen so that $\Sigma$ is mapped
 to $\{0\}\times\Sigma$. Since the statement is invariant under conformal
 transformations of the metric we can therefore assume without loss of
 generality that
 $$
 (M,g)=(\R_t\times\Sigma,\der t^2-h_t).
 $$

 Set
 $$
 A=J^-(\overline{\calU})\cap D^+(\Sigma).
 $$
 The set $A$ is compact. Indeed,
 $J^-(\overline{\calU})\cap\Sigma$ is compact and, setting
 $C=J^-(\overline{\calU})\cap\Sigma$, we have
 $$
 A=J^+(C)\cap J^-(\overline{\calU}),
 $$
 which is compact by global hyperbolicity.

 We can therefore choose $T_2>T_1>0$, a compact subset $Q\subset\Sigma$
 and a relatively compact open neighborhood $V$ of $Q$ such that $
 A\subset [0,T_2]\times Q$
 and $
 [0,T_1]\times\overline V\subset\calO.$
 Choose $\eta\in C^\infty_0(\Sigma)$, $0\leq\eta\leq1$, such that
 $\eta=1$ on $V$. We also choose a compactly supported smooth function
 $$
 \chi\in C^\infty_0((0,2T_2)),\qquad 0\leq\chi\leq1,
 $$
 which equals one on $[\frac12T_1,T_2]$. Then
 $$
 Z(t,x)=-\eta(x)\chi(t)\partial_t
 $$
 is a smooth compactly supported vector field in
 $\mathrm{int}(D^+(\Sigma))$ which is either zero or past-directed timelike.

 Let $G_s$ be the flow generated by $-\chi(t)\partial_t$ on $\R$.
 On $\R\times V$ the flow of $Z$ is given by
 $$
 \phi_s(t,x)=(G_s(t),x),
 $$
 and
 $$
 \frac{\der G_s(t)}{\der t}>0.
 $$
 Since $G_s(t)\leq t$ and $\chi=1$ on
 $[\frac12T_1,T_2]$, there exists a finite $s_0>0$ such that
 $$
 \phi_{s_0}(A)\subset [0,T_1]\times Q.
 $$
 We can therefore choose an open neighborhood $\calV$ of
 $\phi_{s_0}(A)$ with compact closure in
 $$
 \calO\cap(\R\times V).
 $$
 This proves (a) and (b).

 On $\R\times V$ we have, for all $s\in\R$,
 $$
 \phi_s^*g
 =
 \left(\frac{\der G_s(t)}{\der t}\right)^2\der t^2
 -
 h_{G_s(t)}.
 $$
 Since $\frac{\der G_s(t)}{\der t}>0$, the function $t$ is temporal
 for all these metrics and $\partial_t$ is timelike for all of them.
 This proves (c) and (d).
\end{proof}

\section{Shale's theorem and further properties of the implementation} \label{Shalendix}

Given a real Hilbert space $H$ with a complex structure $J$ we can write its complexification $H \otimes_\R \C$ as $H_J \oplus \overline{H_J}$.
We view $H$ as a real symplectic vector space with symplectic form $\sigma(\cdot,\cdot) = -\langle J \cdot,\cdot \rangle$.
As in Section \ref{Fockspace} we denote the symmetric Fock space over $H_J$ by $\mathcal{F}(H_J)$ and the dense finite particle subspace by
$\mathcal{F}_\fp(H_J)$. For $v \in H$ the field operators $\phi(v)$ are then maps $\phi(v): \mathcal{F}_\fp(H_J) \to \mathcal{F}_\fp(H_J)$.
Here $H_J$ is fixed and we will therefore simply write $\mathcal{F}$, $\mathcal{F}_\fp$, etc.

We will denote by $\numb$ the number operator, i.e. the self-adjoint operator that acts by multiplication by
$k$ on the $k$-particle subspace. The domain of smoothness $\cap_{k=1}^\infty \mathrm{dom}(\numb^k)$ will be denoted by
$\mathcal{F}_\infty$. This is the space of vectors
$(v_0,v_1,\ldots,v_n,\ldots)$ with $v_n$ in the completed tensor product $\hat\otimes_\mathrm{s}^n H_J$ with the property that
$\| v_n \|$ is rapidly decreasing in the sense that the sequence $(n^k \| v_n \|)_n$ is bounded for any $k \in \mathbb{N}_0$.
The best constants in the implied estimates introduce a Fr\'echet space topology on $\mathcal{F}_\infty$.
The field operators $\phi(v)$ extend continuously to maps $\mathcal{F}_\infty \to \mathcal{F}_\infty$. Of course the finite particle subspace
$\mathcal{F}_\fp$ is a dense subspace of  $\mathcal{F}_\infty$.
Given $n >0$ we will denote by $\mathcal{F}_n$ the range of the spectral projection of $\numb$ on the interval $[0,n]$, in other words the finite direct sum
\begin{align}
 \mathcal{F}_n = \bigoplus_{k=0}^n {\hat\bigotimes_S}^k H_J.
\end{align}

Assume now that $W: H \to H$ is an invertible real linear symplectic map, and we will denote the complex linear map on $H_J \oplus \overline{H_J}$ by 
the same letter. Then, with respect to the decomposition $H \otimes_\R \C = H_J \oplus \overline{H_J}$, we have
\begin{align}
W = \left( \begin{matrix} q & \overline{r} \\ r & \overline{q} \end{matrix} \right).
\end{align}
The terms $r, \overline{r}$ appear because of the possible failure of $W$ to commute with $J$.

We are looking for an implementer of $W$ on the bosonic Fock space, i.e. we are looking for a unitary operator $U$ on Fock space with $U \mathcal{F}_\fp \subset \mathcal{F}_\infty$ such that
 \begin{align}
 U \phi(f) U^* = \phi( W f)
\end{align}
on $\mathcal{F}_\fp$.
If such an implementer exists we say $W$ can be implemented. In this case of course
\begin{align}
 U  p(\phi(f_1) \cdots \phi(f_n)) U^* =  p(\phi(W f_1) \cdots \phi(W f_n))
\end{align}
for any (non-commutative) polynomial expression $p$ in the fields.

In the following we describe the theory of implementation in Fock space which is now very well established (\cite{MR0452323,MR0516713,MR2105911,MR2297950}, see also the textbook \cite{MR3060648}*{Ch. 11} for a comprehensive treatment).
We will however need several statements that go beyond what we could find in the literature. We first describe some known results that are largely contained in \cite{MR0452323,MR0516713} and introduce convenient notations where we follow mostly
\cite{MR2297950}.

In case $W J - J W$ is Hilbert-Schmidt 
the operator $q$ is invertible and the operator $r$ is Hilbert-Schmidt. Then the operators
\begin{align}
 K = \overline{r q^{-1}}, \quad L = - q^{-1} \overline{r}.
\end{align}
are Hilbert-Schmidt operators with $\| K \| < 1$ and $\| L \|<1$. These operators are symmetric in the sense that $\overline{K} = K^*$ and $\overline{L} = L^*$.
Given a symmetric rank one operator $A = \langle \overline{v}, \cdot \rangle v$ one defines $a^*(A) = a^*(v) a^*(v)$.
The map $A \mapsto a^*(v) a^*(v)$ extends by linearity to the set of symmetric finite rank operators. On the $n$-particle subspace we have
\begin{align}
 \| a^*(A) \psi \| \leq \| A\|_{\mathrm{HS}} \cdot \|\sqrt{(\numb+2)(\numb+1)} \psi\|
\end{align}
and we can extend the map by continuity to the set of symmetric Hilbert-Schmidt operators.
For a symmetric Hilbert-Schmidt operator we also define
\begin{align}
 a(A) = \left(a^*(A)\right)^*
\end{align}
and note
\begin{align}
 \| a(A) \psi \| \leq \| A\|_{\mathrm{HS}} \| \numb \psi\|.
\end{align}

For any vector $\psi$ that is a finite linear combination of vectors of the form $a^*(f_1) \cdots a^*(f_N) \Omega$ one can define
\begin{align}
 e^{-\frac{1}{2} a^*(K)} \psi &= \sum_{j=0}^\infty (-1)^j \frac{1}{2^j} \frac{1}{j!} \left( a^*(K) \right)^j \psi,\\
 e^{\frac{1}{2} a(L)} \psi &= \sum_{j=0}^\infty \frac{1}{2^j} \frac{1}{j!} \left( a(L) \right)^j \psi
\end{align}
with the sums converging in the norm. 
Vectors of this form are exactly the finite linear combinations of simple tensor products and the space of these vectors is denoted by
$\mathcal{F}_{\sfp}$, where the subscript stands for simple finite particle space.
Given any bounded operator $A: H_J \to H_J$ one defines the second quantisation $\Gamma(A)$ as 
\begin{align*}
 \Gamma(A) = \mathbf{1} \oplus A \oplus A \otimes A \oplus \cdots 
\end{align*}
as defined on $\mathcal{F}_{\sfp}$. The space $\mathcal{F}_{\sfp}$ is left invariant under this map.
One then has the following result and formula.

\begin{theorem}[Shale's theorem]
 The symplectic map $W$ can be implemented if and only if $W J - J W$ is Hilbert-Schmidt. Then the operator
\begin{align} \label{impl}
 U = \left(\mathrm{det}(1 - K^* K) \right)^\frac{1}{4} e^{-\frac{1}{2} a^*(K)} \Gamma( (q^{-1})^*) e^{-\frac{1}{2} a(L)}.
\end{align}
defined on the subspace of finite linear combinations of vectors of the form $a^*(f_1) \cdots a^*(f_n) \Omega$ 
extends by continuity to a unitary operator that implements the transformation $W$.
\end{theorem}

\begin{definition}
The $U$ defined above in \eqref{impl} is called the canonical implementer. 
\end{definition}
Unitary implementers are unique up to a phase in $U(1)$, the
canonical implementer is uniquely characterised by $\langle \Omega, U  \Omega\rangle >0$.

\begin{proposition} \label{smoothy}
Given a Hilbert-Schmidt operator $K$ with $\| K \|_{\mathrm{HS}}< 1$  the sum
 \begin{align}
 e^{-\frac{1}{2} a^*(K)} &= \sum_{j=0}^\infty (-1)^j \frac{1}{2^j} \frac{1}{j!} \left( a^*(K) \right)^j
\end{align} 
converges for any $n \in \mathbb{N}_0$ in the Fr\'echet space $\mathcal{L}(\mathcal{F}_n, \mathcal{F}_\infty)$. 

Assume that $K=K(s)$ depends additionally on a parameter $s$ in a compact parameter manifold $I$ and we assume that $K(s)$ is a $C^k$-function of the parameter $s$ with values in the Banach space of Hilbert-Schmidt operators and such that $\| K(s) \|_{\mathrm{HS}}< 1$ for all $s$ in the parameter manifold.
Then, for every $n\in\mathbb{N}_0$, the sum converges in
$$
 C^k\bigl(I,\mathcal{L}(\mathcal{F}_n,\mathcal{F}_\infty)\bigr).
$$
Here the $C^k$ topology is defined using the Fr\'echet semi-norms on $\mathcal{F}_\infty$.
\end{proposition}
\begin{proof}
 Let $\psi$ be a vector in the $n$-particle subspace and $m \in \mathbb{N}_0$ . Then we have the bound
 \begin{align}
 \| \numb^m (a^*(K))^j \psi \| &  \leq (n+2j)^m(n+1)^\frac{1}{2}(n+2)^\frac{1}{2} \ldots (n+2j)^\frac{1}{2} \| K \|_\mathrm{HS}^j \nonumber \\ & \leq  (n+2j)^m \left(\frac{(n+2j)!}{n!}\right)^\frac{1}{2} \| K \|_\mathrm{HS}^j \| \psi \|.
\end{align}
For the sums we therefore get
\begin{align}
 \sum_{j=0}^\infty &  \|\numb^m \frac{1}{2^j} \frac{1}{j!} \left( a^*(K) \right)^j \psi \| \leq
 \sum_{j=0}^\infty \frac{1}{2^j} \frac{(n+2j)^m}{j!} \left( \frac{(n+2j)!}{n!}\right)^\frac{1}{2} \| K \|_\mathrm{HS}^j \| \psi \| \leq C_{m,n} \| \psi \|,
\end{align}
as the sum converges by the quotient criterion. 

It remains to show convergence in case we have parameter dependence. By compactness of the parameter space
we have $\| K(s) \|_{\mathrm{HS}} < 1 - \delta$ for some $\delta>0$. The above argument then shows also that we have convergence in the Fr\'echet space of 
continuous functions in $s$ taking values in $\mathcal{L}(\mathcal{F}_n, \mathcal{F}_\infty )$.
For the first derivative with respect to any coordinate in a local chart we have 
\begin{align}
 \frac{\partial}{\partial s_r} \left( a^*(K(s)) \right)^j &= j\, a^*(  \frac{\partial}{\partial s_r}  K(s)) \left( a^*(K(s)) \right)^{j-1}.
\end{align}
For the higher derivatives of order $q$ of $\left( a^*(K(s)) \right)^j$ we obtain finite linear combination of terms that are
polynomial expressions in $j$ and in derivatives of $a^*(K(s))$ which is positively homogeneous of degree $j$ in $K$. 
In the sums this produces an additional factor $p(j)$,
where $p(j)$ is a polynomial in $j$ of degree $k$. To test the convergence in the $C^k$-norm for $\psi$ in the $n$-particle subspace we end up having to estimate
 \begin{align}
 \sum_{j=0}^\infty &  p(j) \|\numb^m \frac{1}{2^j} \frac{1}{j!} \left( a^*(K) \right)^{j-k} \psi \| 
 \end{align}
We obtain convergence in the Fr\'echet space of $C^k$-functions on the parameter manifold taking values in
$\mathcal{L}(\mathcal{F}_n, \mathcal{F}_\infty)$ for any $n \in \mathbb{N}_0$.
\end{proof}

A slightly weaker conclusion holds without the smallness assumption on the Hilbert-Schmidt norm but rather on the norm.
\begin{proposition} \label{smoothytwo}
Assume $K=K(s),L=L(s)$ are $C^k$-functions of the parameter $s$ in some compact parameter manifold taking values in the Banach space of Hilbert-Schmidt operators 
such that $\| K(s) \|< 1$ for all $s$ in the parameter manifold. 
Assume also that $v(s)$ is a smooth family
of functions taking values in $\mathcal{F}_\sfp$ which is a finite linear combination of simple vectors of the form
$$
 v_1(s) \otimes_\mathrm{s} \ldots \otimes_\mathrm{s} v_N(s)
$$
where $v_1(s), \ldots, v_N(s)$ are smooth functions taking values in $H_J$.
Then the sums
 \begin{align*}
 e^{-\frac{1}{2} a^*(K)} v(s)&= \sum_{j=0}^\infty (-1)^j \frac{1}{2^j} \frac{1}{j!} \left( a^*(K(s)) \right)^j v(s),\\
 e^{-\frac{1}{2} a(L)} v(s)&= \sum_{j=0}^\infty \frac{1}{2^j} \frac{1}{j!} \left( -a(L(s)) \right)^j v(s)
 \end{align*} 
converge in $C^k(I,\mathcal{F})$.
\end{proposition}
\begin{proof}
 The family of vectors is a finite linear combination of vectors of the form $v(s) = a^*(v_1(s)) \cdots a^*(v_N(s)) \Omega$. 
 Since $a^*(v_1(s)) \cdots a^*(v_N(s))$ commutes with the individual terms in the sum it will be sufficient to show that the sum
 $$
   \sum_{j=0}^\infty (-1)^j \frac{1}{2^j} \frac{1}{j!} \left( a^*(K(s)) \right)^j \Omega
 $$
 converges in $C^k(I,\mathcal{F}_\infty)$. 
 Taking derivatives and bearing in mind that the terms $a^*(K(s))$ commute with $a^*(K')$ for any Hilbert-Schmidt operator $K'$ the same argument shows that we only need to
 show convergence of
 $$
   \sum_{j=0}^\infty (-1)^j \frac{1}{2^j} \frac{1}{j!} p(j) \left( a^*(K(s)) \right)^{j-k} \Omega
 $$
 in $\mathcal{F}_\infty$ where $p(j)$ is a polynomial of degree $k$. This is implied by the convergence of 
 $$
  \sum_{j=0}^\infty \left( \frac{1}{2^j} \frac{1}{j!} p(j) \right)^2 \| a^*(K(s))^{j-k} \Omega \|^2
 $$
 for any $\ell\in\mathbb N_0$ and any polynomial $p$ of degree $k + 2 \ell$, which follows by comparison to series 
 $$
  \sum_{j=0}^\infty \left(\frac{1}{2^{j}} \frac{1}{(j!)} \right)^2 \| a^*(K(s)) ^j \Omega \|^2 z^j
 $$
 whose radius of convergence is $\|K(s)\|^{-2}$ (see \cite{MR0516713}*{Lemma 4.3}).
 For $e^{-\frac{1}{2} a(L)}$ the sum is in fact finite, so the statement follows immediately.
 \end{proof}

The set of symplectic maps $W : H_J \oplus \overline{H_J} \to H_J \oplus \overline{H_J}$ of the form
$$
\left( \begin{matrix} q & \overline{r} \\ r & \overline{q} \end{matrix} \right)
$$
is sometimes also called the restricted symplectic group (see for example \cite{MR3060648}*{Ch. 11}). We endow this group with
the topology induced by using the strong topology for the component $q$ and the Hilbert-Schmidt norm topology for the $r$ component.
Hence, a sequence $W_j$ converges to $W$ if and only if $q_j \to q$ strongly and $\| r_j -r\|_\mathrm{HS} \to 0$. A side-remark is that our Hilbert spaces are separable by construction, hence this topology is in fact metrisable.
We have the following continuity statement.
\begin{proposition} \label{strongprop}
 For each $W$ in the restricted symplectic group let $U_W$ be its natural Bogoliubov implementer.
 Then the maps $W \mapsto U, W \mapsto U^*$ are continuous from the restricted symplectic group to the unitary operators endowed with the strong topology.
\end{proposition}
\begin{proof}
Since it is sufficient to check strong continuity on a dense subset it is sufficient to show that the map $W \mapsto U v$ is continuous for any $v$ 
in the simple finite particle subspace $\mathcal{F}_{\sfp}$. 
Continuity at the identity element in the group follows immediately from the explicit formula
\begin{align}
 U = \left(\mathrm{det}(1 - K^* K) \right)^\frac{1}{4} e^{-\frac{1}{2} a^*(K)} \Gamma( (q^{-1})^*) e^{-\frac{1}{2} a(L)}
\end{align}
bearing in mind that for $v \in \mathcal{F}_{\sfp}$ the vector $e^{-\frac{1}{2} a(L)} v$ is a finite linear combination
of simple vectors, independent of $L$, with coefficients that depend on $L$ in a multilinear fashion, bounded by the Hilbert-Schmidt norm of $L$.
We use here a consequence of the restricted symplectic identities
that the maps
$$
  W\mapsto q^{-1},\qquad W\mapsto K,\qquad W\mapsto L
$$
are continuous for the restricted symplectic topology, where $q$ is equipped
with the strong operator topology and the off-diagonal blocks with the
Hilbert-Schmidt topology. Indeed, this follows from
$$
  q^*q-r^*r=\mathbf 1
$$
and from the formulas expressing $K$ and $L$ in terms of $q^{-1}$ and the
Hilbert-Schmidt block $r$.
Finally, it is easy to see that for any sequence $K_j$ converging to $0$ in the Hilbert-Schmidt norm we have that
$$
 e^{-\frac{1}{2} a^*(K_j)} - \mathbf{1}
$$
converges to zero in operator norm as maps $\mathcal{F}_n\to\mathcal{F}_\infty$.
To establish strong continuity at non-zero points we can use the formula
\begin{align}
 U = \sigma(W_1,W_2) U_1 U_2
\end{align}
in case $W = W_1 W_2$, where the cocycle $\sigma(W_1,W_2) \in U(1)$
is explicitly given by
\begin{align}
  \sigma(W_1,W_2) =\mathrm{det}(1 - K_2 L_1^*)^\frac{1}{2}   \mathrm{det}(1 - K^* K)^\frac{1}{4} \mathrm{det}(1 - K_1^* K_1)^{-\frac{1}{4}}  \mathrm{det}(1 - K_2^* K_2)^{-\frac{1}{4}}.
\end{align}
Here the square root denotes the continuous branch fixed by its value
$1$ at the identity.
Since the Fredholm determinant defines a continuous map $A \to \mathrm{det}(1+A)$ from the set of trace-class operators to $\C$ this also shows continuity of the cocycle.
The same argument also applies to the adjoint as it is given by the formula
\begin{align}
  U^* = \left(\mathrm{det}(1 - K^* K) \right)^\frac{1}{4}  e^{-\frac{1}{2} a^*(L)}  \Gamma( (q^{-1})) e^{-\frac{1}{2} a(K)}.
\end{align}
\end{proof}

\section{Parameter dependent fundamental solutions} \label{Appendparam}

Let $M$ be a smooth manifold with a smooth family of globally hyperbolic metrics $g_s$ indexed by a parameter $s \in I$ in some finite dimensional compact smooth manifold $I$ (with or without boundary). We will be given a smooth family of second order differential operators
$P_s$, each $P_s$ having principal symbol $g_s^{-1}$.
In this section we will show that under natural conditions the relevant fundamental solutions depend smoothly on the parameter $s$ and we will make this statement precise. For our purposes information about the propagation of singularities is required and we will therefore use the Fourier integral operator method of \cite{MR0388464} (see also \cite{MR1362544}) to construct the fundamental solutions.
For the sake of concreteness we use the volume form for the metric $g_0$ to identify locally integrable functions with distributions. Alternatively one can also consider operators acting on half densities, in which case one does not have to commit to such a choice. We will identify continuous maps $C^\infty_0(M) \to \mathcal{D}'(M)$ with the distributional kernels in $\mathcal{D}'(M \times M)$ by means of the Schwartz kernel theorem.

We are assuming that
\begin{itemize}
 \item $(M,g_0) = (\R \times \Sigma, g_0 = N^2 \der t^2 - h_t)$ is globally hyperbolic with Cauchy time function $t$ and foliation by Cauchy surfaces $\{t\} \times \Sigma$.
 \item there is a compact set $K \subset M$ such that $\mathrm{supp}(g_s-g_0) \subset K$ for all $s \in I$,
 \item the function $t$ is a time-function for all metrics $g_s, s \in I$,
 \item the vector field $\partial_t$ is a time-orientation for all metrics $g_s, s \in I$.
\end{itemize} 

For each $s\in I$ we consider the geodesic relation
$$
 (x,\xi)\sim_s(x',\xi')
$$
on the null covectors $\dot T^*M$, where $(x,\xi)$ and $(x',\xi')$ are on the
same orbit of the Hamiltonian flow of
$$
  H_s(x,\xi)=-\frac12 g_s^{-1}(\xi,\xi).
$$
This defines a homogeneous Lagrangian submanifold
$$
  \Lambda_s\subset N_{g_s} \times N_{g_s} \subset T^*M\times T^*M\setminus 0,
$$
where the symplectic form on $T^*M\times T^*M$ is the difference of the two
canonical symplectic forms. Here $N_{g}$ denotes the set of null-covectors with respect to the metric $g$.
We write $\Lambda_s^\pm$ for the two components
according to whether the covectors are future or past directed.

Let $\Sigma$ be a Cauchy surface. Then $\Lambda_s^\pm$ admits the following
parametrisation. For $(y,\eta)\in\dot T^*\Sigma$, let $\eta_s^\pm$ be the
unique future, respectively past, directed null lift of $\eta$ to $T^*M$ with
respect to $g_s$. If $G_s^t$ denotes the Hamiltonian flow of $H_s$, set
$$
  \rho_s^\pm(t;y,\eta)=G_s^t(y,\eta_s^\pm).
$$
Then
$$
  (y,\eta,t_1,t_2)
  \longmapsto
  \bigl(\rho_s^\pm(t_1;y,\eta),\rho_s^\pm(t_2;y,\eta)\bigr)
$$
parametrises $\Lambda_s^\pm$.

We now define the parameter-dependent canonical relation
$\Lambda\subset T^*I\times T^*M\times T^*M$.
For a point of the above parametrisation, the $T_s^*I$-component is defined by
$$
  \tau_s^\pm(y,\eta,t_1,t_2)
  =
  -\int_{t_2}^{t_1}
  \partial_s H_s \bigl(\rho_s^\pm(r;y,\eta)\bigr)\,\der r.
$$

Equivalently, if $\gamma_s^\pm(r)$ is the projection of
$\rho_s^\pm(r;y,\eta)$ to $M$, then
$$
  \partial_s H_s(x,\xi)
  =
  \frac12 \partial_s g_s(\xi_s^\sharp,\xi_s^\sharp),
  \qquad
  \xi_s^\sharp=(g_s)^{-1}\xi,
$$
and hence
$$
  \tau_s^\pm(y,\eta,t_1,t_2)
  =
  -\frac12
  \int_{t_2}^{t_1}
  \partial_s g_s\bigl(\dot\gamma_s^\pm(r),\dot\gamma_s^\pm(r)\bigr)\,\der r.
$$
Since we do not assume null-geodesic completeness the maps
$$
  (s,y,\eta,t_1,t_2)
  \longmapsto
  \bigl(
    s,\tau_s^\pm(y,\eta,t_1,t_2),
    \rho_s^\pm(t_1;y,\eta),
    \rho_s^\pm(t_2;y,\eta)
  \bigr).
$$
may be defined only on a subset of 
$$
  I\times\dot T^*\Sigma\times\mathbb R\times\mathbb R.
$$
We choose the maximal domain of definition for the flow.
We define $\Lambda^\pm$ to be the image of this map. Finally set
$$
  \Lambda=\Lambda^+\cup\Lambda^-.
$$

This is precisely the Hamiltonian flow-out of the identity relation, with the
additional $T^*I$-component determined by the equation
$$
  \dot\tau=-\partial_s H_s.
$$
Since it is obtained by Hamiltonian flow-out from the identity relation,
$\Lambda$ is a homogeneous Lagrangian submanifold. The smoothness of the above
parametrisation follows from the smooth dependence of solutions of ordinary
differential equations with respect to parameters.

The twisted relation relevant for distribution kernels is
$$
  \Lambda'
  =
  \{(s,\tau,x,\xi,x',-\xi'):(s,\tau,x,\xi,x',\xi')\in\Lambda\}.
$$
Thus Lagrangian distributions in $I^*(I\times M\times M,\Lambda')$ have wavefront set contained in $\Lambda'$.

Moreover,
$$
 \{(s,\tau,x,\xi): \exists x' \text{ with }
 (s,\tau,x,\xi,x',0)\in\Lambda'\setminus 0\}
$$
is empty, because the Hamiltonian flow preserves non-zero covectors and
$\tau_s^\pm$ is homogeneous of degree one in $\eta$. Therefore, by
Hörmander's criterion \cite{Ho1}*{Theorem 8.2.12}, every
$$
  A\in I^*(I\times M\times M,\Lambda')
$$
defines a continuous map
$$
  C^\infty_0(M)\to C^\infty(I\times M).
$$

In our case we already know that such a parametrix exists for every fixed
$s$. We choose a neighbourhood $U$ of a Cauchy surface lying to the past of
$K$. On $U$ the metric is independent of $s$, and hence the restriction of the
kernel of $G_s=G$ to $U\times U$ may be regarded as a Lagrangian distribution
in
$$
  I^*(I\times U\times U,\Lambda_U'),
$$
with polyhomogeneous symbol. Here $\Lambda_U$ denotes the restriction of
$\Lambda$ to $T^*I\times T^*U\times T^*U$.

We now extend this symbol along $\Lambda$ by solving the transport equations.
Consider the principal symbol of $P_s$ as a function on $T^*(I\times M)$,
$$
  p(s,\tau,x,\xi)=g_s^{-1}(\xi,\xi).
$$
Although $p$ is independent of $\tau$, its Hamilton vector field is
$$
X_p =X_{p_s} - \sum_\alpha (\frac{\partial}{\partial s^\alpha}p_s)\partial_{\tau_\alpha}.
$$
By construction of $\Lambda$, this vector field is tangent to the left
projection of $\Lambda$, and $p$ vanishes on this projection. Hence the formula
for the principal symbol of the product of a differential operator with a
Lagrangian distribution whose principal symbol vanishes on the Lagrangian
\cite{MR0388464}*{Theorem 5.3.1} gives the transport equations on $\Lambda$.

In the parametrisation of $\Lambda^\pm$ by the domain in
$I\times \dot T^*\Sigma\times\mathbb R\times\mathbb R$, these transport
equations are ordinary differential equations in the flow parameters, with
coefficients depending smoothly on $s,y,\eta,t_1,t_2$. We use as initial
condition the symbol of the parametrix in the region $U\times U$, where the
metrics are independent of $s$. Smooth dependence of solutions of ordinary
differential equations then gives a smooth solution on $\Lambda$. Since the
flow and the transport equations are homogeneous in the covector variables,
the solution is homogeneous of the required degree. Solving the successive
transport equations recursively, as in \cite{MR0388464}*{Theorem 5.3.2}, and
applying Borel summation gives a classical polyhomogeneous symbol on
$\Lambda$.

Using this symbol we construct a Lagrangian distribution $\tilde G$ on
$I\times M\times M$, whose fibre at $s$ will be denoted by $\tilde G_s$, with
$$
  \tilde G\in
  I^{-\frac{3}{2}-\frac{\dim(I)}{4}}(I\times M\times M,\Lambda').
$$
We may choose $\tilde G$ to be supported near the projection of $\Lambda$ to
$M\times M$. By construction,
$$
  P_s\tilde G_s\in C^\infty(I\times M\times M).
$$
Moreover, by uniqueness of the solutions of the transport equations, for each
fixed $s\in I$ the kernel $\tilde G_s-G_s$ is smooth on $M\times M$.

Let $\Sigma$ be a hypersurface that is a Cauchy surface for all metrics $g_s$,
$s\in I$. Since the solution operator for the Cauchy problem is expressed in
terms of $G_s$, cf. Equation \eqref{solop}, the kernel $\tilde G_s$ gives an
operator
$$
  E_\Sigma:
  C^\infty_0(\Sigma)\oplus C^\infty_0(\Sigma)
  \longrightarrow C^\infty(I\times M),
$$
which is an approximate solution of the Cauchy problem in the sense that
$$
  r_\Sigma E_\Sigma-\mathbf 1
  \in
  C^\infty(I\times\Sigma)\oplus C^\infty(I\times\Sigma),
  \qquad
  P_sE_\Sigma\in C^\infty(I\times M).
$$
By construction,
$$
  E_\Sigma\in
  I^{-\frac{3}{2}-\frac{\dim(I)}{4}+\frac14}
  (I\times M\times\Sigma,\Lambda_\Sigma').
$$
Here $\Lambda_\Sigma$ is obtained by pulling back $\Lambda$ to
$T^*I\times T^*M\times T^*\Sigma$ via the inclusion
$I\times M\times\Sigma\hookrightarrow I\times M\times M$.

Using Duhamel's principle, solving the inhomogeneous Cauchy problem is
equivalent to solving the homogeneous Cauchy problem. The following is an invariant way to invoke this principle.
We use the function $t$
as a global time function for all metrics $g_s$, so that the level sets are
Cauchy hypersurfaces. The restrictions
$$
  \tilde G_s|_{\{s\}\times\Sigma_t\times\Sigma_{t'}}
$$
are well-defined and depend smoothly on $s,t,t'$ as distributions on
$\Sigma\times\Sigma$. The distribution $G_s$ restricted to
$\Sigma_t\times\Sigma_t$ vanishes, whereas its normal derivative gives the
identity on $\Sigma_t$. By subtracting a kernel in
$C^\infty(I\times M\times M)$ with prescribed values and normal derivatives on
$I\times\Sigma_t \times\Sigma_t$, we may arrange that $\tilde G_s$ satisfies these
two identities as well.

Let
$$
  D_s=\tilde G_s-G_s .
$$
We now modify $\tilde G_s$ by a smooth kernel so that $D_s$ vanishes to
infinite order at equal times. Let
$$
  \mathcal{S}
  =
  \{(s,t,x,t',x')\in I\times M\times M : t=t'\}.
$$
This is a hypersurface in $I\times M\times M$.

By the preceding construction, the equal-time Cauchy data of $\tilde G_s$ and
$G_s$ agree modulo smooth kernels. More explicitly, the restrictions of $D_s$
and of the normal derivatives of $D_s$ in each spacetime variable to
$\Sigma_t\times\Sigma_t$ are smooth. Since the normal derivative to
$\mathcal{S}$ is a smooth linear combination of the normal derivatives in the
two spacetime variables, with coefficients depending smoothly on the
coefficients of $g_s$, it follows that the restriction of $D_s$ and its first
normal derivative to $\mathcal{S}$ are smooth.

Moreover
$$
  P_sD_s=P_s\tilde G_s
$$
has smooth kernel in $I\times M\times M$, since $P_sG_s=0$. Since $P_s$ is a
second-order normally hyperbolic operator and each $\Sigma_t$ is a Cauchy
surface for $g_s$, the hypersurfaces $\Sigma_t$ are non-characteristic for
$P_s$. Hence the equation $P_sD_s=P_s\tilde G_s$ determines the second normal
derivative of $D_s$ along $\mathcal{S}$ in terms of lower normal derivatives,
tangential derivatives, and the smooth right-hand side. Differentiating the
equation repeatedly in the normal direction gives, by induction, smooth
expressions for all higher normal derivatives of $D_s$ along $\mathcal{S}$.

By Borel's theorem for hypersurfaces (\cite{Ho1}*{Theorem 1.2.6}, see also \cite{MR1852334}*{Section 3.3}) there exists
a smooth kernel $B_s\in C^\infty(I\times M\times M)$ whose full normal jet
along $\mathcal{S}$ agrees with that of $D_s$. Replacing $\tilde G_s$ by
$\tilde G_s-B_s$, we may therefore assume, without loss of generality, that
$\tilde G_s-G_s$ vanishes to infinite order along $\mathcal{S}$. In particular,
multiplication by $\theta(t-t')$ introduces no additional singularities from
this difference.

Now define
$$\tilde G_{s,\mathrm{ret}}(t,x,t',x')
  = \theta(t-t')\tilde G_s(t,x,t',x'),
$$
which is well-defined by the usual wavefront set calculus. Here $\theta$ is
the Heaviside function. We compute
$$
  P_s\tilde G_{s,\mathrm{ret}}
  =
  \mathbf 1+\theta(t-t')P_s\tilde G_s.
$$
By our assumptions on the $t$-derivatives of $\tilde G_s$, the remainder term
$$
  \theta(t-t')P_s\tilde G_s
$$
is a smooth kernel in $C^\infty(I\times M\times M)$. Hence
$$
  \WF(\tilde G_{s,\mathrm{ret}})
  \subset
  \Lambda'\cup\Delta^*,
$$
where
$$
  \Delta^*
  =
  \{(s,0,x,\xi,x,-\xi): (x,\xi)\in T^*M\setminus 0\}
$$
is the wavefront set of the kernel of the identity map, considered as a smooth
family in $s$.

So far we have shown, for each fixed $s\in I$, that
$G_{s,\mathrm{ret}}-\tilde G_{s,\mathrm{ret}}$ is smooth on $M\times M$. It
remains to show that this difference is in fact smooth as a function of
$s$, i.e. that it belongs to $C^\infty(I\times M\times M)$.

We have
\begin{align}
 P_s \tilde G_{s,\mathrm{ret}} = \mathbf{1} + R_s,
\end{align}
where $R_s$ is a smooth kernel in $C^\infty(I \times M \times M)$. The kernels $\tilde G_{s,\mathrm{ret}}(t,x,t',x')$ and $R_s(t,x,t',x')$ vanish when $t< t'$. 
 We denote by $C^k_+(M)$ the space of functions $f \in C^k(M)$ such that
 there exists $T \in \R$ with $f(t,x)=0$ whenever $t< T$. We similarly define $C^k_+(I \times M)$ allowing additional $s$-dependence.
 By the support properties of the kernel of $\tilde G_{s,\mathrm{ret}}$
the Fourier integral operator $\tilde G_{s,\mathrm{ret}}$ defines a continuous map from $C^\infty_+(M) \to C^\infty_+(I \times M)$.
Both $\tilde G_{s,\mathrm{ret}}$ and $R_s$ can also be understood as continuous maps $C^\infty_+(I \times M) \to C^\infty_+(I \times M)$. For $R_s$ we in fact have
\begin{align}
 (R_s f)(s,x) = \int_M R_s(x,x') f(s,x') \der \mathrm{Vol}_{g_0}(x')
\end{align}
which is a continuous map $C_+(I \times M) \to C^\infty_+(I \times M)$.
We will now show the following Lemma.

\begin{lemma} \label{sdflknslf}
 Given any Cauchy surface $\Sigma$ and Cauchy data $(f_s,g_s) \in C^\infty(I \times \Sigma) \oplus C^\infty(I \times \Sigma)$ the unique solution $u_s$ with $P_s u_s =0$ and $R_{s,\Sigma}(u_s) = (f_s,g_s)$ depends smoothly on $s$.
\end{lemma}
\begin{proof}
 By finite speed of propagation and since smoothness is a local property it is sufficient to show this is the case of compact Cauchy surfaces $\Sigma$ and we will assume this now for the duration of the proof, otherwise keeping the notations as above.
 Given an open subset $\calU \subset M$ we define $C_{+,b}(I \times \calU)$ as the space of functions in $C_{+}(I \times \calU)$ that are bounded and past-compact and equip this space with the Banach space norm $\| f \|_{C_{+,b}} = \sup \{ | f(s,t,x) | \mid s \in I, (t,x) \in \calU \}$.
We can now choose a sufficiently small neighborhood $\calU$ of $\Sigma$  of the form $(-T,T) \times \Sigma$ so that the Neumann series for $(\mathbf{1} + R_s)^{-1}$ converges as a map
$C_{+,b}(I \times \calU) \to C_{+,b}(I \times \calU)$. To see that this is possible we note that for each point $x \in M$ there exists a compact set $M_x$ so that
\begin{align}
 (R_s f)(s,x) = \int_\calU R_s(x,x') f(s,x') \der \mathrm{Vol}_{g_0}(x') = \int_{M_x \cap \calU } R_s(x,x') f(s,x') \der \mathrm{Vol}_{g_0}(x')
\end{align}
and we choose $\calU$ so that for all $x \in \calU$ the set $M_x \cap \calU$ satisfies
$$
\mathrm{vol}(M_x \cap \calU)
\sup_{x' \in M_x\cap\calU,\, s\in I}| R_s(x,x') |
\leq 1-\delta .
$$
 for some $\delta>0$.
Since $R_s$ has smooth kernel, applying the above estimate to each derivative
of the kernel shows that the Neumann series converges in
$C^\infty(I\times\calU\times\calU)$.
Thus, we can apply the Neumann series to construct the true retarded fundamental solution on $[-T,T] \times \Sigma$ as an operator whose kernel has its wavefront set in $\Lambda' \cup \Delta^*$. 
This now shows that the restriction of $G_{s,\mathrm{ret}}$ to $I \times \calU \times \calU$ has integral kernel with its wavefront set contained in the set $\Lambda' \cup \Delta^*$. A similar argument also applies to $G_{s,\mathrm{adv}}$. This also shows that the restriction of $G_s$ to  $I \times \calU \times \calU$ is a Fourier integral operator in
$I^{-\frac{3}{2}-\frac{\dim(I)}{4}}(I \times \calU \times \calU, \Lambda')$. Since this is the solution operator for the Cauchy problem it can now be used directly to conclude the statement of the Lemma for $[-T,T] \times \calU$ for some sufficiently small $T>0$. A simple compactness argument establishes the Lemma for all $[-T,T] \times \Sigma$ and arbitrary $T>0$. Hence it holds for all of $M$.
\end{proof}

We choose the Cauchy surface $\Sigma$ in the region where the metric
perturbation vanishes. Hence $g_s=g$, and therefore $P_s=P$, in a
neighbourhood of $\Sigma$ for all $s\in I$. In this neighbourhood the
fundamental solutions $G_s$ agree with the fixed fundamental solution $G$.
Moreover, the construction of $\tilde G_s$ was initialized from the same
kernel in this region. Thus the restriction of $\tilde G_s-G_s$, together
with its Cauchy data along $\Sigma\times\Sigma$, is smooth as a function of
$s$.
Applying the preceding smooth-dependence result for the Cauchy problem, we
conclude that $\tilde G_s-G_s$ has a distributional integral kernel in
$C^\infty(I\times M\times M)$. It now follows that
$G_s$ is a Fourier integral operator in $$I^{-\frac{3}{2} - \frac{\mathrm{dim}(I)}{4}}(I \times M \times M,\Lambda').$$ This argument can also be applied to $G_{s,\mathrm{ret}}$ and to $G_{s,\mathrm{adv}}$ and we have therefore proved the following statement.

\begin{theorem} \label{fiorepra}
 The parameter dependent fundamental solutions $G_{s,\mathrm{ret}}$ and $G_{s,\mathrm{adv}}$ 
 have integral kernels in $\mathcal{D}'(I \times M \times M)$ with wavefront sets contained in $\Lambda' \cup \Delta^*$.
 The parameter dependent Pauli-Jordan distribution $G_s$ is a Fourier integral operator in $\mathrm{I}^{-\frac{3}{2} - \frac{\mathrm{dim}(I)}{4}}(I \times M \times M, \Lambda')$ with polyhomogeneous symbol.
\end{theorem}

Note that the statements are then local on $I$ and therefore Lemma \ref{sdflknslf} and Theorem \ref{fiorepra} hold for all smooth maps $I \to \mathcal{M}$.

\end{document}